\documentclass{article}

\usepackage{authblk}
\usepackage{comment}
\usepackage{amsfonts}
\usepackage{rlepsf}
\usepackage{graphics}
\usepackage{xcolor}
\usepackage{amssymb}  
\usepackage{verbatim}
\usepackage{amsmath}
\usepackage{esint}  
\usepackage{amssymb}
\usepackage{amssymb,graphicx}
\usepackage{amsthm}
\usepackage{amscd}
\usepackage{euscript}
\usepackage{multicol}
\usepackage{mathtools}
\usepackage[all]{xy}
\usepackage{hyperref}

\newcommand{\VB}{{\mathcal {VB}}}    
\newcommand{\WB}{{\mathcal {WB}}} 
\newcommand{\B}{{\mathcal {B}}}    
\newcommand{\VVV}[2]{V_{#1}[#2]}   
\newcommand{\SSS}[2]{S_{#1}^+[#2]}   
\newcommand{\SSM}[2]{S_{#1}^-[#2]}   
\newcommand{\SSPM}[2]{S_{#1}^{\pm}[#2]}   

\def \Cc {{\cal C}}
\def \Obj {{\rm Obj}}
\def \Mor {{\rm Mor}}
\def \Rc{{{\cal R}}}

\def \tp {{\triangleright'}}
\def \l {\lessdot}
\def \t {\gtrdot}
\def \td {{\gtrdot}}

\def \bw {\overline {w}}

\def \bt {\overline{t}}

\def \bJ{\overline {J}}
\def \bK {\overline{K}}

\usepackage{rlepsf}
\usepackage{float}
\usepackage[all]{xy}
\newcommand{\gwrr}[1]{\Gamma^{\wr #1}}  

\newcommand{\gwr}[1]{\Gamma^{#1} \rtimes \Sigma_{#1}}  

\usepackage{amsfonts}
\usepackage{amsmath}

\makeatletter
\newcommand*{\da@rightarrow}{\mathchar"0\hexnumber@\symAMSa 4B }
\newcommand*{\da@leftarrow}{\mathchar"0\hexnumber@\symAMSa 4C }
\newcommand*{\xdashrightarrow}[2][]{%
  \mathrel{%
    \mathpalette{\da@xarrow{#1}{#2}{}\da@rightarrow{\,}{}}{}%
  }%
}
\newcommand{\xdashleftarrow}[2][]{%
  \mathrel{%
    \mathpalette{\da@xarrow{#1}{#2}\da@leftarrow{}{}{\,}}{}%
  }%
}
\newcommand*{\da@xarrow}[7]{%
  \sbox0{$\ifx#7\scriptstyle\scriptscriptstyle\else\scriptstyle\fi#5#1#6\m@th$}%
  \sbox2{$\ifx#7\scriptstyle\scriptscriptstyle\else\scriptstyle\fi#5#2#6\m@th$}%
  \sbox4{$#7\dabar@\m@th$}%
  \dimen@=\wd0 %
  \ifdim\wd2 >\dimen@
    \dimen@=\wd2 %
  \fi
  \count@=2 %
  \def\da@bars{\dabar@\dabar@}%
  \@whiledim\count@\wd4<\dimen@\do{%
    \advance\count@\@ne
    \expandafter\def\expandafter\da@bars\expandafter{%
      \da@bars
      \dabar@ 
    }%
  }%
  \mathrel{#3}%
  \mathrel{%
    \mathop{\da@bars}\limits
    \ifx\\#1\\%
    \else
      _{\copy0}%
    \fi
    \ifx\\#2\\%
    \else
      ^{\copy2}%
    \fi
  }%
  \mathrel{#4}%
}
\makeatother

\def \U {\underline{U}}
\def \V {\underline{V}}
\def \W {\underline{W}}

\def \G {\Gamma}
\def \uR {\underline{R}}
\newcommand{\st}[1]{(\xrightarrow{{#1}})}

\usepackage[utf8]{inputenc}
\usepackage{dashbox}

\newcommand{\Gam}{(\Gamma_1,\Gamma_0,\sigma, \tau, \iota, \star)} 
\newcommand{\Gamp}{(\Gamma_1',\Gamma_0',\sigma', \tau', \iota', \star')} 

\newcommand{\trll}{\trl^*}  

\newcommand{\N}{{\mathbb N}}

\usepackage{rlepsf}
\usepackage{float}

\def \g {\gamma}

\def \Homeo {{\rm Homeo}}

\def \trr {{\triangleright}}
\def \trl {{\triangleleft}}
\def \tn {{ \otimes}}
\def \ra {\xrightarrow}
\def \d {{\partial}}
\def \Gc {{\cal G}}
\def \id {{\rm id}}
\def \AUT {{\rm AUT}}
\def \LBG {{\rm LB}}
\def \WBG {{\rm WB}}
\def \VBG {{\rm VB}}
\def \MCG {{\rm MCG}}
\def \BG {{\rm B}}
\def \TRANS {{\rm TRANS}}
\def \F {{ \cal F}}
\def \C {{ \mathbb{C}}}
\def \R {{ \mathbb{R}}}
\def \Z {{ \mathbb{Z}}}

\newcommand{\Spic}{
{\underbrace{\hspace{-3mm}
\xymatrix@R=22pt@C=0pt
{
  &\,1\ar@{{}{ }{}}@/_1.3pc/[rrrr]^{\,\,\,\,\, ...}  \mbox{} \ar[d]
  & \mbox{ } \ar[d]<1.2ex> \ar[d]<-1.2ex> \ar[d]<0ex>
  &   \mbox{ a}  \ar[drrr]|\hole
 \ar@{{}{ }{}}@/_1.3pc/[rrrrrr]^{\,\,\,\,\,\,...} &&& \mbox{\hspace{-5mm} a+1} \ar[dlll]
 & & {}\ar[d]<1.2ex> \ar[d]<-1.2ex> \ar[d]<0ex> & \ar[d] \mbox{n} \\ 
& {} & {}  & \hspace{.13in}   && {}& & &{} & \hspace{.1in} 
}
\hspace{-2mm}
}_{n \textrm{ strands}}
}
}
\newcommand{\Smpic}{
{\underbrace{\hspace{-3mm}
\xymatrix@R=22pt@C=0pt
{
  &\,1\ar@{{}{ }{}}@/_1.3pc/[rrr]^{\,\,\,\,\,\, ...}  \mbox{} \ar[d]
  & \mbox{ } \ar[d]<1.2ex> \ar[d]<-1.2ex> \ar[d]<0ex>
  &   \mbox{ a}  \ar[drrr]
 \ar@{{}{ }{}}@/_1.3pc/[rrrrrr]^{\,\,\,\,\,\,\,\,...} &&& \mbox{\hspace{-5mm} a+1} \ar[dlll]|>>>>>>>\hole
 & & {}\ar[d]<1.2ex> \ar[d]<-1.2ex> \ar[d]<0ex> & \ar[d] \mbox{n} \\ 
& {} & {}  & \hspace{.13in}   &&& {}& & &{} & \hspace{-6mm} 
}
\hspace{-4mm}
}_{n \textrm{ strands}}
}
}

\newcommand{\Vpic}{
{\underbrace{\hspace{-3mm}
\xymatrix@R=22pt@C=0pt
{
  &\,1\ar@{{}{ }{}}@/_1.3pc/[rrrr]^{\,\,\,\,\, ...}  \mbox{} \ar[d]
  & \mbox{ } \ar[d]<1.2ex> \ar[d]<-1.2ex> \ar[d]<0ex>
  &   \mbox{ a}  \ar[drrr]
 \ar@{{}{ }{}}@/_1.3pc/[rrrrrr]^{\,\,\,\,\,\,\,\,...} &&& \mbox{\hspace{-5mm} a+1} \ar[dlll]
 & & {}\ar[d]<1.2ex> \ar[d]<-1.2ex> \ar[d]<0ex> & \ar[d] \mbox{n} \\ 
& {}  {}  & \hspace{.13in}   &&& {}& & &{} & \hspace{.1in} 
}
\hspace{-2mm}
}_{n \textrm{ strands}}
}
}

\newcommand{\In}{
{\underbrace{\hspace{-3mm}
\xymatrix@R=15pt@C=0pt
{
 &\,1\ar@{{}{ }{}}@/_1.3pc/[rrrrr]^{\,\,\,\,\,\,\, ...}  \mbox{} \ar[d]
  & \mbox{ } \ar[d]<1.2ex> \ar[d]<-1.2ex> \ar[d]<0ex>
  & 
 && & {}\ar[d]<1.2ex> \ar[d]<-1.2ex> \ar[d]<0ex> & \ar[d] \mbox{n} \\ 
 &\hspace{.13in}  & &&&  &  & \hspace{.1in} 
}
\hspace{-2mm}
}_{n \textrm{ strands}}
}
}

\newtheorem{Theorem}{Theorem}

\newtheorem{Lemma}[Theorem]{Lemma}
\newtheorem{Proposition}[Theorem]{Proposition}
\newtheorem{Fundamental Theorem}{Fundamental Theorem}

\newcounter{minidef}[section]
\renewcommand{\theminidef}{\thesection.\arabic{minidef}}

\newcommand{\mdef}{\refstepcounter{minidef} 
\noindent {\bf [\theminidef]} }

\newcommand{\peq}[1]{\hyperref[#1]{{\bf[\ref{#1}]}}}

\theoremstyle{definition}  

\newtheorem{Example}[Theorem]{Example}
\newtheorem{Definition}[Theorem]{Definition}
\newtheorem{Remark}[Theorem]{Remark}

\newcommand{\ignore}[1]{}
\newcommand{\beq}{\begin{equation}}
\newcommand{\eq}{\end{equation}}

\newcommand{\AGr}[2]{{\mathbf \Gamma}_{#2}(#1)}  

\begin{document}


\newcommand{\biker}{bikoid}
\newcommand{\Biker}{Bikoid}

\newcommand{\wbiker}{welded \biker}                
\newcommand{\wwbiker}{W-\biker}

\title{Representations of the Loop Braid Group and 
  Aharonov-Bohm {like} effects in discrete (3+1)-dimensional higher gauge theory}
%


\author {Alex Bullivant\footnote{E-mail address: 
a.l.bullivant@leeds.ac.uk}; \quad Jo\~ao Faria Martins\footnote{E-mail address: j.fariamartins@leeds.ac.uk},\quad
 Paul Martin.\footnote{E-mail address: 
p.p.martin@leeds.ac.uk}\\

	{School of Mathematics, University of
		Leeds, Leeds, LS2 9JT, United Kingdom}
}
\maketitle



{We show that representations  of the loop braid group  arise from Aharonov-Bohm like effects in  finite 2-group  (3+1)-dimensional topological higher gauge theory.
For this we introduce a minimal categorification  of biracks, which we call W-bikoids (welded bikoids).
Our main example of  W-bikoids  arises from finite 2-groups, realised as crossed modules of groups. Given a W-bikoid, and hence a groupoid of symmetries, we construct a family of unitary  representations of the  loop braid group derived from representations of the groupoid algebra. 
We thus give a candidate for higher Bais' flux metamorphosis, and hence also   a version of a `higher quantum group'.}


\section{Introduction}


The motivation for this work is to describe the
observable collective  
properties 
of loop-like quasiparticles \cite {WalkerWang,paul_et_al,LL,CL,EN,KBS}
arising in (3+1)-dimensional topological phases of matter with higher
gauge symmetry \cite{Our1,Our2,kapustin,Wang2017},
thereby modelling
{corresponding higher Aharonov-Bohm  effects}.

{Higher gauge theory is a version of gauge theory that features 2-dimensional holonomies along surfaces \cite{BaezHuerta11,baez_schreiber}, hence in particular  along trajectories of loops moving in 3-dimensional space.}
{In this paper we model Aharonov-Bohm phases that are conjecturaly related to the 2-dimensional holonomies along the surfaces traced by loop-particles as they move.}

In a topological system the observable collective properties of particles are, at
most, their braidings.
%
%
{
The braiding 
of {(unknotted and unlinked)} {oriented circles}  is described by the
{\em $n$-loop braid group} $\LBG_n$  
\cite{
baez_et_al,damiani}.
So we study representations of $\LBG_n$.
Technically this {group} can be realised as 
the mapping class group
$\MCG(D^3,C_n)$
of self-mappings of the 3-disk $D^3$ that
are the identity on $\partial D^3$ and 
fix $C_n$,
an unlinked union of {unknotted oriented circles, setwise, in addition preserving the orientation on the circles}. 
}%
(We review the 
definitions in \S \ref{motion}.)
%
%
%
{The group $\LBG_n$ can also be described 
in terms of loop motions \cite{baez_et_al,GS}; see \cite[\S 2]{damiani}.
}

Let us give a visualization of  
motions
of loops  in {(3+1)-dimensions} 
that would be expected to 
induce `Aharonov-Bohm like' effects \cite{LL,CL,EN,KBS,PachosB}.
It will be helpful (cf. \cite{baez_et_al,damiani})
first to work with two circular and parallel loops, and in a frame in which
one of the loops is fixed.
Thus for example
consider two loops that are initially coaxial
(as on the left in \eqref{eq:fig1} below), and work in a frame in which the
upper loop is fixed. We can then visualize a motion by showing the
`Dirac sheet', the surface swept out by the other loop in this frame.
At times $t=0,1,2$
we might then have:
\newcommand{\fig}[2]{\includegraphics[width=#1cm]{#2}}
\beq \label{eq:fig1}
\fig{2.9}{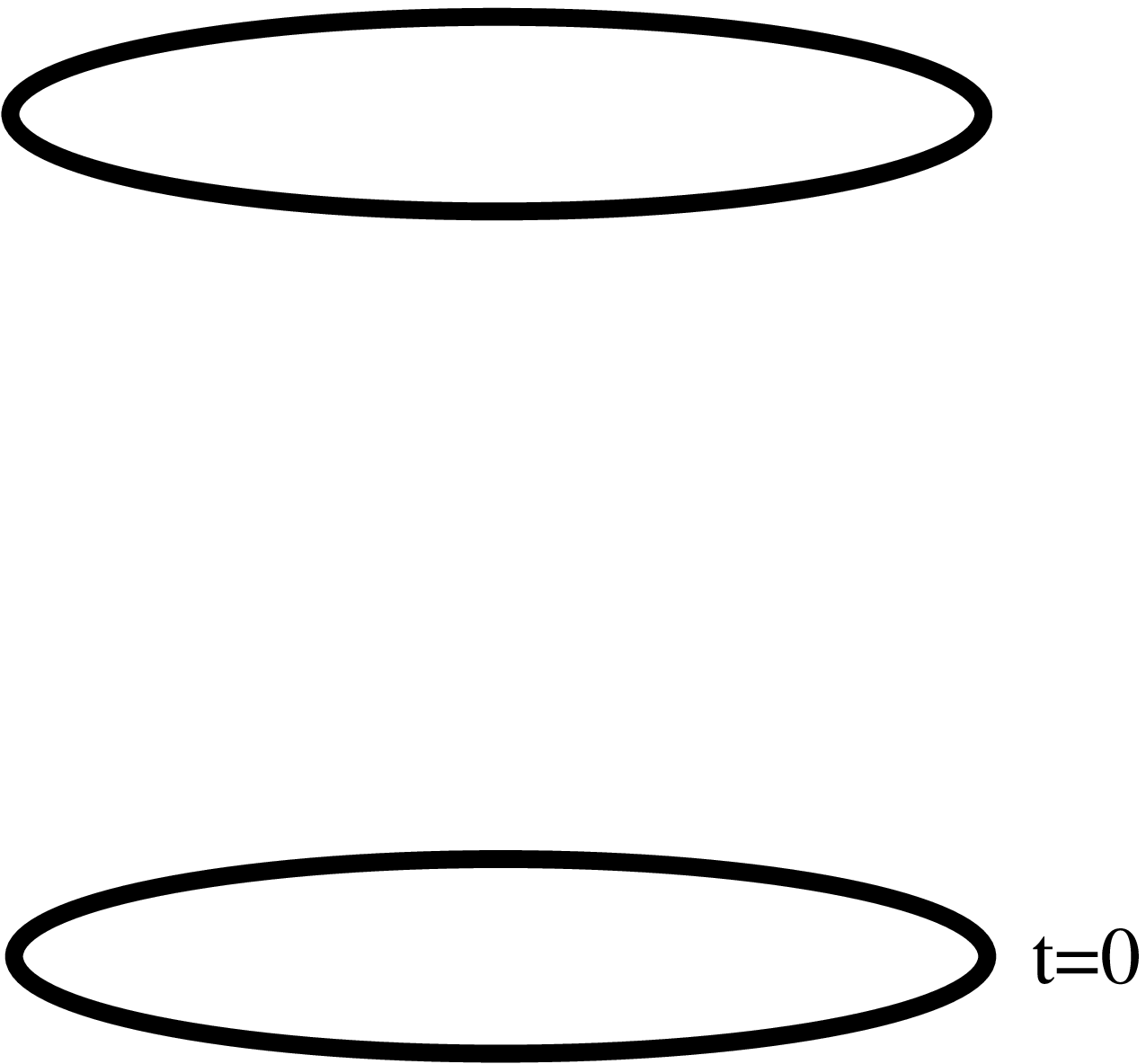} \qquad\hspace{.8631in}
\fig{3.0}{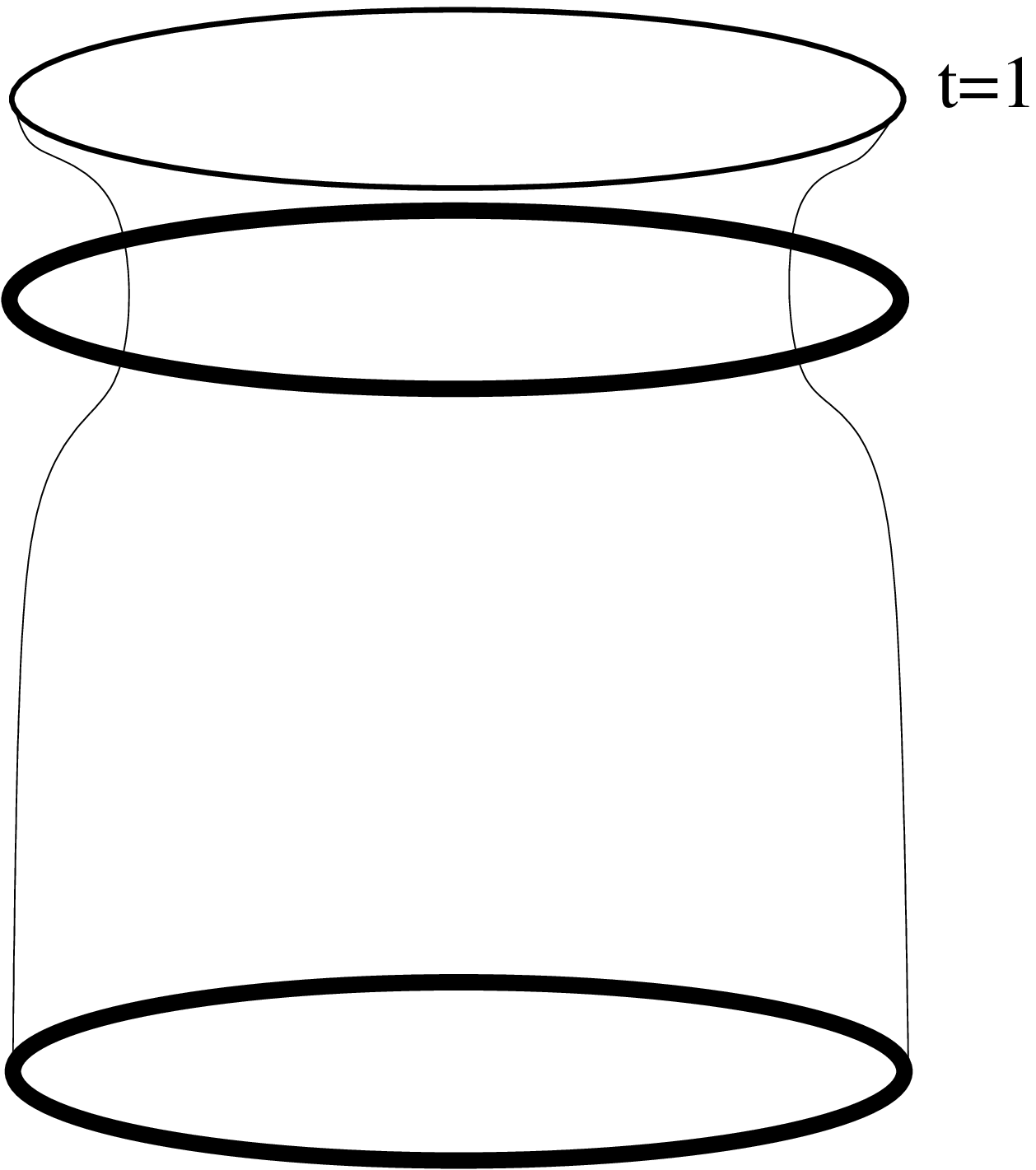} \qquad\hspace{.631in}
\fig{4.2}{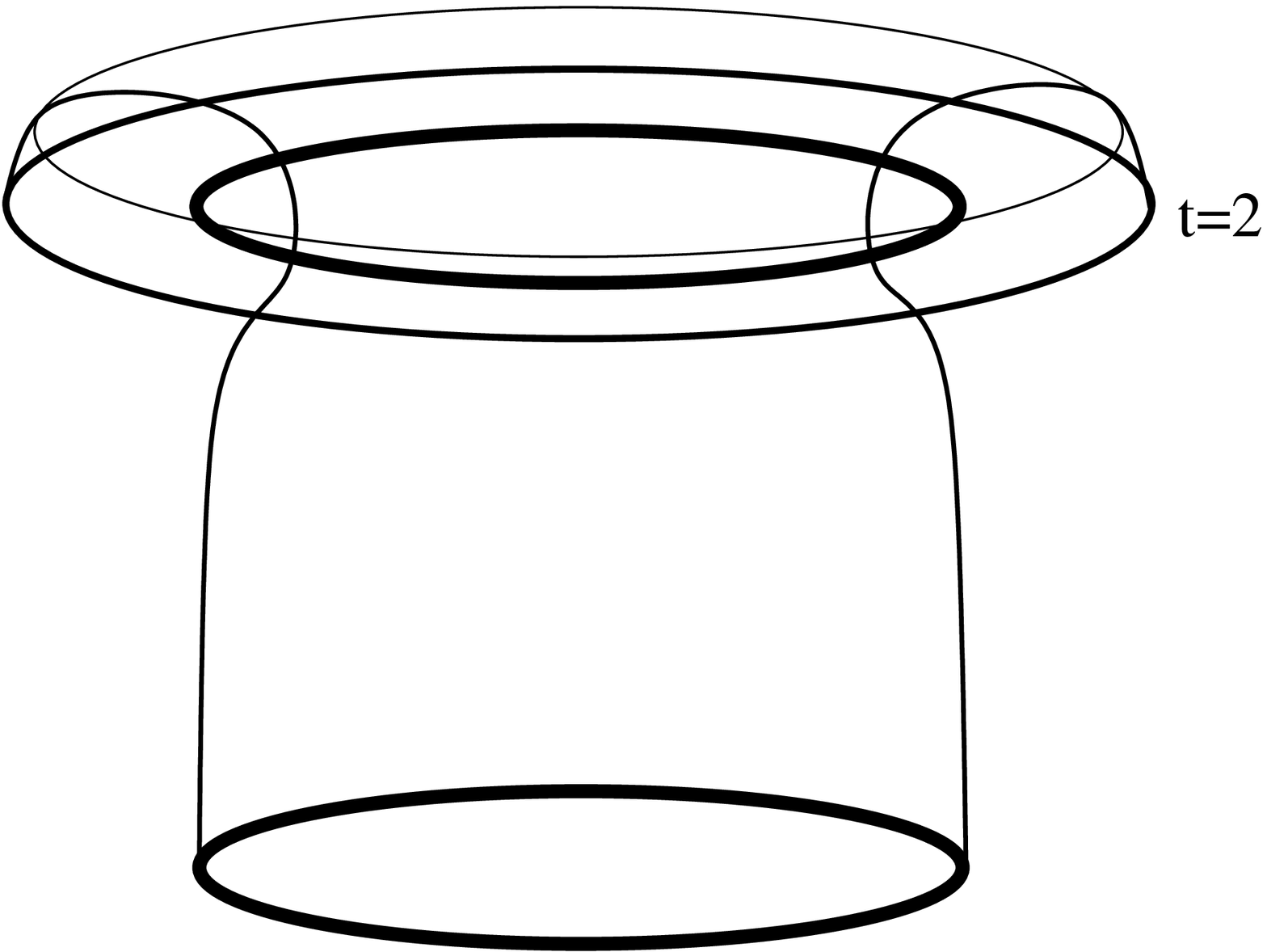}
\eq
{This is  analogous to tracking the world-line of a
point-particle moving around another point in a plane. Indeed our
particular motion is reproduced by rotating this plane into {3 dimensions} about a suitable
axis in the plane (which axis then becomes the coaxis of the two
loops).}
Note that our motion up to $t=1$ can be completed to a pure loop-braid in two
topologically distinct ways, that are both {distinct} from
the `un-motion'.
Firstly the `moving' loop can return to
its original position first as at $t=2$, and from there in the obvious
direct way, {as shown in  figure \eqref{eq:fig1.2} below.}
\beq \label{eq:fig1.2}
\fig{3.5}{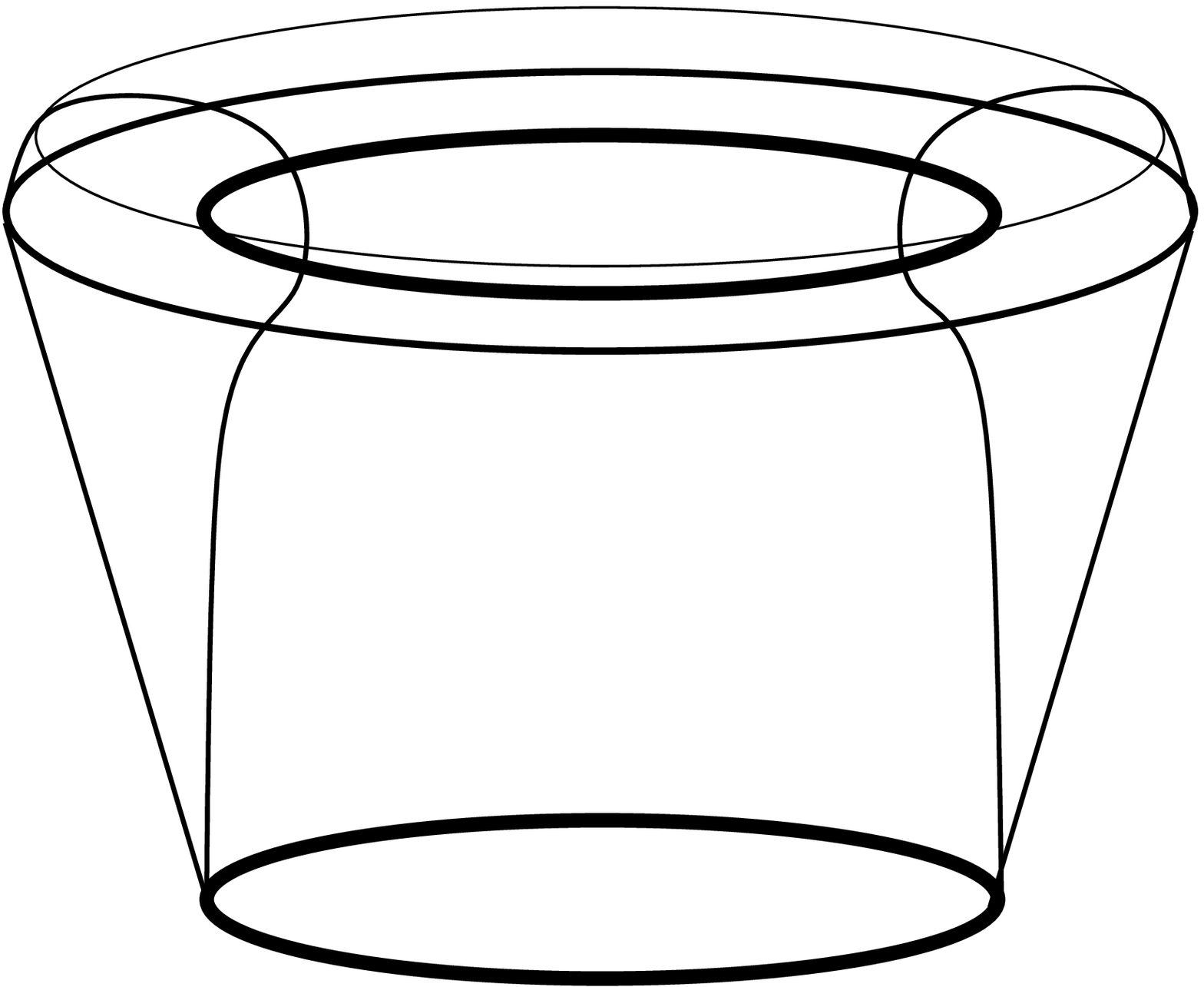}
\eq
Secondly from $t=1$ we can then break the axial symmetry and move to
some position high on the left, say, then return to base in the
obvious way from there; {see  figure \eqref{eq:fig2} below.}

\beq \label{eq:fig2}
\fig{6.5}{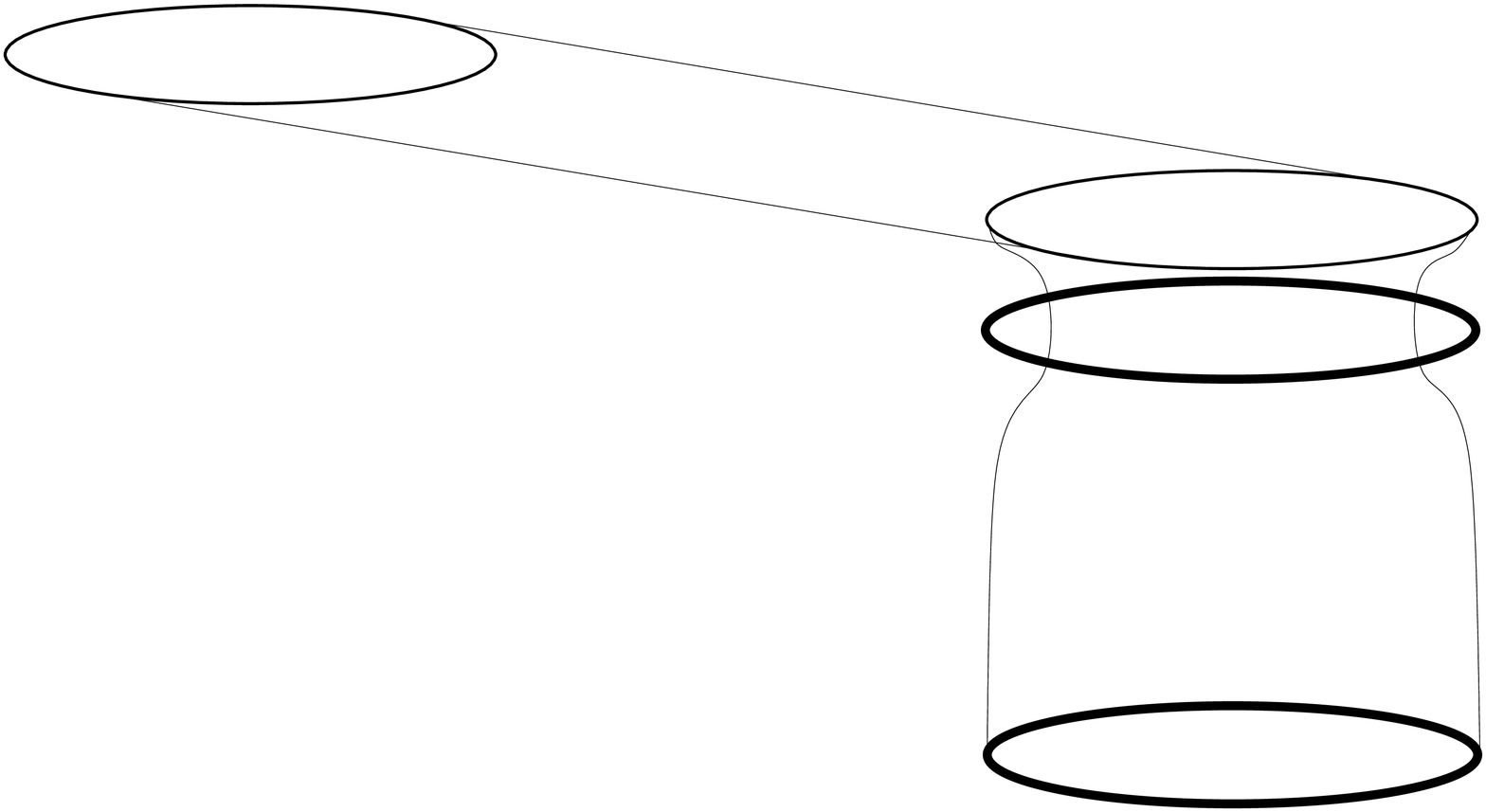} \hspace{.5in}
\fig{6.5}{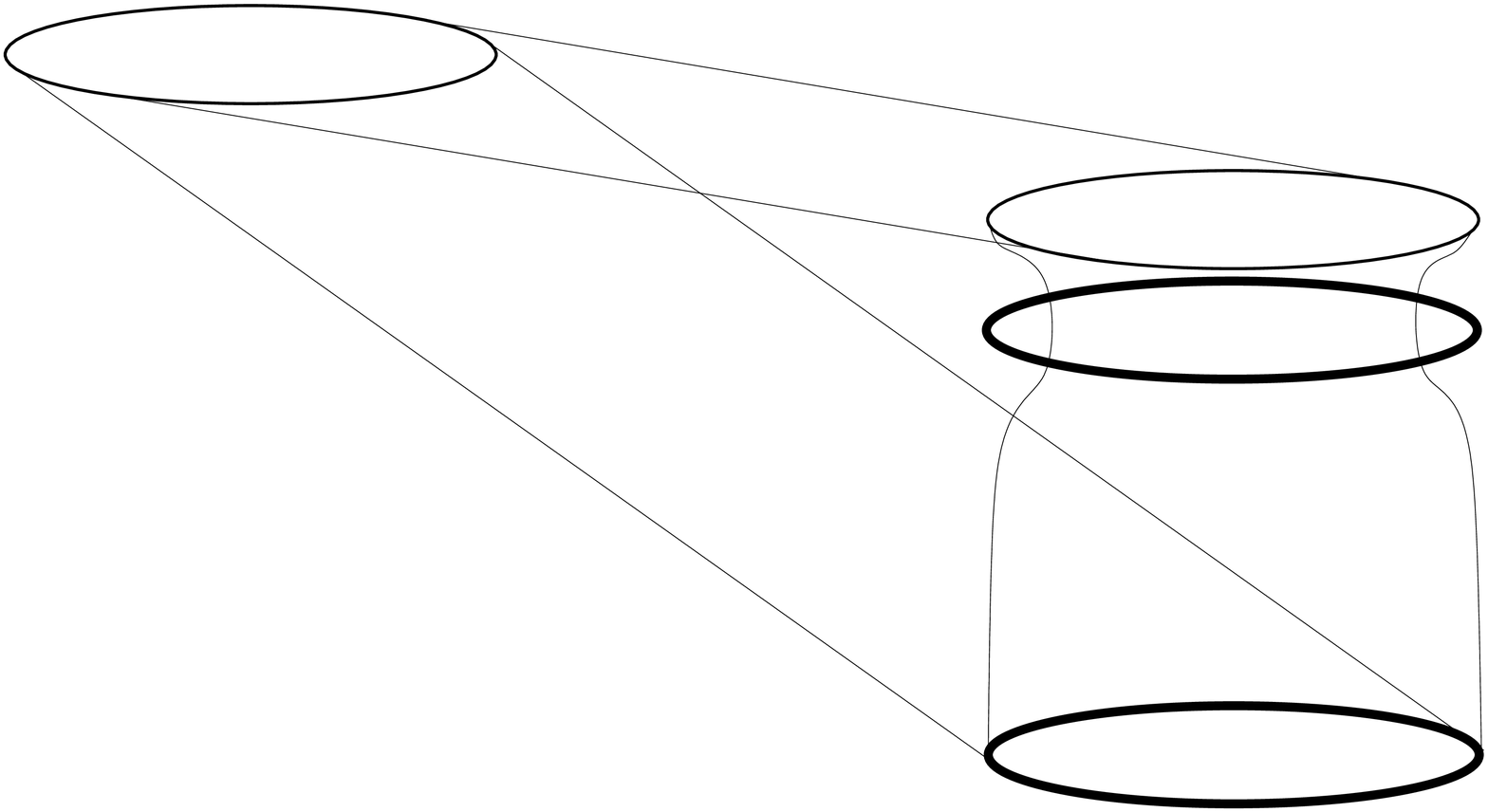}
\eq
{These two topologically distint motions of a loop around
another one play a key role in understanding Aharonov-Bohm like
effects for loop-particles; see \S\ref{phys}.}


{Let $\sigma$ denote the 
motion class 
represented by $t=0\rightarrow 1$ in \eqref{eq:fig1};
and $\nu$ the class  
represented by $t=1\rightarrow 3$ in \eqref{eq:fig2}.
(Keep in mind the moving frame
--- the loops are exchanged by $\sigma$ and $\nu$.)
Extending to a stack of $n$ loops, 
it turns out that these motions generate $\LBG_n$.
Thus given (I) a corresponding $n$-loop-particle Hilbert space
${\mathcal H}$, {in some topological theory},
and (II) transformations on ${\mathcal H}$
corresponding to these motions,
we will have a representation of {the loop braid group} $\LBG_n$.
So far there is no higher analogue of the
solenoid/electron holography setup or 
experiments telling us how to obtain (I,II).
But there are formal recipes coming from higher gauge theory
-- as we shall show in \S\ref{physmot}.}
{Since these
recipes
are necessarily partly {\it ad hoc,} we then rely on combinatorial methods to
verify the well-definedness Theorem {for the representations of $\LBG_n$},
which takes up the remainder of the paper.
}


\smallskip

Writing $\sigma_i$ for the $\sigma$ motion exchanging loops $i$ and $i+1$ it is
easy to verify that the $\sigma$ and $\nu$ generators obey braid
relations:
$\sigma_i \sigma_{i+1}\sigma_i = \sigma_{i+1}\sigma_i\sigma_{i+1}$,
and certain mixed braid relations between them.
These relations give a homomorphism with
a group
defined by generators and relations.
{
For the combinatorial proof
we will use 
the {\em isomorphism} \cite{baez_et_al,damiani} of $\LBG_n$
to the {\em welded braid group} $\WBG_n$ 
\cite{fenn_et_al,kamada,KLam}.
This presentation arises in virtual knot theory.
}%
\ignore{{
This $\WBG_n$ is a quotient of the  {\em virtual braid group}.
$\VBG_n$.
Hence, we also connect with virtual knot theory \cite{kauffman,kamada}.
}}%
%
%
In virtual knot theory we have `virtual' knot and braid diagrams
\cite{damiani,kamada,kauffman,KLam}; {see \S \ref{ss:vbg}.}
These diagrams have positive, negative and virtual crossings between strands, {as below}:
$$
\xymatrix@R=1pt{\\ \\\textrm{positive crossing}}\hskip-1cm
\xymatrix{ &\ar[dr]|\hole &\ar[dl] \\ && }\xymatrix@R=1pt{\\\\\\\,\,,} \qquad \qquad
\xymatrix@R=1pt{\\ \\\textrm{negative crossing}} \hskip-1cm
\xymatrix{ &\ar[dr] &\ar[dl]|\hole \\ && }\xymatrix@R=1pt{\\\\\\\,\,,}\qquad \qquad
\xymatrix@R=1pt{\\ \\\textrm{virtual crossing}}\hskip-1cm
\xymatrix{ &\ar[dr] &\ar[dl] \\ && }\xymatrix@R=1pt{\\\\\\\,\,.}
$$
%
%
\smallskip

\mdef \label{def:VBG}
The $n$-strand {\em virtual braid group} $\VBG_n$ {(discussed in \S\ref{ss:vbg})}
is generated, {as a monoid}, by the elements $\SSS{a}{n}$, {$\SSM{a}{n}$} and $\VVV{a}{n}$, where $a\in
\{1,\dots,n-1\}$, 
{subject to certain relations; see  \S\ref{ss:vbg}.}
{These generators 
 may be depicted as $n$-strand virtual braid diagrams, as below,
with composition via vertical stacking.}
\beq\label{eq:SV}
\hskip-2mm{\xymatrix@R=1pt{\\ \SSS{a}{n} =}}
\quad {\Spic}\xymatrix@R=1pt{\\\\\\\,,} \quad \,
{\xymatrix@R=1pt{\\ \SSM{a}{n} =}}
\quad
\Smpic\xymatrix@R=1pt{\\\\\\\ ,}
\quad\,
        {\xymatrix@R=1pt{\\ \VVV{a}{n} =}}
\Vpic \xymatrix@R=1pt{\\\\\\\,.}\quad
\eq
{The relations in $\VBG_n$ between generators are derived from the usual classical, virtual and mixed Reidemeister moves  between virtual braid diagrams \cite{kamada,KLam,kauffman,BBD};  see  equations (\ref{eq:SRI})--(\ref{eq:VRIII}), and the paragraph below.}
{One passes from  $\VBG_n$  to $\WBG_n$ by furthermore imposing the  the so-called  {\em``Forbidden Move $F_1$''} \cite{KLam,BBD,fenn_barth}, here called the {\em W-move} \cite{kamada}'', a short for {\em welded move}, {displayed in Equation \eqref{eq:F}.}}

{On generators, the isomorphism $\WBG_n \rightarrow \LBG_n$ is given by the motion in  Fig. \ref{ISotopy}.}

\smallskip

\mdef {We have categories $\VB$ and $\WB$ of virtual and welded braid groups. The set of objects of both is $\N$, and $\hom(n,n)$ are $ \VBG_n$ and  $ \WBG_n$, respectively. (There are no morphisms $n\to m$ if $n\neq m$.)}
%
%
The composition $\otimes$ of diagrams {through} horizontal stacking
 gives that $\VB$ and $\WB$ are monoidal categories.
Note that $\SSS{2}{3} = 1_1 \otimes \SSS{1}{2}$, $\SSS{2}{3} =  \SSS{1}{2} \otimes 1_1$ and so on. {Here $1_1$ is the identity morphism $1 \to 1$.}

{The Reidemeister III-move \eqref{eq:R3} in $\VB$ and $\WB$, 
 is a consequence of the monoidal category structure in $\VB$ and $\WB$, and the particular case $\SSS{1}{3}\,\SSS{2}{3}\,\SSS{1}{3}=\SSS{2}{3}\,\SSS{1}{3}\,\SSS{2}{3}$, which can be written as:}
\begin{equation}\label{bisRIII}{(\SSS{1}{2} \otimes 1_1)\, (1_1 \otimes \SSS{1}{2})\, (\SSS{1}{2}\otimes 1_1)= (1_1 \otimes \SSS{1}{2})\, (\SSS{1}{2} \otimes 1_1)\, (1_1 \otimes \SSS{1}{2}).}
\end{equation}
{And the welded-move in $\WB$ is a consequence of the monoidal structure and the particular case:}
\begin{equation}\label{bisW}{(\VVV{1}{2} \otimes 1_1)\, (1_1 \otimes \SSS{1}{2})\, (\SSS{1}{2}\otimes 1_1)= (1_1 \otimes \SSS{1}{2})\, (\SSS{1}{2} \otimes 1_1)\, (1_1 \otimes \VVV{1}{2}).}
\end{equation}

%
 \smallskip

\mdef {One can show
{(a proof is in e.g. \cite[Thm. 2]{Manturov2005})}
that the category $\VB$
of virtual braid groups
contains the braid category
$\B$
(defined in  e.g. in  \cite[\S XIII.2]{Kassel}) as a subcategory} --- {it is the part
monoidally generated in $\VB$  by the $\SSPM{1}{2}$.}
{Suppose we have a monoidal
representation of $\B$ {in a symmetric category; the main example to have in mind is the category of vector spaces.}
This lifts automatically {\cite{Brochier}} to a representation of
$\VB$ by mapping {$\VVV{1}{2}$} to the transposition of tensor factors.
However, for representations of
the loop braid group we require
the welded braid group, and in this 
the {\em welded move} \eqref{bisW} is also satisfied.
The satisfaction of the welded move is one key outcome of our higher
gauge construction.}




\smallskip

{In order to formalise the representations of $\LBG_n$ derived
from finite 2-group higher gauge theory \cite{BaezHuerta11} in 3+1 dimensions, we develop
a (lightly) categorified notion  of biracks  
{\cite{fenn_barth,fenn_biquandle,nelson}}, which we call
``W-\biker s''.
The connection 
to loop-excitations in
(3+1)-dimensional topological higher gauge theory is in \S  \ref{phys}.}

\smallskip
\mdef
A
{\em birack}  $(X,/,\backslash)$ is a set $X$
with two operations
$  X \times X \rightarrow X$ written
$(x,y)  \mapsto y / x$     
and
$(x,y)  \mapsto x \backslash y$    
%
%
such that $(x,y)\mapsto (y/x, x \backslash y)$
defines an invertible map
$S: X \times X \to X \times X$,
satisfying the set-theoretical Yang-Baxter equation \cite{st},
$\big(S \times {\rm \id}) \circ (\id\times S) \circ (S \times \id)
=(\id \times S) \circ (S \times \id)\circ (\id \times S)$;
and for any $a\in X$ the maps   
$x \mapsto x/a$ and
$x \mapsto     x\backslash a$
are invertible. {Our conventions for biracks are spelled out in \S \ref{c:biracks}.}

\smallskip

As we  {now recall}, a birack  encodes the combinatorics of virtual braid
diagrams (and more generally of virtual knot and link diagrams)
\cite{kamada,KLam,kauffman,unsolved,manturov} up to
Reidemeister II and III moves between them.

\smallskip
\mdef
Given a birack $(X,/,\backslash)$, a $(X,/,\backslash)$-colouring
\cite{fenn_biquandle,NelsonBook} of a virtual braid diagram (more generally of a virtual {link} diagram \cite{fenn_et_al,kamada,KLam,kauffman})
is a map from
the set of edges of the diagram to $X$, satisfying the relations below
when four edges meet at a positive, negative or virtual crossing:
\begin{align}\label{birack-colour}
\xymatrix@R=15pt@C=15pt{&x\ar[dr]|\hole &y\ar[dl]\\ & y/x & x \backslash y}
\xymatrix@R=1pt{ \\ \\ \quad \textrm{, }}
\xymatrix@R=15pt@C=15pt{&y/x\ar[dr] &x\backslash y\ar[dl]|\hole\\ & x &  y}
\xymatrix@R=1pt{ \\ \\ \quad \textrm{ and }}
\xymatrix@R=18pt@C=15pt{&x\ar[dr] &y\ar[dl]\\ & y & x}\xymatrix@R=1pt{\\\\\\\,.}
\end{align}
(A crossing breaks an edge into two components, regardless of
it being over, under  or virtual crossing.)

A {\em biquandle} is a birack also imposing
Reidemeister I moves \cite{fenn_biquandle,NelsonBook} between virtual
knot diagrams.
The axioms of  {biquandles}
{ensure} that the number of
$(X,/,\backslash)$-colourings of a virtual link diagram is invariant
under all {classical} Reidemeister moves.
This implies that the number of $(X,/,\backslash)$-colourings is also invariant under virtual {and mixed} Reidemeister moves between virtual {link} diagrams \cite{kauffman}, hence defining a virtual {link}  invariant, called the {\it colour counting invariant} \cite{NelsonBook}.

{A {\em welded birack} \cite{fenn_barth} (called here a {\em W-birack})
makes the number of $(X,/,\backslash)$-colourings of virtual braid diagrams in addition
invariant under the welded move {\eqref{bisW}}. 
 In the biquandle case this leads to the definition of invariants of welded knots
\cite{fenn_et_al,kamada,kauffman,martins_kauffman,Audoux,welded_link_groups}.}


\smallskip

\mdef \label{pa:repX}
A birack  $(X,/,\backslash)$ yields, for each $n$, 
a (in this paper right)
action $\trl$ of
$\VBG_n$ 
on
$X^n$.     
The action on
$\underline{x} = (x_1,x_2,...,x_n) \in  X^n$ 
of $\SSS{a}{n}$ and of
$ \VVV{a}{n}$ is given by the {bottom lines of the diagrams below:}
$$
\xymatrix@R=15pt@C=0pt{ & x_1\ar[d]
  &\dots\ar[d]<1.2ex> \ar[d]<-1.2ex> \ar[d]<0ex>
  &x_{a-1}\ar[d] & x_a \ar[drrr]|\hole &&& x_{a+1}\ar[dlll]
  &x_{a+2}\ar[d]& \dots\ar[d]<1.2ex> \ar[d]<-1.2ex> \ar[d]<0ex> & \ar[d] x_n\\ 
& x_1 &\dots &x_{a-1}& x_{a+1}/x_a &&& x_{a}\backslash x_{a+1} &x_{a+2}& \dots & x_n
} \xymatrix@R=1pt{\\ \\ \quad \textrm{ and } \qquad } \hskip-0.4cm
\xymatrix@R=15pt@C=0pt{ & x_1\ar[d]
  &\dots\ar[d]<1.2ex> \ar[d]<-1.2ex> \ar[d]<0ex>
  &x_{a-1}\ar[d] & x_a \ar[drrr] &&& x_{a+1}\ar[dlll]
  &x_{a+2}\ar[d]& \dots\ar[d]<1.2ex> \ar[d]<-1.2ex> \ar[d]<0ex> & \ar[d] x_n\\ 
& x_1 &\dots &x_{a-1}& x_{a+1} &&& x_{a} &x_{a+2}& \dots & x_n
}\xymatrix@R=1pt{\\\\\\\,.}
$$
This is to say that:
\begin{align*}
  (x_1,\dots,x_{a-1},x_a,x_{a+1},\dots, x_n)
  \; \trl \; \SSS{a}{n}&\;
  =( x_1,\dots, x_{a-1},x_{a+1}/x_a,x_{a}\backslash x_{a+1}, x_{a+2}, \dots , x_n),\\
  (x_1,\dots,x_{a-1},x_a,x_{a+1},\dots, x_n)
  \; \trl\; \VVV{a}{n} &\; =( x_1,\dots, x_{a-1},x_{a+1}, x_{a}, x_{a+2}, \dots , x_n).
\end{align*}
Linearising, we  have a right-representation of $\VBG_n$ on
$(\C X)^{n  \otimes} \cong  \C (X^n)$.

\smallskip

{There are several ways to enrich  biquandle
{colouring counting invariants}  of links and welded links.
See for instance \cite{C1,C2,C3} for quandle cohomology classes, and
\cite{eisermann,Winter} for meridian-longitude refinements.
In this paper we
give the first representation theoretic steps in the development of 
a light categorification-based enrichment.
We will focus on  the corresponding representations of
the loop braid group. 
Invariants of welded and virtual knots 
will be addressed elsewhere \cite{prox}.}


\newcommand{\Gwr}[1]{\Gamma^{\wr #1}}

\smallskip

The underlying notion is that of a
{\em \biker} $(\Gamma,X^+_{\Gamma})$; see \S\ref{s:biker}.
Let us sketch their definition.  First of all,
given a groupoid \cite{Higgins} $\Gamma$, we write $\Gamma_0$
for the set of objects, $\Gamma_1$  for the set of morphisms {(arrows)},
and $\sigma,\tau: \Gamma_1 \rightarrow \Gamma_0$ for the source and
target maps.
 We write $\Gamma^n$ for the $n$-fold product groupoid and
 $\Gwr{n}$ for the wreath product {\S\ref{thegr}}--- the semidirect product
 $\Gamma^n \rtimes \Sigma_n$ with the symmetric group $\Sigma_n$ acting
 by permutation of factors.
 We draw elements of $\Gwr{n}$ as top-to-bottom oriented permutation
 diagrams (like {$\VVV{a}{n}$} in \eqref{eq:SV} above) with edges decorated by
 elements of $\Gamma_1$. Multiplication is then by vertical stacking, {followed by the composition of the groupoid arrows living in the same strand {of a permutation diagram}.  

 \smallskip
A {\biker\ } is  a groupoid
$\Gamma$ 
and a birack $(\Gamma_0, /, \backslash )$, {called the underlying birack of the bikoid},
together  with two maps
$\Gamma_0 \times \Gamma_0 \to \Gamma_1$,
of the form: $$(x,y)\mapsto (x\ra{L(x,y)} x \backslash y) \textrm{ and } (x,y) \mapsto (y \ra{R(x,y)} y / x).$$ The $L,R$ maps  
are
combined as $L \otimes R$ in $\Gwr{2}$, and $X^+_\G = (L \otimes R)T$,
where $T$ is the elementary transposition. A bikoid can be graphically represented as in Equation \eqref{gbirackcolouring} below. {The over-under effect in the last version of $X_\G^+(x,y)$
carries no information. But this redundant information is useful for later tracking reasons, and in order to formulate bikoid colourings of braid diagrams, in a parallel way to birack colourings as in \eqref{birack-colour}.}
\begin{equation}\label{gbirackcolouring}
\xymatrix@R=1pt{\\\\(x,y)\stackrel{X^+_\Gamma}{\longmapsto}}
  \xymatrix{& x  \ar[dr]
    \ar[dr]|<<<<<{L(x,y)\,\,\bullet\,\quad \quad\,\,\,}
    & y  \ar[dl]\ar[dl]|<<<<<{\quad \quad \,\, \,\,\bullet\,\, R(x,y)}\\
& y/ x &x\backslash y
}\xymatrix@R=1pt{\\\\\\\ \quad=}\xymatrix{& x  \ar[dr]|\hole \ar[dr]|<<<<<{L(x,y)\,\,\bullet\quad \,\quad\,\,\,}|\hole & y  \ar[dl]\ar[dl]|<<<<<{\quad \quad \,\,\,\, \bullet\,\, R(x,y)}\\
& y/ x &x\backslash y
}\xymatrix@R=1pt{\\\\\\\,.}
\end{equation}

{The axioms that $X^+_\Gamma$ in Equation \eqref{gbirackcolouring} should obey are in  Def. \ref{de:biker}. A graphical way to state them is in   Equation  \eqref{eq:r3}, which, `component-wise' (i.e. only looking at the $L,R$ maps) is equivalent to  Equation  \eqref{compbikoids}.}
\smallskip

As we will see,  
constructions for \biker s appear in finite (2-)group
topological gauge theory. There, the $L(x,y)$ and $R(x,y)$ arrows
encode Aharonov-Bohm phases \cite{LL,Bais1,Bais2,Bais3,Bais4} arising from flat connection holonomy and flat 2-connection 2-holonomy obtained when point-particles move in 2-dimensional space {and loop-particles} move in 3-dimensional space; see \S\ref{motivation1} and \S\ref{phys}. Hence, we call $L$ and $R$ ``holonomy arrows''.


\smallskip

One can  formulate the notion of a \biker\ colouring
of a virtual braid diagram (and more generally of a  virtual link diagram), which close to a positive crossing should follow the
pattern indicated in \eqref{gbirackcolouring}.
(At negative crossings we have the inverse of $X^+_{\Gamma}$, and at
virtual crossings only identity holonomy arrows are inserted.)
The axioms of \biker s
precisely ensure that, fixing colours of top and bottom strands, and
given equivalent virtual braid diagrams $B$ and $B'$, there is a
one-to-one correspondence between colourings of $B$ and colourings of
$B'$, which moreover preserves the composition of  holonomy arrows
living in each strand.

\smallskip

\mdef{A {\em welded \biker}
(abbrv.  {\em \wwbiker}) is a \biker\ that also obeys {the welded relation} \eqref{bisW}. In concrete terms this means \eqref{eq:wr1} and \eqref{weldedtoR}.
This ensures that the one-to-one correspondence between bikoid colourings of
the paragraph above also holds {if} $B$ and $B'$ are related by moves
between welded braid diagrams. }

%


\smallskip
\mdef\label{upper-rep}
{Given a \biker\ $(\Gamma,X^+_\G)$, then not only is $\Gamma_0$ a birack, but also 
$\Gamma_1$ 
is a birack.}
%
%
%
{Thus, by \peq{pa:repX} in particular we have {a right-action}} of  
$\VBG_n$ on $ \Gamma_1^n$, hence a representation of $\VBG_n$ on $ \C\Gamma_1^n$.}
{Let us write $\underline{\gamma} \in  \Gamma_1^n$ as:}
$$
\underline{\gamma}  =    
\raisebox{.351in}{
  \xymatrix@R=30pt@C=0pt{ & x_1\ar[d]|{\gamma_1}
    &\dots\ar[d]<1.2ex> \ar[d]<-1.2ex> \ar[d]<0ex>
    &x_{a-1}\ar[d]|{\gamma_{a-1}} & x_a \ar[d]|{\gamma_a}
    & x_{a+1}\ar[d]|{\gamma_{a+1}} &x_{a+2}\ar[d]|{\gamma_{a+2}}
    & \dots\ar[d]<1.2ex> \ar[d]<-1.2ex> \ar[d]<0ex> & \ar[d]|{\gamma_{n}} x_n\\ 
    & y_1 &\dots &y_{a-1}& y_a & y_{a+1} &y_{a+2}& \dots
    & y_n } }\xymatrix@R=1pt{\\.}
$$
\label{act12}
Let $\star$ denote composition of arrows in $\Gamma$.
Then this {representation}  is given by:
\newcommand{\trls}{\trl^{\!*}}
\begin{equation*}
  \hskip-0.5cm
  \xymatrix@R=35pt@C=0pt{ & x_1\ar[d]|{\gamma_1}
    &\dots\ar[d]<1.2ex> \ar[d]<-1.2ex> \ar[d]<0ex>
    &x_{a-1}\ar[d]|{\gamma_{a-1}} & x_a \ar[d]|{\gamma_a}
    & x_{a+1}\ar[d]|{\gamma_{a+1}} &x_{a+2}\ar[d]|{\gamma_{a+2}}
    & \dots\ar[d]<1.2ex> \ar[d]<-1.2ex> \ar[d]<0ex> & \ar[d]|{\gamma_{n}} x_n\\ 
    & y_1 &\dots &y_{a-1}& y_a & y_{a+1} &y_{a+2}& \dots
    & y_n } 
\xymatrix@R=1pt{\\ \\\trls \SSS{a}{n} = \quad } \hskip-0.8cm
  \xymatrix@R=35pt@C=5	pt{ & x_1\ar[d]|{\gamma_1}  &
    \dots\ar[d]<1.2ex> \ar[d]<-1.2ex> \ar[d]<0ex>
    &x_{a-1}\ar[d]|{\gamma_{a-1}}
    & x_{a+1} \ar[d]|<<<<{\gamma_{a+1}\star R(y_a,y_{a+1})}
    & x_{a}\ar[d]|>>>>>{\gamma_{a}\star L(y_a,y_{a+1})}
    &x_{a+2}\ar[d]|{\gamma_{a+2}}
    & \dots\ar[d]<1.2ex> \ar[d]<-1.2ex> \ar[d]<0ex> & \ar[d]|{\gamma_{n}} x_n\\ 
& y_1 &\dots &y_{a-1}& y_{a+1}/y_a & y_{a}\backslash y_{a+1} &y_{a+2}& \dots & y_n}\xymatrix@R=1pt{\\\\.}
\end{equation*}
\begin{equation*}\label{act2}
  \xymatrix@R=25pt@C=0pt{ & x_1\ar[d]|{\gamma_1}  &\dots\ar[d]<1.2ex> \ar[d]<-1.2ex> \ar[d]<0ex>  &x_{a-1}\ar[d]|{\gamma_{a-1}} & x_a \ar[d]|{\gamma_a} & x_{a+1}\ar[d]|{\gamma_{a+1}} &x_{a+2}\ar[d]|{\gamma_{a+2}}& \dots\ar[d]<1.2ex> \ar[d]<-1.2ex> \ar[d]<0ex> & \ar[d]|{\gamma_{n}} x_n\\ 
& y_1 &\dots &y_{a-1}& y_a & y_{a+1} &y_{a+2}& \dots & y_n
  } 
\xymatrix@R=1pt{\\ \\\trls \VVV{a}{n} =  } \hskip-0.4cm
  \xymatrix@R=25pt@C=7pt{ & x_1\ar[d]|{\gamma_1}
    &\dots\ar[d]<1.2ex> \ar[d]<-1.2ex> \ar[d]<0ex>
    &x_{a-1}\ar[d]|{\gamma_{a-1}} & x_{a+1} \ar[d]|{\gamma_{a+1}}
    & x_{a}\ar[d]|{\gamma_{a}} &x_{a+2}\ar[d]|{\gamma_{a+2}}
    & \dots\ar[d]<1.2ex> \ar[d]<-1.2ex> \ar[d]<0ex> & \ar[d]|{\gamma_{n}} x_n\\ 
& y_1 &\dots &y_{a-1}& y_{a+1} & y_{a} &y_{a+2}& \dots & y_n}\xymatrix@R=1pt{\\\\.}
\end{equation*}
\wwbiker s yield representations of the  welded braid group $\WBG_n$
defined  in the same way.

\smallskip
\Biker s have a level of structure which biracks do not have,
which is their underlying  groupoid $\Gamma$.
This  {leads the primary reason}
to introduce them. 
In particular consider the following. 

\smallskip

\mdef \label{de:gpoidalg}
The groupoid algebra $\C(\Gamma)$ of a groupoid $\Gamma$
\cite{ga1,willerton2008twisted,Morton}
is the free vector space $\C \Gamma_1$
with product:
\begin{equation}\label{productinGA}
\big (x \ra{\gamma} y\big) \big (x' \ra{\gamma'} y'\big)
=\delta(y,x') \big (x \ra{\gamma \star \gamma'} y' \big ).
\end{equation}
{(Given $x,y \in \Gamma_0$, we put $\delta(y,x)$ to be $1$ if $y=x$ and $0$ otherwise.) If the set of objects of $\Gamma$ is finite, then $\C(\Gamma)$ is a unital algebra with unit: $\displaystyle 1_{\C(\Gamma)}=\sum_{x \in \Gamma_0} \iota(x)=\sum_{x \in \Gamma_0} \big( x \ra{\id_x} x\big) $. {We also have a *-structure \peq{star}.}}

%

\smallskip

A \biker\  $(\G,X^+_\Gamma)$ gives rise to the following
invertible element in $\C(\Gamma)\tn \C(\Gamma)$: 
\begin{equation}\label{defofR-1}
\Rc_{(\G,X^+_\Gamma)} =\Rc=\sum_{x,y \in \Gamma_0}(x \ra{L(x,y)} x\backslash y) \tn (y \ra{R(x,y)} y/x)\in \C(\Gamma) \tn \C(\Gamma).  
\end{equation}
This $\Rc$  satisfies the  relation below
(cf. the relation satisfied by an R-matrix in a quasi-triangular bialgebra): 
\begin{equation}\label{propR-1}
\Rc_{12}\Rc_{13}\Rc_{23}=\Rc_{23}\Rc_{13}\Rc_{12}, \textrm{ in }  \C(\Gamma)\tn\C(\Gamma)\tn \C(\Gamma);
\end{equation}
see e.g. \cite[Thm. VIII.2.4]{Kassel}. Here  $$\Rc_{13}=\sum_{x,y \in \Gamma_0}(x \ra{L(x,y)} x\backslash y) \tn \id_{\C(\Gamma)}\tn (y \ra{R(x,y)} y/x), \quad \Rc_{12}=\Rc \tn \id_{\C(\Gamma)} \textrm{ and } \Rc_{23}=  \id_{\C(\Gamma)}\tn \Rc.$$
(N.B.: $\Rc$ is not in general an R-matrix \cite{Kassel},
and $\C(\Gamma)$ is not, a priori, a quasi-triangular bialgebra.)
Furthermore, a \biker\ is welded if, and only if, in $\C(\Gamma)\tn\C(\Gamma)\tn \C(\Gamma)$, {it holds that:}
\begin{equation}\label{propRw-1}
\Rc_{13}\Rc_{23}=\Rc_{23}\Rc_{13}. 
\end{equation}
A main result of this paper is that the representation in {\peq{upper-rep}} of $\VBG_n$ 
can be generalised to {\it braid}
any $n$-tuple $(V_1,\dots,V_n)$ of representations
of $\C(\Gamma)$
-- {see Thm. \ref{rep12} and \ref{uf}.} In particular, \eqref{propR-1}   implies that,  if $V$ is a  representation of $\C(\Gamma)$, then there is a representation $\trl^*$ of $\VBG_n$ on $V^{n \otimes}=V \tn \dots \tn V$, such that:
\begin{multline}\label{actb}
(v_1 \tn \dots\tn v_{a-1}  \tn v_a \tn v_{a+1}  \tn v_{a+2}  \dots \tn v_n)\trl^* S^+_a(n)\\
=\displaystyle\sum_{x,y \in \Gamma_0}v_1 \tn \dots \tn  v_{a-1}
  \tn v_{a+1}.\big (y \ra{R(x,y)} y /x\big)\tn v_a.
    \big (x \ra{L(x,y)} x\backslash y \big) \tn v_{a+2} \tn \dots \tn v_n,
\end{multline}
\begin{equation}\label{actc}(v_1 \tn \dots\tn v_{a-1}  \tn v_a \tn v_{a+1}  \tn v_{a+2}  \dots \tn v_n)\trl^* V_a(n)
=v_1 \tn \dots \tn  v_{a-1}
  \tn v_{a+1}\tn v_a \tn v_{a+2} \tn \dots \tn v_n.
\end{equation}
 This is a unitary representation if $V$ is a unitary representation of the groupoid algebra.
This coincides with the representation in  {\peq{upper-rep}} if $V=\C(\Gamma)$.
If   $(\G,X^+_\G)$ is a  welded \biker, then this representation of {$\VBG_n$} in $V^{n \otimes}$ descends to a representation of  $\WBG_n$. In particular the {welded move} {{\eqref{bisW}} follows from Equation \eqref{propRw-1}.

If $\C(\Gamma)$ can be given a quasi-triangular bialgebra structure and  $\Rc$ in \eqref{defofR-1} is its R-matrix then \eqref{actb} yields  exactly the {braid group} representation yielded by the quasi-triangular structure in  $\C(\Gamma)$; see \cite[VIII-3]{Kassel}.


\smallskip

{Of course the utility of all this abstract machinery depends on the
ability to construct \biker s and W-\biker s. We address this core point now.}

\smallskip

Let $G$ be a finite group. 
An example of unitary  representation of the welded braid group that follows from  {W-\biker s}
 is the representation of $\WBG_n$ on $V^{n \tn}$, arising from the R-matrix in the quantum double $D(G)\cong \C(\AUT(G))$  of the group algebra of $G$ \cite{Alt}, which is a quasi-triangular Hopf algebra. Here $V$ is a representation of $D(G)$.
This represention of $\WBG_n$ is normally only stated to be a representation of the braid group $\BG_n$
\cite{Kassel,Alt,Turaev,RT}, by using \eqref{actb}. However \eqref{actc}  extends it  trivially to a representation of  $\VBG_n$. Since the  R-matrix of $D(G)$ satisfies \eqref{propRw-1},  the representation in (\ref{actb},\ref{actc}) descends to $\WBG_n$.

\smallskip

\mdef\label{aut-def}{Let us now hence relate {W-\biker s} with the quantum-double $D(G)$.}
For $G$  a finite group, $\AUT(G)$ is defined to be the action groupoid of
the conjugation action of $G$ on $G$. The objects of 
$\AUT(G)$ are elements $g \in G$ and arrows  
have the form
$(g \ra{a} aga^{-1})$, where $g,a \in G$.
Composition in $\AUT(G)$ is
by group multiplication in reverse order. Hence in $\C(\AUT(G))$ the product on generators is: \begin{equation} \label{AUTprod} 
\big(g \ra{a} aga^{-1})\big(g'\ra{a'} a'g'{a'}^{-1})
= \delta(aga^{-1},g')
\big (g\ra{aa'} aa'g{a'}^{-1} a^{-1} \big),\textrm{ where } g,g',a,a' \in G.
\end{equation} %

\mdef\label{deQD} The quantum double $D(G)$ is the algebra $\C(\AUT(G))$  
made \cite{Kassel} a quasi-triangular 
Hopf algebra
by:
\begin{equation} \label{eq:hopf}
 \begin{split}
\Delta(x \ra{g} gxg^{-1})
  & =\sum_{yz=x} ( y \ra{g} gyg^{-1}) \tn (z \ra{g} gzg^{-1}),
   \hspace{.2632in} \epsilon((x \ra{g} gxg^{-1})) = \delta(x,1_G) ,
      \hspace{.28in} 
      \\
S(x \ra{g} gxg^{-1})
  & =(gx^{-1}g^{-1} \ra{g^{-1}} x^{-1}) , 
   \hspace{.7543in}    \Rc  =\sum_{g,h \in G}  (g\ra{h^{-1}} h^{-1}gh)  \tn (h \ra{1_G} h) .
\end{split}
\end{equation}

\mdef \label{deQD-1}
The \biker\ $X^+_G$ associated to $\AUT(G)$
is given, cf. Equation (\ref{gbirackcolouring}), by:
\begin{equation}\label{grouprack}
  \xymatrix@R=0pt{\\ \\  (g,h)\stackrel{X^+_G}{\longmapsto}}
 \xymatrix{& g \ar[dr]^{\,\,\,\bullet\,\, 1_G}|\hole &h\ar[dl]_<<<<<<<{{h}^{-1}\,\,	\,\bullet\,\,\,\,\,\, \,\,\,}\\
                                                          &h &  h^{-1}gh}\xymatrix@R=1pt{\\\\.}
\end{equation}
Therefore $\Rc_{(\G,X^+_G)}$ obtained from Equation \eqref{defofR-1} coincides with $\Rc$ in Equation \eqref{eq:hopf}. 
Note that the underlying birack of $X^+_G$ is the `conjugation quandle' in $G$: $h/ g=h$ and $g \backslash h=h^{-1} g h$; see \cite{C1,C2,C3}.

\smallskip
So we can see that even though $X^+_G$ in \eqref{grouprack} is apparently only a very simple
spin-off of the  conjugation quandle,  the holonomy arrows in $X^+_G$ add  the information needed for generating
the representations of the braid group derived from the R-matrix in the quantum double $D(G)$. Hence $X^+_G$ is considerably stronger than the conjugation  quandle of $G$ alone. {In particular, the invariants of knots derived from the $R$-matrix of $D(G)$ {take into account not only the knot group of a knot}, as does the conjugation quandle of $G$, but also \cite{prox} the entire peripheral structure of a knot \cite{eisermann}, which the conjugation quandle alone cannot uncover.}


{As we will see below in \S\ref{motivation1}, the holonomy arrows in $X^+_G$ naturally arise from
Aharonov-Bohm phases expected in finite group topological field theory
in the 2-disk; see \cite{Bais1,Bais2,Bais3,Bais4,Mochon1,bianca}.}

\smallskip

\mdef
For future reference,
{for $X^+_G$,} the representation of $\VBG_2$ on $\C(\AUT(G))\tn \C(\AUT(G))$ from \peq{upper-rep} is:
\begin{align}
  \big((a^{-1} ga  \ra{a} g) \tn  (b^{-1} hb  \ra{b} h)\big)
  \triangleleft^* \SSS{1}{2}
  &\; = (b^{-1} hb  \ra{b} h) \tn (a^{-1} ga  \ra{h^{-1}a} {h^{-1}gh})\label{partact12}, \\
  \big((a^{-1} ga  \ra{a} g) \tn  (b^{-1} hb  \ra{b} h)\big)
  \triangleleft^* \VVV{1}{2}
  &\; = (b^{-1} hb  \ra{b} h) \tn  (a^{-1} ga  \ra{a} g).
\end{align}
%
%

\smallskip

\mdef
Our main example here {\eqref{compact} and \eqref{X3}  is  a W-\biker\ $X^+_R$ in  a certain groupoid ${\TRANS\big(T^2_R(\Gc)\big)}$ motivated
by discrete higher gauge theory in the 3-disk.}
Here $\Gc$ is a finite
2-group (represented by a crossed module).  In this case, the underlying birack of $X^+_R$ is a full fledged {W-birack}  and not simply a quandle.

\smallskip

In this $({\TRANS\big(T^2_R(\Gc)\big)},X^+_R)$ case, no
underpinning quasi-triangular Hopf algebra appears to be in hand for modelling the
representations of the welded braid group thereby obtained.
Groupoid algebras are examples of weak Hopf algebras
\cite{Gabriella,turaev_et_al,etingof_et_al}. 
{However} 
the
representations of the welded braid group derived from
${\TRANS\big(T^2_R(\Gc)\big)}$  do not appear to arise from a {\it bona
  fide} R-matrix inside $\C({\TRANS\big(T^2_R(\Gc)\big)})$, at least as far
as the definition of quasi-triangular weak Hopf algebras appearing in
\cite{turaev_et_al} is concerned.
These difficulties in interpreting
our representations of $\WBG_n$ in terms of already know
construction were our main reason  to introduce \biker s,
which in addition have the already mentioned advantage that
they naturally incorporate Aharonov-Bohm phases featuring in
topological field theory in their very construction.  

\smallskip
\mdef {The  W-bikoid $X^+_R$, motivated by higher gauge theory, can be derived from one of the  form $X^+_{gr}$, below. Here we have an abelian $gr$-group, i.e. a group $G$ left-acting, by automorphisms, in an additive abelian group $A$. Form the semidirect product $G \ltimes A$ and define $\TRANS(G,A)=\AUT(G\ltimes A)$ to be the action groupoid of the conjugation action of $G \ltimes A$ on itself. Thus arrows of
$\TRANS(G , A)$ 
 are of the form:} 
$$
\big((g,a)\ra{(w,k)} (w,k)(g,a)(w,k)^{-1}\big)
  = \big(wgw^{-1},k+w \trr a-(wgw^{-1}) \trr a)\big),
      \textrm{ where }  g,w \in G \textrm{ and } a, k \in A.
$$
{The corresponding W-\biker\   $X^+_{gr}$ takes the form:}
\begin{equation}\label{xgr1}
\begin{split}
\xymatrix@R=1pt{\\\\ X^+_{gr}\Big((z,a),(w,b)\Big)=}\hskip-1cm
\xymatrix{& (z,a)  \ar[dr]|\hole
  \ar[dr]|<<<<<<<{(w^{-1},0_A)\,\,\,\,\bullet\,\,
    \qquad\quad\,\,\,}|\hole
  & (w,b)  \ar[dl]\ar[dl]|<<<<<<<{\qquad \qquad \,\,\,\, \,\,\,\,\bullet \,\,\,\,(1_G , -{w}^{-1}\trr a)}\\
  &(w,a+b-w^{-1} \trr a)  & \big(w^{-1}zw,w^{-1} \trr a\big)}
\end{split}\xymatrix@R=1pt{\\ .}
\end{equation}
{This W-bikoid is very different from the  $X^+_{G \ltimes A}$ we would obtain from \eqref{grouprack}.}

 {A topological explanation for the existence  of  the underlying W-birack of {$X^+_{gr}$,} in terms of elementary algebraic topology -- namely in terms of the action of $\pi_1$ on $\pi_2$, is done in \S\ref{bandh}. (The pertinent space is the 3-disk $D^3$ minus an unlinked union of unknotted circles.)
The latter W-birack had appeared previously in \cite{martins_kauffman}, where it arised from Yetter TQFT \cite{Yetter,martins_porter} of the complement of  a knotted surface in $S^4$ \cite{martins_JKTR,martins_tams}. In \cite{martins_kauffman} it is proven that this W-birack {can} distinguish between different welded knots with the same knot group.} 
%

 \smallskip

In a future publication \cite{prox}, we will inspect invariants of virtual and welded links derived from \biker s.

\subsubsection*{Structure of the paper}
{In \S \ref{physmot}, we sketch the reason of why W-bikoids appear in the context of topological gauge theory in $D^2$ and topological higher gauge theory in $D^3$. We start in \S \ref{motivation1} by giving an overview of Bais Flux metamorphosis  \cite{Bais1,Bais2,Bais3,Bais4}, and how it is related to quantum doubles and finite group W-bikoids. In \S\ref{phys0} we give a brief overview of discrete higher gauge theory, and in particular of the 2-dimensional holonomy of a 2-connection and its behaviour under gauge transformations. In \S \ref{phys} we explain how our main example of W-bikoids is related to loop-particles moving in $D^3$. In \S \ref{phys} we also propose a version of Bais 
flux metamorphosis in order to handle higher gauge fluxes of unknotted and unlinked loop particles in topological higher gauge theory. Section \ref{physmot} is roughly independent of the rest of the paper. However it motivates the main constructions.}

{In \S \ref{M-Preliminaries} we recap conventions for: groupoids, groupoid algebras, virtual and welded braid groups, and loop braid groups, and sketch the definition of the isomorphism $\WBG_n \to \LBG_n$. In \S\ref{M-bk}, we firstly \S\ref{thegr} explain the wreath product  $\gwrr{n}$ of a groupoid with the symmetric group $\Sigma_n$, and in \S\ref{ss:msgc} we give a graphical calculus for $\gwrr{n}$, providing a visual framework to handle a lot of intricate calculations. Conventions for biracks are in \S\ref{c:biracks}. Bikoids and W-bikoids are defined in \S\ref{s:biker}. In \S\ref{abeliangrgroups} we explain the W-bikoids $X^+_{gr}$ derived from abelian $gr$-groups, and show how they arise from elementary algebraic topology. In \S\ref{representations} we show how W-bikoids give rise to representations of the loop braid group, when loops are coloured with representations of the groupoid algebra of the underlying groupoid of the W-bikoid. In  \S\ref{ss:4} we show how finite 2-groups (described as crossed modules) yield W-bikoids, hence proving that the representations of $\LBG_n$ announced in \S \ref{phys} exist.}

\vspace{.1in} \noindent\thanks{{\bf Acknowledgements.} {Alex Bullivant and Paul Martin thank
  EPSRC for funding under Grant EP/I038683/1. {Jo\~{a}o Faria Martins and Paul Martin thank the Leverhulme trust for funding under the Research Project Grant RPG-2018-029.}
 We all thank  Zoltan Kadar,
  Marcos Cal\c{c}ada and Celeste Damiani for discussions.}}

\section{{Insights from topological gauge theory and higher gauge theory}}\label{physmot}

\subsection{{Finite group  gauge theory in {$D^2$},
  quantum doubles
  and braid groups}}\label{motivation1}

{This section contains a review of Bais flux metamorphosis
\cite{Bais1,Bais2,Bais3,Bais4,Mochon1,bianca},
organised and reinterpreted with our lift to higher gauge theory in mind.
(A connection between finite group topological gauge theory in {$D^2$}
and topological phases in (2+1)-dimensions
can be made via the
Kitaev quantum double model \cite{PachosB,kitaev}.})

Let $G$ be a  
finite group.
We consider a topological gauge theory with gauge group $G$, with
spatial manifold the unit disk
$D^2=[0,1]^2=\{(x,y)\,\,|\,\, 0\leq x, y\leq 1\} $.
Here  
a gauge field is interpreted as being a principal $G$-bundle.
We suppose the restriction to the boundary $S^1$ to be `static' in time. 
{Note that a  principal $G$-bundle, where $G$ is finite has a unique connection, which is automatically flat.}

\medskip

\mdef \label{1sing}
Consider a set of $n$ anyonic point-particles $p_1,\dots,p_n$
moving in the interior of $D^2$.
We model these
by a flat
$G$-connection in $D^{2}$, which becomes singular at the location of
each anyon. At that location  a point-like magnetic vortex arises.
The
magnetic vortex   
at $p_i$ is
formally 
classified by its
magnetic flux.
The latter is
the gauge holonomy $g_i \in G$ 
observed when travelling along a positively oriented small
circle $c_i$, looping around the particle $p_i$. We let  ${*_{c_i}}$ be
the coinciding initial and  end-point of ${c_i}$; {cf. Fig. \ref{generic}}.
Both $c_i$ and
${*_{c_i}}$ will be co-moving with the particle $p_i$.
{A naive  
1-particle Hilbert space
is  given by the
group algebra {$\C G$}.} 

\smallskip
\mdef\label{Amb-1g} Holonomies $g \in G$ are  
not physical  since the
holonomy around a particle is defined  only up to conjugation. 
 The
 reason is twofold.
  Firstly, gauge transformations in the
 connection change holonomy along a  closed path by conjugation by an
 element of $G$. Secondly, it is more realistic to consider the connection holonomy along of a loop $\hat{c_i}$, looping around $p_i$, however starting and ending at a base-point $*$ far away from $p_i$. Different choices {of} $\hat{c_i}$ lead  also to conjugate values for the holonomy.
(In anticipation, we note that in higher gauge theory, discussed in \S\ref{phys},
 the lifts of 
 these two types of tranformation in the way we calculate magnetic flux do not {act} exactly in the same way on the observed 2-dimensional holonomy; see {\peq{lastc}.})

\smallskip
Hence one could consider  
the Hilbert space
{$\C CC(G)$}
spanned by the  conjugacy classes of $G$ to be the relevant one for
describing finite group topological gauge theory particles in the plane.
However in general there exists a further decomposition of the
physical Hilbert space corresponding to internal charge degrees of
freedom. These charge degrees of freedom are truncated when switching
from {$\C G$ to $\C CC(G)$; \cite[Page 1196]{benini}.} For instance, given a flux
$g\in G$, conjugation by the $g$-central subgroup
$Z(g)=\{h\in G|gh=hg\}\subset G$
is a symmetry of the flux degree of freedom and as
such the action can be further decomposed into representations of
$Z(g)$. The latter correspond to possibly non-trivial charges associated to
anyons.

\smallskip
\mdef\label{HS}{Finally, 
cf.  \cite{Bais2,Bais3},
a convenient formulation of the Hilbert space associated to a
point-particle in finite group topological gauge theory in $D^2$ is
given by the underlying vector space of the groupoid algebra
$\C(\AUT(G))$; see \peq{aut-def}.}
This $\C(\AUT(G))$ retains enough information to allow the treatment
of multiple anyons and of charge. {The Hilbert space for $n$-particles configurations is then
$\C(\AUT(G)) \otimes \dots \otimes \C(\AUT(G))$.}

{Cf. the discussion in \cite[\S 2.1]{benini}, the vector space $\C(\AUT(G))$ is a `resolved' version of {$\C CC(G)$,} i.e. a derived quotient of {$\C G$} under the conjugation action. Hence switching from {$\C CC(G)$}  to  $\C(\AUT(G))$ can  be  motivated by the general principle that derived quotients retain more information and  are as a rule well behaved  in comparison to naive quotients, see e.g. \cite[\S 2.1]{Costello}. The latter point of view will be prevalent when discussing topological higher gauge theory in $D^3$; see \S\ref{phys}.}

\smallskip

\mdef \label{whyresolve} {In general, if $X$ is a set of possible formal configurations of a particle and $H$ is a group of symmetries
acting on $X$, then we retain enough information in the Hilbert space of the particle if we take this space to have basis the morphisms of the action
groupoid (Def. \ref{de:ag}) of the action of $H$ on $X$.  A `resolved'  state of the system is thus a formal state $x \in X$ together
with a datum $h \in H,$ which  encodes how it was measured.}


An inner product in $\C(\AUT(G))$ which renders different arrows orthogonal is:
$$
\big\langle(g \ra{a} aga^{-1}), (g' \ra{a'} a'g'{a'}^{-1}) \big\rangle
   =\delta(a,a') \delta(g,g').
$$ 
Note that the group $G$ left-acts in $\C\big(\AUT(G)\big)$ by
$\langle,\rangle$-unitary
transformations as:
\begin{equation}\label{actionnothigher}
\smash{
h. ( g \ra{a} aga^{-1})= ( g
\ra{ha} haga^{-1}h^{-1}), \textrm{ where } g,h,a \in G.}
\end{equation}

The groupoid algebra $\C(\AUT(G))$  is additionally an example of a  Hopf algebra  (see \peq{deQD}) describing the
local symmetries of  finite group topological gauge theory
\cite{Bais2,Bais3,Bais4,kitaev} in the 2-disk. 
References justifying the latter fact by {identifying} the
quantum-double algebra as the coarse graining algebra for topological
gauge theory are \cite{bianca,wen_lan}.
Particle types in finite-group topological gauge theory should  thus be labelled by irreducible representation of $\C(\AUT(G))$.  
%



\smallskip

A Hilbert space on its own has {very} little information
(Hilbert spaces of equal dimension are isomorphic).
We must
specify what the observables are and what they mean physically.

Given a $g \in G$,
the {\em flux operator}
$F_g\colon \C(\AUT(G))\to \C(\AUT(G))$ for a single particle is defined as
$F_g(x \ra{a} axa^{-1})= \delta(g,axa^{-1}) (x \ra{a} axa^{-1})$.
Classically, {particle fluxes}
depend on the way we travel {around a particle} and also on the choice of
gauge. Let us specify what the flux observables mean by stating how
the magnetic flux  
(i.e. the connection holonomy around each particle) is calculated. 

\smallskip

 \mdef \label{gaugerest} Let our base-point $*$ be $(1/2,0)\in \d D^2$.  For simplicity, we will only consider gauge transformations on our discrete gauge fields which are the identity in $\d D^2$. {This is a very strong restriction, which we will also make when discussing topological higher gauge theory in \S\ref{phys2}; see \peq{defR}.}
Given that our gauge fields are stable in $*$, we can also choose and fix a point in the fibre of $*$. Hence the holonomy along a path starting and ending at $*$ will from now on be uniquely defined. This will not make the definition of flux around a particle $p_i$ unique, as some freedom still exists in the actual choice of {closed path} looping around $p_i$.

{
Let $c_{i}$ be a counter-clockwise oriented path enclosing $p_{i}$ and no other $p_{j}$ for $j\neq i$, with initial point $*_{c_{i}}$, lying on the straight line connecting the base-point $*\in \partial D^{2}$ and $p_{i}$. We then choose a path $\gamma_i$, call it a ``connecting path'', 
connecting the base-point $*$ to the initial point ${*_{c_i}}$. {This} leads to an unambiguous definition of magnetic flux $g_i \in G$ of a discrete gauge field around $p_i$, to be the holonomy along the closed path $\hat{c_i}=\gamma_i c_i \gamma_i^{-1}$, connecting $*$ to $*$.
}{Homotopically distinct choices of paths in the definition of $\g_i$ generally yield
conjugate values for $g_i$.
However in the following, by allowing the motion of the 
particles $p_1,\dots,p_n$, there is no way paths
$\gamma_i$, $i=1,\dots,n$ connecting $*$ to ${*_{c_i}}$ can
be chosen both once and for all, and in a
way such that they depend continuously of time. For this reason we will consider particles in generic positions.}

\smallskip

\mdef
{Let us say that
the particles $p_1,\dots,p_n$ are in {\em generic positions} if their
$x$-coordinates are all different (see Fig. \ref{generic}). 
In particular straight lines $\gamma_i$ connecting $*$ to ${*_{c_i}}$ do not intersect any particle.}
For generic configurations of particles, we will assume that the flux  of each particle 
corresponds to the one arising from the holonomy along $\gamma_i c_i \gamma_i^{-1}$, {for  $\gamma_i$ a straight line connecting $*$ to ${*_{c_i}}$.}

\begin{figure}
\centerline{\relabelbox 
\epsfysize 3.5cm 
\epsfbox{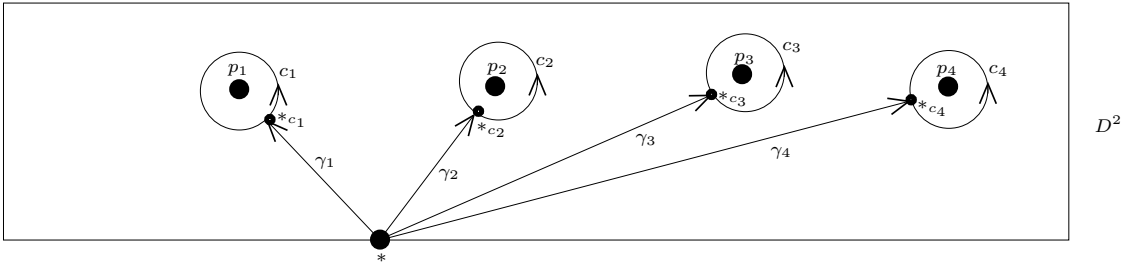}
\relabel{s}{$\scriptstyle{*}$}
\relabel{g1}{$\scriptstyle{\gamma_1}$}
\relabel{g2}{$\scriptstyle{\gamma_2}$}
\relabel{g3}{$\scriptstyle{\gamma_3}$}
\relabel{g4}{$\scriptstyle{\gamma_4}$}
\relabel{D}{$\scriptstyle{D^2}$}
\relabel{c1}{$\scriptstyle{c_1}$}
\relabel{c2}{$\scriptstyle{c_2}$}
\relabel{c3}{$\scriptstyle{c_3}$}
\relabel{c4}{$\scriptstyle{c_4}$}
\relabel{a}{$\scriptstyle{{*_{c_1}}}$}
\relabel{b}{$\scriptstyle{{*_{c_2}}}$}
\relabel{c}{$\scriptstyle{{*_{c_3}}}$}
\relabel{d}{$\scriptstyle{{*_{c_4}}}$}
\relabel{p1}{$\scriptstyle{p_1}$}
\relabel{p2}{$\scriptstyle{p_2}$}
\relabel{p3}{$\scriptstyle{p_3}$}
\relabel{p4}{$\scriptstyle{p_4}$}
\endrelabelbox}
\caption{\label{generic}A generic configuration of four particles $p_1,p_2,p_3,p_4$ in the disk $D^2$.}
\end{figure}


%
In the following we utilise the convention that the  tensor
components in the Hilbert space  $\C(\AUT(G)) \otimes \dots \otimes \C(\AUT(G)$
correspond to the particles ordered by increasing $x$ value, here left to right.

\smallskip

\mdef
We now  consider the  {transformations} on
the Hilbert space under motions of the particles. 
Moving particles adiabatically in a way such that
the configuration remains generic
does not change the state of the system.
Let the particles $p_i, i=1,\dots,n$ move, in between $t=0$ and $t=1$,
in such a way that we momentarily pass (say at $t=1/2$) through a
non-generic configuration.
{For each $i$, and at each  $t$, let $\g_i^t$ be the path, in $D^2 \setminus \{p_1,\dots,p_n\}$, obtained from a straight line from $*$ to the position of ${*_{c_i}}$, at time $t$. Note that $\gamma_i^t$ might be undefined at $t=1/2$, since the straight line from $*$ to $*_{c_i}$ may cross   another particle.}

For each $i$, let also $\phi_i^t$ be the path obtained by
concatenating $\gamma_i^{0}$ with the trajectory of ${*_{c_i}}$ in
between time $0$ and time $t$. Note that $\phi_i^0=\gamma_i^0$ and
that  $\phi_i^t$ is homotopic to $\gamma_i^t$, if $0\leq t <
1/2$. However $\gamma_i^t$ is in general not homotopic to $\phi_i^t$
for $t>1/2$. See Fig. \ref{trans}, below, for a two particles
configuration.
\begin{figure}[H]
\centerline{\relabelbox 
\epsfysize 3.8cm 
\epsfbox{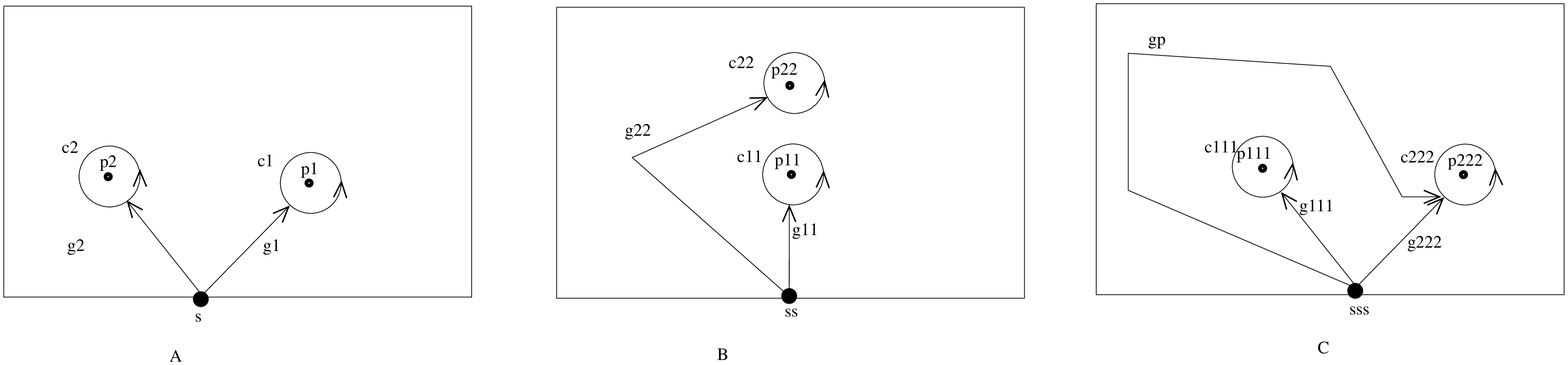}
\relabel{s}{$\scriptstyle{*}$}
\relabel{ss}{$\scriptstyle{*}$}
\relabel{sss}{$\scriptstyle{*}$}
\relabel{g1}{$\scriptstyle{\gamma_2^t}$}
\relabel{g2}{$\scriptstyle{\gamma_1^t=\phi_1^t}$}
\relabel{g11}{$\scriptstyle{\gamma_2^t\cong\phi_2^t}$}
\relabel{g22}{$\scriptstyle{\phi_1^t}$}
\relabel{g111}{$\scriptstyle{\gamma_2^t\cong\phi_2^t}$}
\relabel{g222}{$\scriptstyle{\gamma_1^t}$}
\relabel{gp}{$\scriptstyle{\phi_1^t}$}
\relabel{c1}{$\scriptstyle{c_2}$}
\relabel{c2}{$\scriptstyle{c_1}$}
\relabel{c11}{$\scriptstyle{c_2}$}
\relabel{c22}{$\scriptstyle{c_1}$}
\relabel{c111}{$\scriptstyle{c_2}$}
\relabel{c222}{$\scriptstyle{c_1}$}
\relabel{p1}{$\scriptstyle{p_2}$}
\relabel{p2}{$\scriptstyle{p_1}$}
\relabel{p11}{$\scriptstyle{p_2}$}
\relabel{p22}{$\scriptstyle{p_1}$}
\relabel{p111}{$\scriptstyle{p_2}$}
\relabel{p222}{$\scriptstyle{p_1}$}
\relabel{A}{$t=0$}
\relabel{B}{$t=1/2$}
\relabel{C}{$t=1$}
\endrelabelbox}
\caption{\label{trans} {A movement of particles passing non-generic configurations. Note  $\gamma_{\,\,2}^{\, \,t}$ is homotopic to $\phi_2^t$, for all $t$.} }
\end{figure}

When passing by non-generic configurations,  `Aharonov-Bohm phases' \cite{Bais1,Bais2,Bais3,Bais4,Nayak} $A_i\in G$ must be inserted. They make up for the fact that $\gamma_i^t$ and  $\phi_i^t$ might not be homotopic when $t=1$.
A possible convention is that, for each {$i=1,\dots,n$,} the group
element  $A_i$  is given by the holonomy of the $G$-connection along
the path {$\phi_i^t ({\gamma_i^t})^{-1}$,} at time $t=1$. 


These holonomies $A_i$ can  be {determined}
\cite{Bais1,Bais2,Bais3,Bais4}.
For instance in Fig \ref{trans},  if the particle $p_1$ carries
 flux $g$, and the particle $p_2$ carries flux $h$, then the holonomy along 
the path $\phi_2^t ({\gamma_2^t})^{-1}$ is trivial, when $t=1$,
i.e. $A_2=1_G$. {This is because the path is homotopically trivial, hence the associated holonomy is  the group identity.}
On the other hand, the holonomy along  {$\phi_1^t ({\gamma_1^t})^{-1}$,}
at $t=1$, is $A_1=h^{-1}$. {(Since at $t=1$, $\phi_1^t ({\gamma_1^t})^{-1}$ is homotopic to {$\gamma_2^tc_2^{-1}({\gamma_2^t})^{-1}$.)}}
Therefore, the  movement of particles in Fig. \ref{trans}  induces the following unitary transformation
in the Hilbert space $\C(\AUT(G))\tn\C(\AUT(G))$
(Bais calls this ``flux metamorphosis'' \cite{Bais1}):
\begin{equation} \label{eq:rhorep}
\begin{split}
  (a^{-1} ga  \ra{a} g) \tn  (b^{-1} hb  \ra{b} h)&
  \; \mapsto  \; A_2.(b^{-1} hb  \ra{b} h) \;\tn\;  A_1.(a^{-1} ga  \ra{a} g)\\
  &\;= \; (b^{-1} hb  \ra{b} h) \;\tn\;  h^{-1}.(a^{-1} ga  \ra{a} g)\\
   &\;=\; (b^{-1} hb  \ra{b} h) \;\tn\; (a^{-1} ga  \ra{h^{-1}a} h^{-1}gh^{}).
\end{split}
\end{equation}
Note that the operative part of this can be expressed as $g \tn h
\mapsto h \tn h^{-1} gh$.

\mdef
Remark: Equation \eqref{eq:rhorep} is exactly  \eqref{partact12}.
Hence, the  \biker\ $X^+_G$ in \eqref{grouprack} is related to
finite group topological gauge theory in $D^2$
(just forget about the `elementary transposition' component).
Considering $n$-particle configurations \eqref{eq:rhorep}
yields  representations  of the braid group
$\BG_n$ in  
$\C(\AUT(G))^{n \otimes}$ \cite{Bais1,Bais2}.

%
%

\smallskip

\mdef\label{mpg} Fix particle locations $p_1,\dots,p_n \in D^2$. Adiabatic exchanges of
particles, are modelled by diffeomorphisms
$(D^2,\{p_1,\dots,p_n\}) \to (D^2,\{p_1,\dots,p_n\})$
(by definition these are diffeomorphisms
$D^2 \to D^2$ which are the identity in the boundary $S^1$ of $D^2$ and send
$\{p_1,\dots,p_n\}$ to $\{p_1,\dots,p_n\}$).
The braid group  $\BG_n$ is
isomorphic to the group of isotopy classes of diffeomorphisms
$(D^2,\{p_1,\dots,p_n\})\to (D^2,\{p_1,\dots,p_n \} )$; see
\cite{birman,brendle_hatcher}. In topological gauge theory,  maps
induced by adiabatic exchanges of particles on Hilbert spaces should
depend only on isotopy classes of diffeomorphisms
$(D^2,\{p_1,\dots,p_n\})\to  (D^2,\{p_1,\dots,p_n\})$. This is the
case when maps on Hilbert spaces are, as outlined above, induced by
holonomies (to encode Aharonov-Bohm phases) of flat connections in
$D^2$ with magnetic vortices at particle locations. The same
principle can be made to work in the loop braid group case and higher
gauge theory in the 3-disk $D^3$, as we now explain.


\subsection{Higher gauge theory with a finite 2-group: preliminaries}
\label{phys0}

{Higher gauge theory is a  
generalisation
  of gauge theory which
enables  the definition of non-abelian holonomies along surfaces
embedded in a manifold $M$,
when a principal 2-bundle with a 2-connection over $M$ is given; see
\cite{BaezHuerta11,baez_schreiber,schreiber_waldorf1,schreiber_waldorf2,martins_picken}.
See \cite{martins_picken} for a cubical approach
\cite{brown_higgins_sivera} to the holonomy of a 2-connection.
This is the point of view used here. In particular, a  2-bundle with a 2-connection over $M$, will be subordinated to an open cover of $M$, whose open sets are called {\em coordinate neighbourhoods} \cite[\S 3.1]{martins_picken}.}

{In discrete higher gauge theory, 
  for the
  gauge field, instead of a principal bundle
  with a finite group $G$ of structure, we have a principal 2-bundle
  (see \cite{baez_schreiber,Wockel} {and \cite[\S 3.1]{martins_picken}})
  with a structure finite  2-group. 
  Note that a principal 2-bundle with a finite 2-group of structure has a unique 2-connection.}

\ignore{{
{DELETE THIS PARA.?:}
  Topological phases with higher gauge symmetry are discussed in \cite{Our1,Our2,kapustin,Wang2017}.
A motivation for this paper
was to determine  whether {representations of the loop braid group} can be defined
from finite 2-group topological higher gauge theory
{\cite{Our1,Our2,Martins_Mikovic,Zucchini}} in the 3-disk.
This would be via 
2-dimensional holonomies of loop-particles' trajectories, to encode
higher order Aharonov-Bohm effects inherent to loops moving in
(3+1)-dimensional space.
Similar ideas are in \cite{EN,LL}. The loop braid group
models  the space of adiabatic exhanges of unknotted unlinked
horizontal loops in the 3-disk, modulo isotopy.
}} 


\smallskip\medskip

\mdef\label{PREL} In the strict version of higher gauge theory used
here, a 2-group is equivalent to a
crossed module $\Gc=(\d\colon E \to G,\trr)$ of groups,
\cite{brown_higgins_sivera,brown_hha,baez_lauda}. Here $G$ and $E$ are groups,
$\d\colon E \to G$ is a group map, and $\trr$ is a left action of $G$
on $E$ by automorphisms, such that the {1st and 2nd} Peiffer relations are satisfied:
 $\d(g\trr e)=g\d(e)g^{-1}$ and $\d(e) \trr f=efe^{-1}$, where
$g \in G$ and $e,f \in E$. We will suppose that $E$ and $G$ are  finite.

A {\em square} in $\Gc$ \cite[\S2.2.1]{martins_picken} is, {by definition}, a diagram {like:}
\begin{equation}\label{sg}
\hspace{.31in}
  \xymatrix@R=5pt@C=5pt{&\ast
   \ar[rr]^{x}\ar@{<-}[dd]_z && \ast \ar@{<-}[dd]^y\\ &&e&
   \\ &\ast \ar[rr]_w && \ast } \xymatrix@R=0pt{ \\
    \;\;\;\;
   \textrm{, where } x,y,z,w \in G \textrm{ and } e \in E \textrm{ satisfy } \d(e)=wyx^{-1}z^{-1}.}
\end{equation}

\medskip

\mdef\label{dgp} Squares in $\Gc$ can be composed vertically and  horizontally,
{if} their sides match \cite[\S2.2.1]{martins_picken}:  
\begin{equation}\label{compos}
\xymatrix@R=5pt@C=5pt{&\ast
   \ar[rr]^{x}\ar@{<-}[dd]_z && \ast \ar@{<-}[dd]^y\ar[rr]^{x'}   && \ast\ar@{<-}[dd]^{y'} \\ &&e& & e'
   \\ &\ast \ar[rr]_w && \ast \ar[rr]_{w'} && \ast }\xymatrix@R=0pt{\\\\=}\xymatrix@R=5pt@C=5pt{&\ast
   \ar[rr]^{xx'}\ar@{<-}[dd]_z && \ast \ar@{<-}[dd]^{y'}\\ &&(w \trr e')\,\,e&
   \\ &\ast \ar[rr]_{ww'} && \ast } \xymatrix@R=0pt{\\\\\textrm{and}}
 \xymatrix@R=5pt@C=5pt{&\ast
   \ar[rr]^{x}\ar@{<-}[dd]_z && \ast \ar@{<-}[dd]^y\\&&e'\\&\ast
   \ar[rr]|{w}\ar@{<-}[dd]_{z'} && \ast \ar@{<-}[dd]^{y'}\\&&e&
   \\ &\ast \ar[rr]_{w'} && \ast\ } \xymatrix@R=0pt{\\\\=}\xymatrix@R=5pt@C=5pt{&\ast
   \ar[rr]^{x}\ar@{<-}[dd]_{z'z} && \ast \ar@{<-}[dd]^{y'y}\\ &&e\,\, z'\trr e'&
   \\ &\ast \ar[rr]_{w'} && \ast } \xymatrix@R=1pt{\\\\.}
\end{equation}
{These compositions are associative, have vertical and horizontal identities, and satisfy the interchange law;  see \cite[2.2.1]{martins_picken}, \cite[\S 6.6]{brown_higgins_sivera} or \cite{Our2} for explanation. Also, each square in $\Gc$ has a vertical and a horizontal inverse.}

\medskip

In higher gauge theory,  fields are  2-connections in a 2-bundle over a
  manifold $M$.
A 2-connection is a cubical-$\Gc$-2-bundle with a connection,
as  defined in  \cite[Defs. 3.1 and 3.4]{martins_picken}.
There is an equivalence relation \cite[Def. 4.18]{martins_picken} on the set of 2-connections, by gauge transformations. It is derived from a  2-groupoid of  2-connections, gauge transformations and 2-gauge transformations between gauge transformations \cite{schreiber_waldorf1,schreiber_waldorf2,martins_picken}.

\medskip

\mdef
If we have a 2-connection  in $M$,
then given a piecewise smooth map $P\colon [0,1]\times [0,1] \to M$,
{and a choice of coordinate neighbourhoods covering the image $P([0,1]^2)$,}
we can define \cite[Def. 5.1]{martins_picken} {the 2-dimensional holonomy $2{\rm hol}(P)$ along $P$}, which is a
square
in $\Gc$ \cite{Our2,martins_picken,brown_higgins_sivera},} as in \eqref{sg}.
{This $2{\rm hol}(P)$ depends only on the choice of coordinate neighbourhoods covering $P(\d [0,1]^2)$; see \cite[Cor. 5.5]{martins_picken}.}

{If $\gamma \colon [0,1] \to M$ is a path, there is also a 1-dimensional holonomy ${\rm hol}(\gamma)\in G$ of $\g$; see \cite[\S5.1.3]{martins_picken}. Holonomy $h={\rm hol}(\gamma)$ along a closed path is well defined, up to transformations like $h \mapsto \d(a)ghg^{-1},$ where $g\in G$ and $a \in E$.}  The elements of $G$ appearing in the edges of the square  $2{\rm hol}(P)$  in \eqref{sg} are given by the 1-dimensional holonomies of the paths obtained by restricting $P$ to the edges of the boundary of $[0,1]^2$. 

\smallskip

\mdef\label{prescomp}
Piecewise smooth maps of form $P\colon [0,1]\times [0,1] \to M$
  can  be composed horizontally,
  respectively vertically, when the restrictions to the appropriate
  side of  $[0,1]\times [0,1]$ match, by rescaling \cite[\S2.3.1]{martins_picken}.

  The 2-dimensional holonomy of a  2-connection preserves both horizontal
and vertical compositions of maps $P\colon [0,1]^2 \to M$; see \cite[Thm. 4.3]{martins_picken}. Each map $\Gamma  \colon [0,1]^2 \to M$ has horizontal $\Gamma^{-H}$  and vertical $\Gamma^{-V}$ inverses,  given by  $(t,s) \mapsto \G(1-t,s)$ and $(t,s) \mapsto \G(t,1-s)$. The 2-dimensional holonomy of a 2-connection preserves  horizontal and vertical inverses.

\smallskip

Maps $P\colon [0,1]\times [0,1] \to M$ such that the paths $[0,1] \to M$
obtained by restricting $P$ to $\{0\} \times [0,1]$ and to
$\{1\} \times [0,1]$  coincide are called {\em tubes in $M$}, as they
can be seen as maps $P\colon S^1 \times [0,1] \to M$.
If $P$ is a tube in $M$,  the 2-dimensional holonomy along $P$ has the form:
\begin{equation}
\xymatrix@R=0pt{\\ 2{\rm hol}(P)=}\xymatrix@R=5pt@C=5pt{&\ast
   \ar[rr]^{x}\ar@{<-}[dd]_{{h}} && \ast \ar@{<-}[dd]^{{h}}\\ &&{f}&
   \\ &\ast \ar[rr]_w && \ast } \xymatrix@R=0pt{ \\
   \textrm{, where } x,{{h}},w \in G \textrm{ and } {f} \in E \textrm{ satisfy } \d({f})=w{{h}}x^{-1}{{h}}^{-1}.}
\end{equation}

\mdef We will also consider {\em tori-maps in $M$}, which by definition are tubes in $M$ such that the paths $[0,1] \to M$ obtained by restricting to $ [0,1]\times \{0\}$ and to $ [0,1]\times \{1\}$  coincide. If $P$ is a torus-map in $M$, then the 2-dimensional holonomy of a 2-connection along $P$ has the form below {(cf. \cite[Def. 5.16]{martins_picken}):}
\begin{equation}\label{torusinG}
\xymatrix@R=0pt{\\ 2{\rm hol}(P)=}\xymatrix@R=5pt@C=5pt{&\ast
   \ar[rr]^{x}\ar@{<-}[dd]_{{h}} && \ast \ar@{<-}[dd]^{{h}}\\ &&{f}&
   \\ &\ast \ar[rr]_x && \ast } \xymatrix@R=0pt{ \\
   \textrm{, where } x,{{h}} \in G \textrm{ and } {f} \in E \textrm{ satisfy } \d({f})=x{{h}}x^{-1}{{h}}^{-1}={[x,h]}.}
\end{equation}
{This $2{\rm hol}(P)$, for $P$ a torus-map, is independent of the choice of coordinate neighbourhoods along $P(\d [0,1]^2)$, and the gauge equivalence class of the 2-connection, up to \cite[Thm. 5.17]{martins_picken} transformations of the form of a simultaneous horizontal and vertical conjugation in the `double-groupoid' of squares in $\Gc$, as in \eqref{leftrightconj}, below. {(We note that the squares in the middle left and the middle right are horizontal inverses of each other, and similarly the squares in the middle top and the middle bottom are vertical inverses of each other.)}}

\begin{equation}\label{leftrightconj}
\xymatrix@R=5pt@C=5pt{\\\\\\&\ast
   \ar[rr]^{x}\ar@{<-}[dd]_{{h}} && \ast \ar@{<-}[dd]^{{h}}\\ &&{f}&
   \\ &\ast \ar[rr]_x && \ast }\xymatrix{\\ \\ \mapsto}
\xymatrix@C=2pt@R=5pt{ \ast\ar[rr]^{1_G}  && \ast\ar[rr]^{x'} &&\ast \ar[rr]^{1_G}   &&\ast \\
 & 1_E && g^{-1}\trr {e}^{-1} && 1_E\\
 \ast\ar[rr]|{g} \ar[uu]^{1_G} && \ar[uu]|{g^{-1}}\ast\ar[rr]|{x} &&\ast \ar[rr]|{g^{-1}}\ar[uu]|{g^{-1}}   &&\ast\ar[uu]_{1_G} \\
 & {a}^{-1} && f && g^{-1} \trr {a}\\
 \ast\ar[rr]|{g} \ar[uu]^{h'} && \ar[uu]_{h}\ast\ar[rr]|{x} &&\ast \ar[rr]|{g^{-1}}\ar[uu]^{h}   &&\ast\ar[uu]_{h'} \\
 & 1_E &&  {e} && 1_E\\
 \ast\ar[rr]_{1_G} \ar[uu]^{1_G} && \ar[uu]_{g}\ast\ar[rr]_{x'} &&\ast \ar[rr]_{1_G}\ar[uu]^{g}   &&\ast\ar[uu]_{1_G} 
}\xymatrix{\\ \\ =}\xymatrix@R=5pt@C=5pt{\\\\\\&\ast
   \ar[rr]^{x'}\ar@{<-}[dd]_{{h'}} && \ast \ar@{<-}[dd]^{{h'}}\\ &&{f'}&
   \\ &\ast \ar[rr]_{x'} && \ast }\xymatrix@R=1pt{\\\\\\\\\\\\\\\\.}
\end{equation}
Where:
$\d({e})=x'g{x}^{-1}g^{-1},$ {thus $x'=\d(e)gxg^{-1}$,} \quad \quad $\d({a})^{-1}=ghg^{-1}{h'}^{-1}$, {thus $h'=\d(a) ghg^{-1}$,}  and:  $$f'=  {e} \,\, (gxg^{-1})\trr {a}\,\, g \trr f\,\,  (ghg^{-1})\trr {e}^{-1}\,\, {a}^{-1}=x' \trr {a}\,\, {e} \,\, g \trr f\,\, (ghg^{-1}) \trr {e}^{-1}\,\, {a}^{-1}.$$


\mdef \label{flatprop}  {A 2-connection is {\em 2-flat} if its categorical curvature
vanishes \cite[Def. 2.16]{martins_picken} in all coordinate neighbourhoods  of $M$ \cite[Prop. 3.3.6]{schreiber_waldorf1}. 2-connections on 2-bundles with finite 2-groups are automatically 2-flat.}

\smallskip

\mdef \label{flatHOL} 
{Fix a flat 2-connection in $M$. The 2-dimensional holonomy along $P\colon [0,1]^2 \to M$
depends only on the homotopy class of the map
$P\colon   [0,1]^2 \to M$,
relative to the boundary of $[0,1]^2$, {and the choice of coordinate neighbourhoods covering $P(\d [0,1]^2)$.} This follows by \cite[Thm. 4.5]{martins_picken}, by the same argument as in the proof of 
 \cite[Thm. 5.8]{martins_picken}, or by \cite[Prop. 3.3.6]{schreiber_waldorf1}.}
{For torus-maps $P\colon [0,1]^2 \to M$, we consider homotopies of $P$ relative to $\{0,1\}\times I$,  which stay within the set of torus maps. If we fix coordinate neighbourhoods covering $P(\{0,1\} \times [0,1])$, then $2{\rm hol}(P)$ is homotopy invariant, up to transformations as in  \eqref{leftrightconj}, with {$e=1_E$, $g=1_G$.}}

\smallskip

\mdef\label{holB} Let $B^2=\{z \in \R^2\colon ||z||\leq 1\}$ be based at $*=(-1,0)$. Hence $B^2$ and $D^2$ are homeomorphic. Suppose that $P\colon B^2 \to M$ is piecewise smooth. It is more convenient to consider the 2-dimensional holonomy along $P$ to have the form: $2{\rm hol}(P)=\hskip-0.3cm\xymatrix{
{\,\,\,\,\ast\,\,\,e}\ar@<1ex>@(dr,ur)|{\,\,\,v} }$, where $e \in E$ and $v=\d(e)$; cf. \eqref{sg}, here we are putting $v=wyx^{-1}z^{-1}$. Note that $v$ is the 1-dimensional holonomy of the path obtained by restricting $P$ to the  boundary $S^1$ of $B^2$, oriented counterclockwise.
If $B\subset M$ is diffeomorphic to $B^2$, and has a base point
$* \in \d B$, then the 2-dimensional holonomy along an orientation and
base-point preserving diffeomorphism $B^2 \to B$ depends only on $B$, see
{\cite[Thm. 67]{Our2} or \cite[Thm. 5.14]{martins_picken}.}
It will be denoted $2{\rm hol}(B)$. 

\subsection{Higher gauge theory with a finite 2-group
  and invariants of loop braids} \label{phys} \label{phys2}

\smallskip

{Consider a 2-flat 2-connection \peq{flatprop} in the 3-disk
$D^3=\{(x,y,z)\,|\, 0\leq x,y,z \leq 1\}$.
Let us also assume that the 2-connection is {`static'} in the boundary $\d D^3.$}
We choose an oriented `base-loop' $O\cong S^1$ embedded in  $D^2 \times \{1\} \subset \d D^3$. We let $*\in O$ be the base point of $O$. Hence $O$ is the image of a path {$\g\colon [0,1] \to \d D^3$}, starting and ending at $*$. Also $O$ bounds a 2-dimensional ball $B$ contained in $\d D^3$; see Fig. \ref{cap}, below.

\smallskip

\mdef\label{defR}
{A {parameter}  we will need is the  2-dimensional holonomy {\peq{holB}} along $B$ which has the form:}
\begin{equation}\label{HR}
2{\rm hol}(B)=\hskip-0.3cm\xymatrix{
{\,\,\,\,\ast}\ar@<0.5ex>@(dr,ur)_{\d(R)\,\,\,\,\,\,\,\,\,\,\,\,\,\,\,\,\,\,\,\,\,\,\,\,\,\,} } \hskip-24mm\scriptstyle{R}\qquad \quad.
\end{equation}
Here $R \in E$. The 1-dimensional holonomy along $O$,
with initial point $*$, is $\d(R) \in G$.
Cf. {\peq{gaugerest}, as for gauge theory in $D^2$,
  we will only consider gauge transformations
  \cite[\S4.2.1]{martins_picken} on {2-connections} which are trivial in $\d(D^3)$. We  also fix coordinate neighbourhoods along $O$. Hence we take $R$ and $\d R$  to be constant in time.}

\smallskip

\mdef
Consider a  set of small unknotted
unlinked circles $c_1,\dots,c_n$ (call them ``loop particles'') in $D^3$.
Outside these, the  2-connection is flat.
We allow for a 2-connection to be singular in the loop-particles
$c_1,\dots,c_n$, which therefore may carry `higher gauge magnetic vortices'. 

\smallskip

\mdef\label{allowed}
{We suppose that loop-particles can move, but 
remain horizontal.
In other words each loop-particle is at each time $t$ contained in a plane (whose height may vary with $t$) parallel to $D^2 \times \{0\}\subset D^3 $. }


\smallskip

\mdef Let $c$ be  loop-particle.
A 
formal
observable for $c$ {(here called {\em magnetic primary 2-flux})}
is the
2-dimensional holonomy $F_c=2{\rm hol}(P_c)$ along a {torus-map} $P_c\colon [0,1]^2\to D^3$, starting and ending at the base-loop
$O$, and bounding the torus $T^2_c$ obtained as the boundary  of a
local neighbourhood of $c$ as in Fig. \ref{cap} below.
{Cf. the figure in Equation  \eqref{eq:fig1.2} for
  a visualisation of a tube-map homotopic to $P_c$.}


\smallskip 

Let us explain $P_c$. It is constructed as the vertical composition of three tubes in $D^3$. The torus $T^2_c\cong S^1 \times S^1$ has a meridian $m$ and a longitude $l$. Both $m$ and $l$ are assumed to be co-moving with $c$. The linking number of $l$ and $c$ should be zero. We consider a `connecting tube' $\Gamma_c$ connecting $O$ to $l$. Then consider the torus-map in $D^3$ obtained by sweeping the torus $T^2_c$, in the obvious way, from $l$ to $l$. And finally we go back to $O$ by using the vertical {inverse} $\Gamma_c^{{-V}}$ of $\Gamma_c$. (Hence $\Gamma_c^{{-V}}(t,s)=\Gamma_c(t,1-s)$.)
$$
\centerline{\relabelbox 
\epsfysize 5cm 
\epsfbox{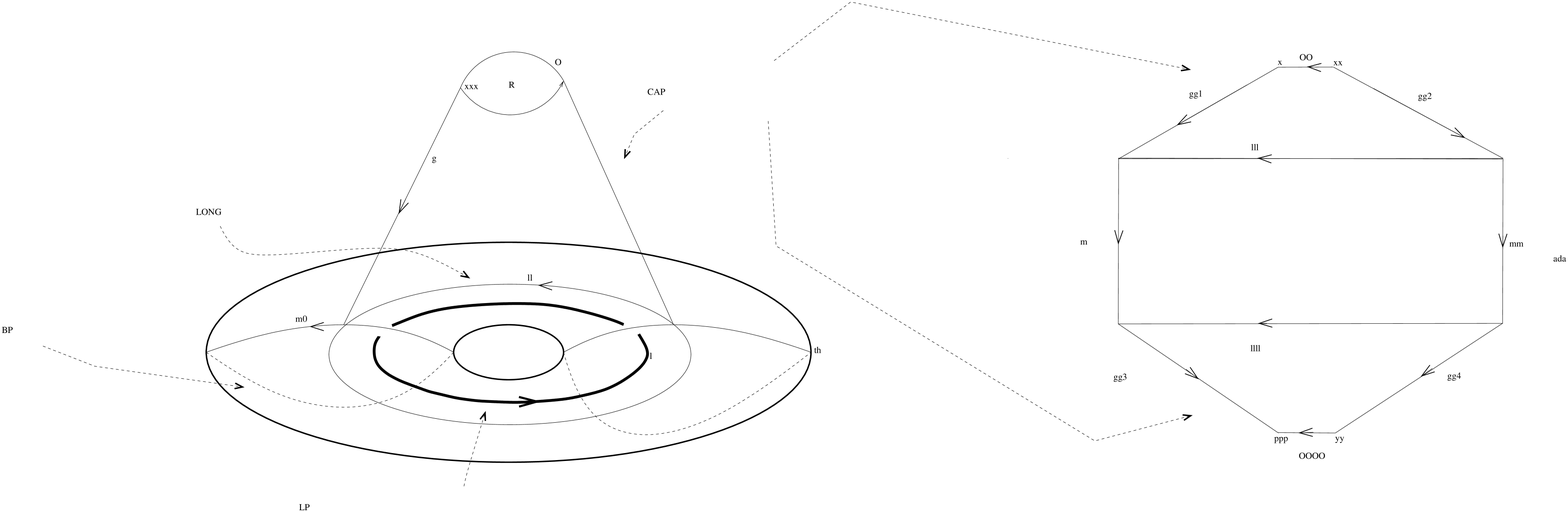}
\relabel{CAP}{$\textrm{connecting tube }\Gamma_c$}
\relabel{BP}{$\textrm{meridian }m$}
\relabel{LONG}{$\textrm{longitude }l$}
\relabel{LP}{$\textrm{loop-particle }c$}
\relabel{ppp}{$\scriptstyle{*}$}
\relabel{yy}{$\scriptstyle{*}$}
\relabel{l}{$\scriptstyle{c}$}
\relabel{ll}{$\scriptstyle{l}$}
\relabel{lll}{$\scriptstyle{l}$}
\relabel{llll}{$\scriptstyle{l}$}
\relabel{m}{$\scriptstyle{m}$}
\relabel{mm}{$\scriptstyle{m}$}
\relabel{m0}{$\scriptstyle{m}$}
\relabel{g}{$\scriptstyle{\gamma}$}
\relabel{gg1}{$\scriptstyle{\gamma}$}
\relabel{gg2}{$\scriptstyle{\gamma}$}
\relabel{gg3}{$\scriptstyle{\gamma^{-1}}$}
\relabel{gg4}{$\scriptstyle{\gamma^{-1}}$}
\relabel{O}{$\scriptstyle{O}$}
\relabel{OO}{$\scriptstyle{O}$}
\relabel{OOOO}{$\scriptstyle{O}$}
\relabel{x}{$\scriptstyle{*}$}
\relabel{xx}{$\scriptstyle{*}$}
\relabel{xxx}{$\scriptstyle{*}$}
\relabel{R}{$\scriptstyle{B}$}
\relabel{ada}{$=P_c$}
\relabel{th}{${T^2_c} $}
\endrelabelbox}
$$
\vskip-1cm
\begin{figure}[H]
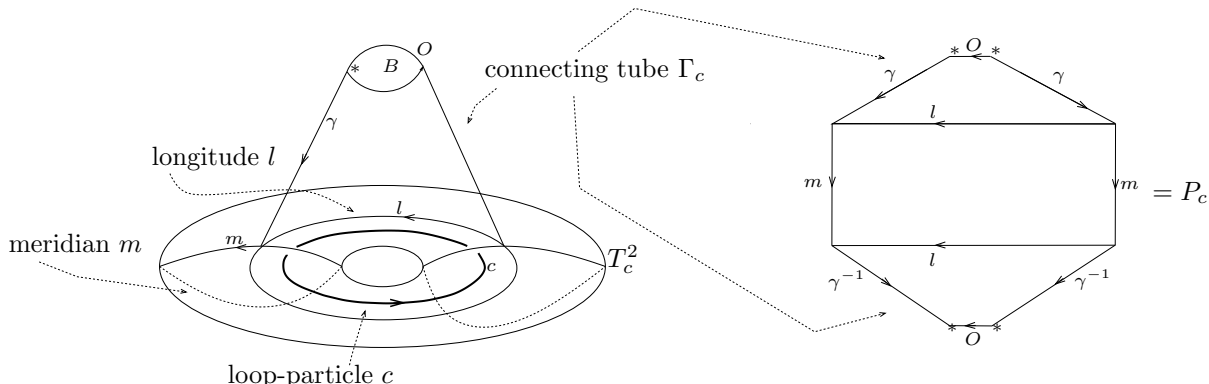

\caption{\label{cap}\protect {Construction of  a torus-map} $P_c\colon [0,1]^2 \to D^3$ in $D^3$, associated to a  loop-particle $c$.}
\end{figure}

\noindent {The 2-dimensional holonomy $2{\rm hol}(P_c)$  
has the form \eqref{torusinG},
and hence by \peq{defR} the form  \eqref{2hol1}
below. 
Note that $h$ is the 1-dimensional holonomy along
$\gamma  m \gamma^{-1}$, and $\d(R)$ is the 1-dimensional holonomy along $O$.}
\begin{equation}\label{2hol1}
\xymatrix@R=0pt{\\2{\rm hol}(P_c)=}
\xymatrix@R=5pt@C=7pt{
&\ast \ar@{->}[rr]^{\d(R)}\ar@{<-}[dd]_{{h}}& &\ast\ar@{<-}[dd]^{{h}}
\\ & &{f}&\\
 &\ast \ar@{->}[rr]_{\d(R)}&&\ast
}\xymatrix@R=1pt{\\\textrm{, where } {h}\in G, {f} \in E \textrm{ and } \d(f)={[\d(R),{h}]}.}
\end{equation}


\mdef\label{2f}
An encoding of magnetic primary 2-flux $F_c$ 
is thus
a triple $F(c)=({h},\d(R),{f})\in G \times G \times E$
such that $\d({f})=[\d{R},{h}]$.
The relation between $F(c)=({h},\d(R),{f})$ and $2{\rm hol}(P_c)$ is as in \eqref{2hol1}. 

We let $T^2_R(\Gc)$ be the set of those triples,  called {\em 2-fluxes}.

{An important property of 2-fluxes is that they form a group (as do fluxes $g \in G$ of gauge theory) under the vertical composition of squares in $\Gc$;
see \eqref{compos}. Hence:}
$(g,\d(R),a) (g',\d(R),a')=(gg', \d(R), a\,\, g \trr a').$

\smallskip

\mdef\label{lastc} {As for  topological gauge theory flux \peq{Amb-1g},
${F(c)=(h,\d(R),f)}$  
  is not physical.
It changes if we  perform gauge transformations in the
2-gauge field \cite[\S 4.2.1 and \S 5.1.7]{martins_picken} or change the chosen coordinate neighbourhoods along $P_c(\d{[0,1]^2)}$; see \cite[Thm. 5.1.4]{martins_picken}.  {Cf. \peq{defR}, in order to simplify the construction, we fix coordinate neighbourhoods along the base loop, and the disk $B$, and only consider gauge transformations which are trivial along $\d D^3$. In particular, we are considering only a subset of gauge transformations which leave the holonomy of the base-loop invariant. This} means that in  \eqref{leftrightconj} we take $g=1_G$ and {$e=1_E$}. Hence $F(c)$ will change as $$(h,\d(R),f)\to ({\d(a)h},\d(R),R\, {a}\, R^{-1}\,\,f\,\,  {a}^{-1}\big),$$ under gauge transformations; here $e \in E$. {We are using the 2nd Peiffer relation here; see \peq{PREL}.}}

\smallskip

{We can also modify the  connecting tube
$\G_c$, connecting the base-loop $O$ to the longitude $l$ (see Fig. \ref{cap}), to a new connecting tube $\G_c'$. This leads to transformations in the 2-flux $F(c)=(h,\d(R),f)$ like below:}
$$(h,\d(R),f)\mapsto (g,\d(R),{e}) \,\,(h,\d(R),f) \,\,(g,\d(R),{e})^{-1}=(ghg^{-1},\d(R), {e}\,\, g \trr f\,\, (ghg^{-1}) \trr {e}^{-1}\big).$$
{Here  $(g,\d(R),{e}) \in T^2_R(\Gc)$.}
{Note that  $(g,\d(R),{e})$ arises as the 2-dimensional holonomy of the {torus-map} obtained as the vertical concatenation of $\Gamma_c'$ and the vertical {inverse} of $\Gamma_c$.}

 {Putting the two types of transformations together, we can see that the appropriate group of  transformations on 2-fluxes is given by the following semidirect product:}
$$
T^2_R(\Gc)\ltimes E=
\{(g,\d(R),{e,a})\in G \times G \times E \times E\colon \d({e})
=[\d{(R)},g]\}.
$$
with group operation
$$
(g,\d(R),{e,a)}\,\, \; (g',\d(R),{e',a')}
\;= \; (gg',\d(R),{e} \,\, g \trr {e}', {a} \,\, g \trr {a}') .
$$
%

{The group $T^2_R(\Gc)\ltimes E$   acts on the set  $T^2_R(\Gc)$ of 2-fluxes as below (calculations are in \S\ref{agrpoids}):}
$${(g,\d(R),{e,a}) \t (h,\d(R),f)  \; = \;
(\d({a})ghg^{-1},\d(R),R\, {a}\, R^{-1}\,\,  {e}\,\, g \trr f\,\, (ghg^{-1}) \trr {e}^{-1}\,\, {a}^{-1}\big).}$$
{This action also arises from \eqref{leftrightconj}, where we put $x'=\d(R)$. (We are using the 2nd. Peiffer relation \peq{PREL}.)}



\smallskip

\mdef
We define $\TRANS(T^2_R(\Gc))$ as the action groupoid of the
action $\t$ of $T^2_R(\Gc)\ltimes E$ on $T^2_R(\Gc)$.
Objects of $\TRANS(T^2_R(\Gc))$ are given by 2-fluxes $F\in T^2_R(\Gc)$.
{Morphisms are, where $T=(g,\d(R),e,a)\in T^2_R(\Gc) \ltimes E$:}
$$
\big ( (h,\d(R),f)
  \ra{(g,\d(R),{e,a})}  (g,\d(R),{e,a}) \t
  (h,\d(R),f)   
  \big)
  \; = \; \big(F \ra{T} T \t F\big).
$$

\mdef \label{2gaugetrans0} {In higher gauge theory, we not only have gauge transformations between
2-connections but also 2-gauge transformations between gauge
transformations.  Hence there is also a quotient groupoid $\overline{\TRANS(T^2_R(\Gc))}$ of  $\TRANS(T^2_R(\Gc))$, where gauge transformation connected by 2-gauge transformations are identified. Throughout most of the paper,  we neglect the role of  {2-gauge transformations}. We will come back to this in \peq{2gaugetrans}.}

\smallskip

\mdef\label{another}
Finally, we note that a different type of 2-flux
$\hat{F}(c)\in T^2_R(\Gc)$, called {\it thin 2-flux},
can be associated to $c$ in Fig. \ref{cap}.
This $\hat{F}(c)$ is given by the  2-dimensional holonomy of  a  {`degenerate' (i.e 1-dimensional)} tube $\hat{P_c}$, from $O$ to $O$. This  $\hat{P_c}~$ is  obtained from contracting $O$ to its base point $*$, along $B$, then concatenating with a {1-dimensional} tube tracing $\gamma m \gamma^{-1}$, and then using $B$ to go back to $O$ again. {Cf. the figure in \eqref{eq:fig2} for a visualisation of a tube-map homotopic to  $\hat{P_c}$.} We have a group morphism $\Theta_R\colon T^2_R(\Gc) \to T^2_R(\Gc)$  given by $(g,\d(R),e)\mapsto (g,\d(R),R\,\, g \trr R^{-1})$.
{From the assumptions in \peq{defR}, we can conclude that primary and thin 2-fluxes of a loop-particle $c$ are related as:}
$$
\hat{F}(c)=\Theta_R\big(F(c)\big) .
$$

\smallskip



\smallskip

{Cf. \peq{HS} and \peq{whyresolve}.}
In analogy with $\C(\AUT(G))$ {in} topological gauge theory, for
an unknotted loop-particle $c$, 
{an approximation of the}
 algebra of
local symmetries of  finite 2-group topological higher gauge theory in $D^3$  is
the groupoid algebra $\C({\TRANS\big(T^2_R(\Gc)\big)})$.
Elementary loop-particles in our {simplified} model correspond to irreducible representations of
 $\C({\TRANS\big(T^2_R(\Gc)\big)})$. {The latter algebra is semisimple, as is
any finite groupoid algebra over $\C$.}

Now recall the $G$-action on $\AUT(G)$ in  \eqref{actionnothigher}.
We have a left-action ``$.$'' of
$T^2_R(\Gc)\ltimes E$ on 
${\TRANS\big(T^2_R(\Gc)\big)}$, such that:
\beq \label{eq:leftact1}
T'. (F \ra{T } T\t F)=(F \ra{T'T} (T'T)\t F),
\eq
where $F\in T^2_R(\Gc)$ and $T,T'\in T^2_R(\Gc) \ltimes E$. 
{(We more generally {should} have a 2-group action \cite{Morton_Picken} of the underlying 2-group in
${\TRANS\big(T^2_R(\Gc)\big)}$.
This will be developed {in a forthcoming publication}.)}

\smallskip
Let ${\rm proj}\colon D^2\times [0,1] \to D^2$ be $(x,y,z)\mapsto (x,y)$.
Consider two loop-particles $c_1,c_2\in D^2 \times [0,1]$ in {\em generic position}. This means that ${\rm proj}(c^1 \cup c^2)$ is the disjoint union of two circles, such that no circle is nested inside the other, and moreover  the middle-points of the disks spanned by $c_1$ and $c_2$ have different $x$ coordinates.
\smallskip

\mdef \label{rect}
If $c_1$ and $c_2$ are in generic position,  we can consider connecting tubes $\Gamma_1$ and $\Gamma_2$, {given by {rectilinear} cylinders connecting the base-loop $O$ to their their longitudes $l_1$ and $l_2$, such that the line connecting the base-point of $O$ to the base-point of $l_{i}$ is a straight line}, as in Fig. \ref{cap2}. We call these connecting tubes {\em rectilinear connecting tubes}. Assuming that the configuration is generic, these rectilinear connecting  tubes $\Gamma_1$ and $\Gamma_2$  do not intersect the regular neighbourhoods of the other loop-particle, if these regular neighbourhoods, with boundary $T^2_{c_1}$ and  $T^2_{c_2}$, are chosen to be thin enough.
\begin{figure}[H]
\centerline{\relabelbox 
\epsfysize 4.5cm 
\epsfbox{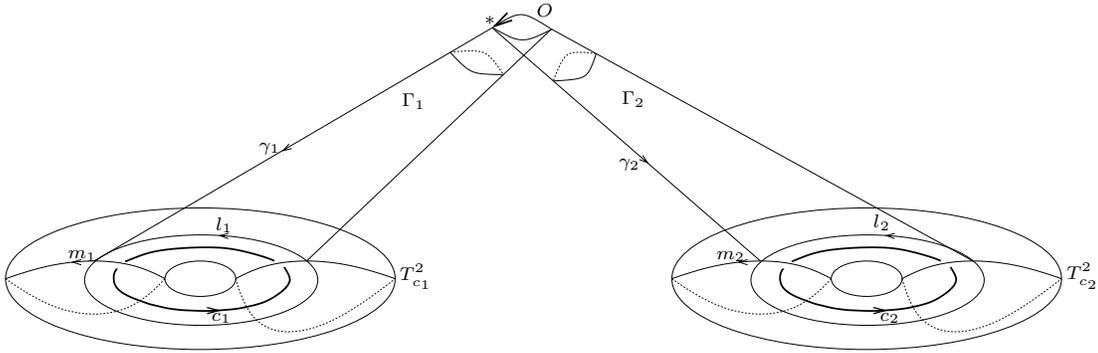}
\relabel{P}{$\scriptstyle{*}$}
\relabel{ll}{$\scriptstyle{l_1}$}
\relabel{l}{$\scriptstyle{c_1}$}
\relabel{l1}{$\scriptstyle{c_2}$}
\relabel{ll1}{$\scriptstyle{l_2}$}
\relabel{m1}{$\scriptstyle{m_2}$}
\relabel{m0}{$\scriptstyle{m_1}$}
\relabel{g}{$\scriptstyle{\gamma_1}$}
\relabel{g1}{$\scriptstyle{\gamma_2}$}
\relabel{G1}{$\scriptstyle{\Gamma_1}$}
\relabel{G2}{$\scriptstyle{\Gamma_2}$}
\relabel{O}{$\scriptstyle{O}$}
\relabel{T1}{$\scriptstyle{T^2_{c_1}}$}
\relabel{T2}{$\scriptstyle{T^2_{c_2}}$}
\endrelabelbox}
\caption{\label{cap2} A generic configuration {of} loop-particles $c_1$ and $c_2$.}
\end{figure}


\mdef\label{primthin} The Hilbert space 
$\C(\TRANS(T^2_R(\Gc))\tn \C(\TRANS(T^2_R(\Gc))$  describes the set of internal states for the pair of loop-particles, $c_1$ and $c_2$, in generic position. We order circles from left to right. The primary and thin 2-fluxes $F(c_i)$ and $\hat{F}(c_i)$ of the loop-particle $c_i$ are calculated by using a rectilinear tube {from} $O$ to $l_i$.

\smallskip

We now wish to consider the  {transformations} on
the Hilbert space arising when the  loop-particles $c_1$ and $c_2$
move.
Note,
cf. \peq{mpg}, that if there is an Aharonov--Bohm {like} effect, then
the pertinent mapping class group is now the loop braid group
$\LBG_2$ in two circles \S\ref{motion}.


\smallskip



\newcommand{\twohol}{2{\rm hol}}  
\newcommand{\TT}{{\mathsf T}}  

\smallskip
\mdef \label{conv} \label{pa:ABconv}
Fix a motion of our system, and in particular a motion of each loop
$c_i$, and hence {of} each longitude $l_i$. 
Let $\Gamma_i^t$ be the rectilinear tube \peq{rect}
from $O$ to the longitude $l_i$, at time $t$, where $i=1,2$.
Let
 ${\rm Traj}_i^{(t_0,t_1)}$ be the `tube' in $D^3$ made by the
trajectory of the longitute $l_i$
{between $t_0$ and $t_1$}.
Let $A_i(t_0,t_1)\in T^2_R(\Gc)$ be
the 2-dimensional holonomy of the
torus $\TT_i^{t_0 , t_1}  = \Gamma_i^{t_0} {\rm Traj}_i^{(t_0,t_1)}
                                 (\Gamma_i^{t_1})^{{-V}}  $, namely:
$$
A_i(t_0,t_1)  = \; 
\twohol( \Gamma_i^{t_0} {\rm Traj}_i^{(t_0,t_1)}
                                 (\Gamma_i^{t_1})^{{-V}} ), 
$$
{where $(\Gamma_i^{t_1})^{-V}$ is the vertical inverse of $\Gamma_i^{t_1}$.}

Our
model  
is that the Aharonov-Bohm phase
applied to $c_i$
for a motion from $t_0$ to $t_1$
is given by:
$A_i =(A_i(t_0,t_1),1_E) \in T^2_R(\Gc)\ltimes E$.
Thus for example if the motion exchanges the loops we will have:
\beq \label{eq:master0}
{\st{T_1}\otimes\st{T_2} \; \mapsto \; \st{A_2 T_2} \otimes \st{A_1 T_1}}.
\eq
%


\mdef \label{solveamb} {The 2-dimensional holonomy along a torus-map is not homotopy invariant as explained in \peq{flatHOL}, whereas in the motion group picture for loop braids groups, the trajectories ${\rm Traj}_i^{(t_0,t_1)}$ are only defined up to isotopy \cite{baez_et_al}. Hence, when determining the 2-dimensional holonomy along $\TT_i^{t_0 , t_1} $ we must fix a representative for the homotopy class of $\TT_i^{t_0 , t_1} $.} 

Crucially for the construction in this paper, these representatives are chosen to derived from the {torus-maps} $P_c$ and $\hat{P_c}$ with which we calculate the primary and thin 2-flux of our loop particles $c$ \peq{primthin}, so we always look to the `closest' $P_c$ and $\hat{P_c}$  to $\TT_i^{t_0 , t_1} $. {Additionally we will assume that homotopically trivial  tubes give rise to trivial 2-holonomies. That such strong 
assumptions give rise to a loop-braid group representation is a {mystery} to the authors which will be investigated in future work. With such approximations a notion of Aharonov-Bohm phases can be determined from the primary $F(c)$ and thin $\hat{F}(c)$ 2-fluxes of our loop particles $c$.} We  now explain this.

\smallskip

\mdef
Suppose that {between $t=0$ and $t=1$} two loops
$c_1$ and $c_2$ exchange places as indicated in \eqref{trivial}, below:  
\vskip-0.3cm
\begin{minipage}[t]{0.85\textwidth}
\begin{equation*}
{\centerline{\relabelbox 
\epsfysize 1.6cm 
\epsfbox{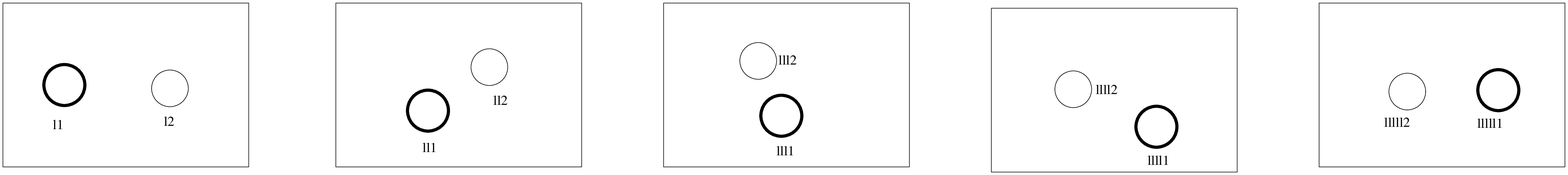}
\relabel{l1}{$\scriptstyle{c_1}$}
\relabel{l2}{$\scriptstyle{c_2}$}
\relabel{ll1}{$\scriptstyle{c_1}$}
\relabel{ll2}{$\scriptstyle{c_2}$}
\relabel{lll1}{$\scriptstyle{c_1}$}
\relabel{lll2}{$\scriptstyle{c_2}$}
\relabel{llll1}{$\scriptstyle{c_1}$}
\relabel{llll2}{$\scriptstyle{c_2}$}
\relabel{lllll1}{$\scriptstyle{c_1}$}
\relabel{lllll2}{$\scriptstyle{c_2}$}
\endrelabelbox}}
\end{equation*}
\end{minipage}
\begin{minipage}[t]{0.11\textwidth}
\quad \newline
\quad \newline
\begin{equation}\label{trivial}\end{equation}
\end{minipage}

\smallskip
\noindent
--- these, and further such figures below, are views from $O$ in the
ceiling of $D^3$.
%
As the 2-dimensional holonomies associated to the tubes traced by
$c_1$ and $c_2$ in \eqref{trivial} are trivial, {since the homotopy classes of  $\TT_i^{t_0 , t_1} $ are trivial, for $i=1,2$},
then 
by \peq{solveamb} the Aharonov-Bohm phases to insert in
\eqref{eq:master0} are trivial.
Thus
this move  is associated to the following map
$V$     
on the Hilbert space: 
\begin{equation}\label{2fmeta1}
  {(T_1^{-1} \t F_1 \ra{T_1}  F_1) \tn (T_2^{-1}\t F_2 \ra{T_2}  F_2)
    \xmapsto{(-)\trl^*V_1[2]} (T_2^{-1} \t F_2 \ra{T_2}  F_2) 
\tn  (T_1^{-1}\t F_1 \ra{T_1}  F_1)}.
\end{equation}
{(We just swap tensor components.)}
We will write $V $  as $ (-) \trl^*V_1[2]$ in anticipation of the role of $\LBG_2$.


%
%


\smallskip

\mdef
Now suppose  that the loop-particles $c_1$ and $c_2$ swap positions
in the way indicated in figure \eqref{lp}-\eqref{lp3} below.
We need to determine Aharonov-Bohm phases $A_1,A_2\in T^2_R(\Gc) \ltimes E $
associated to the movement of the loop-particles.
We suppose that this movement of loop-particles happens in between $t=0$ and $t=3$. 
\vskip-0.3cm
\begin{minipage}[t]{0.85\textwidth}
\begin{equation*}
{\centerline{\relabelbox 
\epsfysize 3.7cm 
\epsfbox{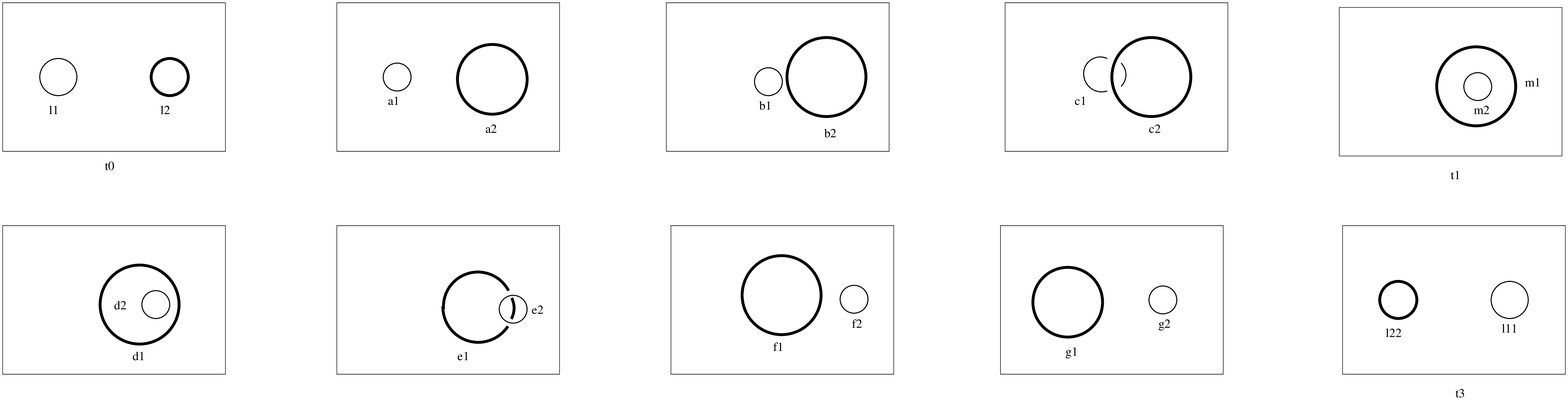}
\relabel{t1}{$\scriptstyle{t=1}$}
\relabel{t3}{$\scriptstyle{t=3}$}
\relabel{t0}{$\scriptstyle{t=0}$}
\relabel{l1}{$\scriptstyle{c_1}$}
\relabel{l11}{$\scriptstyle{c_1}$}
\relabel{l2}{$\scriptstyle{c_2}$}
\relabel{l22}{$\scriptstyle{c_2}$}
\relabel{l1}{$\scriptstyle{c_1}$}
\relabel{l2}{$\scriptstyle{c_2}$}
\relabel{a1}{$\scriptstyle{c_1}$}
\relabel{a2}{$\scriptstyle{c_2}$}
\relabel{b1}{$\scriptstyle{c_1}$}
\relabel{b2}{$\scriptstyle{c_2}$}
\relabel{m2}{$\scriptstyle{c_1}$}
\relabel{m1}{$\scriptstyle{c_2}$}
\relabel{c1}{$\scriptstyle{c_1}$}
\relabel{c2}{$\scriptstyle{c_2}$}
\relabel{d1}{$\scriptstyle{c_2}$}
\relabel{d2}{$\scriptstyle{c_1}$}
\relabel{e1}{$\scriptstyle{c_2}$}
\relabel{e2}{$\scriptstyle{c_1}$}
\relabel{f1}{$\scriptstyle{c_2}$}
\relabel{f2}{$\scriptstyle{c_1}$}
\relabel{g1}{$\scriptstyle{c_2}$}
\relabel{g2}{$\scriptstyle{c_1}$}
\endrelabelbox}}
\raisetag{50pt}
\end{equation*}
\end{minipage}
\begin{minipage}[t]{0.10\textwidth}
\quad \newline
\quad \newline
\begin{equation}\label{lp}\end{equation}
\quad \newline
\quad \newline
\begin{equation}\label{lp3}\end{equation}
\end{minipage}


\smallskip

\mdef
In order to {determine} $A_1, A_2$, first 
recall \peq{2f} and \peq{another} that a loop-particle $c_i$ has two
2-fluxes, primary and thin, associated to it, denoted $F(c_i)$ and
$\hat{F}(c_i)$.
We let $F(c_i,t)$ and $\hat{F}(c_i,t)$ be their value at time $t$.
We assume that $F(c_i,t)$ and $\hat{F}(c_i,t)$ are calculated by using the rectilinear tube from $O$ to  $l_i$. Hence  $F(c_i,t)$ and $\hat{F}(c_i,t)$ are not defined for all $t$.

\smallskip

\mdef
We now calculate $A_1{(0,3)}=A_1{(1,3)}A_1{(0,1)}$
(recall \peq{prescomp} that the 2-dimensional holonomy of a
2-connection preserves the vertical composition of tubes).
Note that $A_1{(1,3)}$ is trivial
{(indeed  
  {$\TT_1^{1,3}$} is homotopic to the constant tube at $O$)}.
On the other hand,
given our definition of thin and primary 2-flux of a loop-particle, {we have that}
$A_1{(0,1)}=\hat{F}(c_2,0)^{-1}$, the
inverse of the thin 2-flux of $c_2$. {Cf. \peq{solveamb}, this is because $\TT_1^{0,1}$ is homotopic to the inverse of the tube $\hat{P}_{c_2}$ in \peq{another}, with which we calculate the thin 2-flux of $c_2$.}
Hence, the primary 2-flux of $c_1$ for $t \ge 1$ is
$$
F(c_1,t)=\hat{F}(c_2,0)^{-1}\t F(c_1,0).
$$
{Here $\t$ denotes the conjugation action of $T^2_R(\Gc)$ on itself.}
\smallskip

\mdef
Determining $A_2(0,3)$ is more complicated. Given our conventions,
$A_2(0,1)$ is trivial.
In order to calculate $A_2(1,3)$, we substitute \eqref{lp3} by an
isotopic movement of loops, {as in \eqref{lp2} below;} {cf. \peq{solveamb}}.
Between $t=1$ and $t=2$, the loop-particles $c_1$ and $c_2$,
exchange {`concentric position'},
with $c_2$ passing behind $c_1$, along the back part of the torus
$T^2_{c_1}$.
Hence the 2-dimensional holonomy of the trajectory of $c_2$ is related to the primary 2-flux of $c_1$ at time $t=1$. 

\vskip-0.5cm
\begin{minipage}[t]{0.85\textwidth}
\begin{equation*}
 {\centerline{\relabelbox 
\epsfysize 1.80cm 
\epsfbox{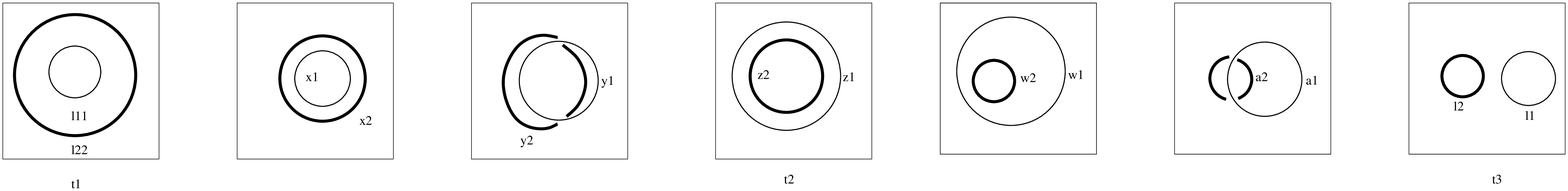}
\relabel{t1}{$\scriptstyle{t=1}$}
\relabel{t3}{$\scriptstyle{t=3}$}
\relabel{t2}{$\scriptstyle{t=2}$}
\relabel{l1}{$\scriptstyle{c_1}$}
\relabel{l11}{$\scriptstyle{c_1}$}
\relabel{l2}{$\scriptstyle{c_2}$}
\relabel{l22}{$\scriptstyle{c_2}$}
\relabel{a2}{$\scriptstyle{c_2}$}
\relabel{a1}{$\scriptstyle{c_1}$}
\relabel{a2}{$\scriptstyle{c_2}$}
\relabel{x1}{$\scriptstyle{c_1}$}
\relabel{x2}{$\scriptstyle{c_2}$}
\relabel{y1}{$\scriptstyle{c_1}$}
\relabel{y2}{$\scriptstyle{c_2}$}
\relabel{z1}{$\scriptstyle{c_1}$}
\relabel{z2}{$\scriptstyle{c_2}$}
\relabel{w1}{$\scriptstyle{c_1}$}
\relabel{w2}{$\scriptstyle{c_2}$}
\endrelabelbox}}
\end{equation*}
\end{minipage}
\begin{minipage}[t]{0.1\textwidth}
\quad \newline
\quad \newline
\begin{equation}\label{lp2}\end{equation}
\end{minipage}

\medskip

Now $A_2(1,3)=A_2(2,3)A_2(1,2)$.
{Putting together
 \peq{pa:ABconv} and \peq{solveamb}, we see that 
$A_2(1,2)$ is given by the inverse of the
primary 2-flux of $c_1$ at time $1$. This is because $\TT_2^{1 , 2} $ is homotopic to the vertical inverse of the torus-map $P_{c_1}$ at time $t=2$.}
Hence
$$
A_2(1,2)=F(c_1,1)^{-1} .
$$
Finally, 
$A_2(2,3)$ is
given by
the  thin 2-flux of $c_1$ at time 2, hence $A_2(2,3)=\hat{F}(c_1,2)=\hat{F}(c_1,1)$.

Let us  give  explicit formulae. {Recall \peq{another}, if}
\begin{align*}
&F(c_1,0)=(z,\d(R),e)=F_1
&&\textrm{and}  &&F(c_2,0)=(w,\d(R),f)=F_2,
\end{align*}
then
$\hat{F}(c_1,0)=(z,\d(R),R\,\, z \trr R^{-1})
 =\Theta_R(F_1)$ and $\hat{F}(c_2,0)=(w,\d(R),R\,\,w \trr R^{-1})=\Theta_R(F_2)$. Hence:
 \begin{align*}
   A_1(0,3)=\hat{F}(c_2,0)^{-1}=\Theta_R(F(c_2,0))^{-1}
=\Theta_R(F_2)^{-1}.\end{align*}
On the other hand:
\begin{align*}A_2(0,3)=A_2(2,3)\,A_2(1,2)&
= \hat{F}\big (c_1,2) \big)\,\, F\big (c_1,1 \big)^{-1}= \hat{F}\big (c_1,1 \big) \,\, F\big (c_1,1 \big)^{-1}\\
&=\Theta_R\big( {F}\big (c_1,1) \big) \,\, F\big (c_1,1 \big)^{-1}
\\&=
\Theta_R\big(\hat{F}(c_2,0)^{-1}\t F(c_1,0)\big) \,\,
\big(\hat{F}(c_2,0)^{-1}\t F(c_1,0)\big)^{-1}\\
&=    \Big(\Theta_R\big({\Theta_R(F_2}^{-1})\t F_1^{-1}\big)\Big)
\,\,\,
\Big( {\Theta_R(F_2)}^{-1}\t {F_1}\Big) .
\end{align*}


\mdef
Putting the stages together,
the move of loop-particles in \eqref{lp} and \eqref{lp3} leads to the
following transformation in the Hilbert space
$\C({\TRANS\big(T^2_R(\Gc)\big)})\tn \C({\TRANS\big(T^2_R(\Gc)\big))}$:
\beq
{\st{T_1}\otimes\st{T_2} \mapsto \st{A_2 T_2} \otimes \st{A_1 T_1}}.
\eq
{In full:}
%
{
\begin{equation}\label{2fmeta2}
\begin{split}
  (T_1^{-1}\t F_1 \ra{T_1}  F_1) \tn (T_2^{-1} \t F_2 \ra{T_2}  F_2)
  &\;\xmapsto{(-)\trl^*S_1^+[2]}\;
       A_2.(T_2^{-1}\t F_2 \ra{T_2}  F_2) \;\tn\;  A_1.(T_1^{-1}\t F_1 \ra{T_1}  F_1)\\
       &\quad \,= \; (T_2^{-1}\t F_2 \ra{A_2 T_2} A_2\t F_2)
       \;\tn\;  (T_1^{-1}\t F_1 \ra{A_1 T_1} A_1 \t F_1),
\end{split}
\end{equation}
}
where
 \begin{align*}
  &A_1=\big(\Theta_R(F_2)^{-1},1_E) &\textrm{ and }
  &&A_2=\Big(\Theta_R\big({\Theta_R(F_2}^{-1})\t F_1\big)\,\,\,
\big({\Theta_R(F_2)}^{-1}\t {F_1}^{-1}\big),1_E\Big).
\end{align*}


 
\begin{Remark}[Higher gauge flux metamorphosis]\label{2fluxmeta}
The unitary transformations $(-)\trl^*V_1[2]$ and
$(-)\trl^*S_1^+[2]$ in
$\C({\TRANS\big(T^2_R(\Gc)\big)})\tn
\C({\TRANS\big(T^2_R(\Gc)\big))}$,
defined in {Equations} \eqref{2fmeta1} and \eqref{2fmeta2},
are our proposal for a simplified higher gauge analogue of flux metamorphosis
of topological gauge theory \eqref{eq:rhorep}; see \cite{LL,Bais1,Bais2,Bais3,Bais4}.
\end{Remark}


It turns out (e.g. by direct calculation) that 
the  unitary transformations $(-)\trl^*V_1[2]$
and $(-)\trl^*S_1^+[2]$ yield a 
representation of the   $2$-circle loop braid group $\LBG_2$
in $\C({\TRANS\big(T^2_R(\Gc)\big)})\tn
\C({\TRANS\big(T^2_R(\Gc)\big))}$,
as the physical setting would suggest.
This can be extended to a representation of
$\LBG_n$ in $\C({\TRANS\big(T^2_R(\Gc)\big))}^{n \tn}$.
One way to {\em prove} this is to note that 
we have an underpinning
{W-\biker} in $\TRANS\big(T^2_R(\Gc)\big)$ {of the form:}  
\begin{equation}\label{compact}\hskip-3.5mm
\xymatrix@R=1pt{\\\\ \big(F_1,F_2\big)\stackrel{X^+_R}{\longmapsto}}\hskip-1cm
\xymatrix{& F_1  \ar[dr]|\hole
  \ar[dr]|<<<<<{A_1\,\,\bullet \quad\,\,\,}|\hole
  & F_2  \ar[dl]\ar[dl]|<<<<<{\quad \quad \,\, \bullet\, A_2\,\,\,\,\,\,\,\,}\\
&A_2 \t F_2 & A_1 \t F_1}
\hskip1mm\raisebox{-.31in}{$, \textrm{
  for }\begin{cases} A_1=(\Theta_R(F_2)^{-1},1_E)\\
  A_2=\Big(
  {
    \Theta_R\big({\Theta_R(F_2}^{-1})\t F_1\big)}
  \,\,\big(
{\Theta_R(F_2)}^{-1}\t {F_1}^{-1}\big), 1_E\Big)\end{cases}
$}\hskip-3mm\xymatrix@R=0pt{\\ \\.}
\end{equation}
And then we can use use the general Theorems {\ref{rep12} and  \ref{uf}.}
This sets the relevance of the work in this paper for modelling Aharonov-Bohm like phenomena for loop-particles in topological higher gauge theory. 

\smallskip {It is proven in \S\ref{refBIK} that
all groupoids ${\TRANS\big(T^2_R(\Gc)\big))}$, where $R\in E$, are isomorphic, and  all {W-bikoids} $X^+_R$ become the same by using a  canonical groupoid isomorphism.} Hence, as far as representations of $\LBG_n$ are concerned, it suffices to consider the case $R=1_E$. In this case  {$T^2_{1_E}(\Gc)\cong G \rtimes_{\trr} A$,} where $A=\ker(\partial)\subset E$; note that {$\ker \partial$} is  an abelian subgroup. The arrows of  {$\TRANS\big(S^2(\Gc)\big)\doteq\TRANS\big(T^2_{1_E}(\Gc)\big)$} take the form:
\begin{equation}\label{transhgt}
  (g,a) \ra{(w,k,m)} \Big (\d(m)wgw^{-1}, k \,+\, w \trr a \,-\,
  (wgw^{-1})\trr k\Big),
  \;\;\;\textrm{ where }g, w \in G,\,\,\ a,k\in \ker(\partial) \textrm{ and } m \in E.
\end{equation}
(Note that we switched to additive notation.) And  the pertinent {W-bikoid $X^+_{{gr^*}}$} takes the form:
\begin{equation}\label{addX2}
\hskip-3mm\xymatrix@R=0pt{\\ \big((z,a),(w,b)\big)\stackrel{X^+_{{gr^*}}}{\longmapsto}}\hskip-1cm
\xymatrix{& (z,a)  \ar[dr]|\hole
  \ar[dr]|<<<<<<<{\overline{w}\,\,\,\,\bullet\,\, \quad}|\hole
  & (w,b)  \ar[dl]\ar[dl]|<<<<<<<{\quad \quad \,\, \bullet -\overline{w}\trr a}\\
&(w,a+b-w^{-1} \trr a)  & \big(w^{-1}zw,w^{-1} \trr a\big)}
\raisebox{-.31in}{$, \textrm{
  for }\begin{cases} \overline{w}=(w^{-1},0_{\ker(\partial)},1_E)
  \\  -\overline{w}\trr a=    (1_G,-w^{-1} \trr a,1_E)
\end{cases}\hskip-3mm\xymatrix@R=0pt{\\ .}
$}
\end{equation}
{This $X^+_{{gr^*}}$ is a spin-off of the abelian $gr$-group W-bikoid of \eqref{xgr1}, but the underlying groupoid of $X^+_{{gr^*}}$ has an additional component, related to gauge transformations on 2-connections, which is however `decoupled' from the rest of the structure. The additional level of structure present in $X^+_{{gr^*}}$ nevertheless becomes very visible when considering 2-gauge transformations between gauge transformations (see \peq{2gaugetrans} below), and is essential when dealing with invariants of welded knots derived from $X^+_{{gr^*}}$, see \cite{prox}.}

\smallskip

\mdef \label{2gaugetrans} {For simplicity we  take $R=1_E$. As mentioned in \peq{2gaugetrans}, there is an important variant of the groupoid $\TRANS\big(T^2_{1_E}{(\Gc)}\big)\cong \TRANS\big(S^2(\Gc)\big)$, 
where   gauge transformations related by 2-gauge transformations \cite[\S 2.3.3]{schreiber_waldorf2}
\cite[\S 4.3.1]{martins_picken} are identified. This yields a quotient groupoid $\overline{ \TRANS\big(S^2(\Gc)\big)}$ of $\TRANS\big(S^2(\Gc)\big)$, where:}
\begin{multline*}
\big ((g,a) \ra{(w,k,m)} (\d(m)wgw^{-1}, k+ w \trr a - (wgw^{-1})\trr k)  \big)
\\ \cong \big ((g,a) \ra{(w\d(e),k,m \, (wg)\trr e \,w\trr e^{-1} )}
  (\d(m)wgw^{-1}, k + w \trr a - (wgw^{-1})\trr k)  \big).
\end{multline*}
{Here $g, w \in G$, $ a,k\in \ker(\partial)$  and $ m,e \in E.$  
The formulae in \eqref{addX2} give a bikoid structure in $\overline{ \TRANS\big(S^2(\Gc)\big)}$. Details on the construction of $\overline{ \TRANS\big(S^2(\Gc)\big)}$ and its relation to higher gauge theory will appear in \cite{prox}.}





\section{Preliminaries}\label{M-Preliminaries}

\mdef
If $\Cc$ is a category, the class of objects of $\Cc$ is denoted
by $\Obj(\Cc)$ or  by $\Cc_0$.
The class of morphisms of $\Cc$ is denoted by $\Cc_1$ or by
$\Mor(\Cc)$.
If  $x,y \in \Cc_0$,  the set of morphisms $x \to y$ is denoted  $\hom_{\cal C}(x,y)$ or $\hom(x,y)$.

\smallskip

\mdef\label{coc} Our default convention for compositions in categories and groupoids is reverse to that of function composition. Given objects $x,y,z$ of ${\cal C}$, composition is a map $(f,g) \in \hom(x,y)\times\hom(y,z) \mapsto f\star g \in \hom(x,z)$.  However the symbol $\circ$  will always denote the usual function composition. I.e., given set maps $f\colon X \to Y$ and $g \colon Y \to Z$, then $g\circ f$ will denote  the usual $(g\circ f)(x)=g(f(x))$.

\smallskip

\mdef We put $\N=\{0,1,2,\dots\}$ and $\Z^+=\{1,2,\dots\}$.

\smallskip

\mdef\label{sym} {The product we use in the symmetric group $\Sigma_n$,  of bijections $f\colon \{1,\dots, n\} \to \{1,\dots,n\}$, is $f.g=g \circ f$.}  


%
%
%
%
%

\subsection{General conventions for groupoids}

Let
$\Gamma=\Gam$
be a groupoid \cite{MacLane,Kassel,Higgins}.	
Here, $\Gamma_1$ is the set of morphisms (arrows) of $\Gamma$,
$\Gamma_0$ is the set of objects of $\Gamma$, and
$\sigma, \tau\colon \Gamma_1 \to \Gamma_0$
denote source and target maps, respectively. Given $x,y \in \G_0$, the set of morphisms $x \to y$ hence is $\hom(x,y)=\hom_{\Gamma}(x,y)=\{\g \in \G_1 \colon \sigma(\gamma)=x \textrm { and } \tau(\gamma)=y\}$.

\mdef An arrow $\gamma$ of $\Gamma$ will sometimes be denoted as
$(x \ra{\gamma} y)$, thus $x=\sigma(\gamma)$ and
$y=\tau(\gamma)$.

\smallskip
The composition map in $\Gamma$ provides, given any triple of objects $(x,y,z) \in (\Gamma_0)^3$, a map of sets: 
\begin{align*}&\hom(x,y) \times \hom(y,z) \to \hom(x,z)\\
&\left (\big (x \ra{\gamma} y \big), \big (y \ra{\phi} z\big) \right) \longmapsto \left (x \ra{\gamma} y \right)\star \left (y \ra{\phi} z \right).
\end{align*} 

\mdef \label{compnot}The composition of  arrows in $\G$ will be denoted in a variety of ways, as indicated below:
 $$\left(x \ra{\gamma} y \right)\star \left (y \ra{\phi} z \right)=\left(x \ra{\gamma} y   \ra{\phi} z \right)= 
\left (x \ra{\gamma \star \phi} z \right).$$

\mdef \label{invnot} The inverse arrow to $(x \ra{\gamma} y)$ can, and will, be denoted in four different ways:
$$
\big (x \ra{\gamma} y\big)^{-1} =\big ( y \ra{\gamma^{-1} }x \big)
 = \big ( y \ra{\overline{\gamma} }x \big)=\overline{(x \ra{\gamma} y)}.
$$
The identity map $\iota\colon \Gamma_0 \to \Gamma_1$
sends $x \in \Gamma_0$ to the arrow
$(x \ra{\id_x} x)$.

\begin{Definition}[Action groupoid] \label{de:ag}
Let the group $G$ have a left-action ``$.$'' on set $X$.
The {\em action groupoid}
$\AGr{X}{G}$ 
has $X$ as set of objects.
The arrows have the form: $(x \ra{g} g.x)$, where $x \in X$ and $g\in G$.
The composition is such that:
$(x\ra{g} g.x)\star (g.x \ra{h} (hg).x)=(x \ra{hg} (hg).x)$.
(Note the order {of} the product.)
Of course, the identity map is such that $\iota(x)=(x \ra{1_G} x)\doteq \id_x$, where $1_G$ is the unit element of $G$.
\end{Definition}

\begin{Example}[$\AUT(G)$] \label{autg}
  {Let $G$ be a group. The action groupoid of the conjugation action of $G$ on $G$ is denoted by $\AUT(G)$. Hence the set of objects of $\AUT(G)$ is $G$, and  the arrows have the form $x \ra{g} gxg^{-1}$, where $g,x \in G$.
  The composition in $\AUT(G)$ {therefore} is such that:
  $$
  (x \ra{g} gxg^{-1}) \star (gxg^{-1} \ra{h} hgxg^{-1} h^{-1})
  =(x\ra{hg}  hgxg^{-1} h^{-1}).
  $$}
\end{Example}

\subsection{Groupoid algebras and their representations}\label{ga}
\begin{Definition}[Groupoid algebra]\label{groupoid_algebra}
Let $\Gamma$ be a groupoid. {Cf. \cite{ga1,willerton2008twisted,Morton}.
The groupoid algebra $\C(\Gamma)$ of $\Gamma$}
is the free $\C$ vector space $\C \Gamma_1$
over the set $\Gamma_1$ of {morphisms of $\Gamma$}. The product on generators is:
\begin{equation}\label{pga}
\big (x \ra{\gamma} y\big)\star \big (x' \ra{\gamma'} y'\big)=
\big (x \ra{\gamma} y\big) \big (x' \ra{\gamma'} y'\big)
=\delta(y,x') \big (x \ra{\gamma \star \gamma'} y' \big ), \textrm{ where } (x \ra{\gamma} y\big), \big (x' \ra{\gamma'} y'\big)\in \Gamma_1.
\end{equation}
(Here $\delta(y,x)=1$, if $y=x$, and $0$ otherwise).
If $\Gamma_0$ is finite,  then $\C(\Gamma)$ is a unital algebra with identity:
$$\displaystyle 1_{\C(\Gamma)}=\sum_{x \in \Gamma_0} \iota(x)=\sum_{x \in \Gamma_0} \big( x \ra{\id_x} x\big) .$$ 
\end{Definition}

\mdef \label{star} The groupoid algebra has a {*-structure} given by \smash{$(a \ra{\gamma} b)^*=(b \ra{\gamma^{-1}} a)$,} for all  \smash{$(a \ra{\gamma} b) \in {\G_1}$.}

\begin{Remark}\label{qdis} If $G$ a finite group, then
$\C(\AUT(G))$ (see {Example}  \ref{autg})
is  \cite{willerton2008twisted} the underlying algebra  of the quasi-triangular  Hopf algebra  $D(G)$ in  \eqref{eq:hopf},  the quantum double of the group algebra of $G$; see 
\cite{Alt,Gould,Kassel,Turaev}.  
\end{Remark}

\mdef In this paper a representation (or right-representation) of a unital algebra $A$ will mean a right $A$- module $V$, with action denoted $(v,a) \in V \times A \mapsto v.a \in A$, such that $v. 1_A=v$, for each $v \in V$. 

\begin{Definition}\label{unirep} {Let $(V,\langle,\rangle)$ be an inner product space. A representation $\rho\colon (v,k) \in V \times \C(\G) \mapsto x.k \in V $ of $\C(\Gamma)$ is called unitary if it is unitary with respect to the $*$-structure in \peq{star}. This means that:} $${\Big \langle u.(x \ra{\gamma} y), v \Big \rangle=\Big \langle u,v.(y \ra{\gamma^{-1}} x) \Big  \rangle \textrm{, for all $u,v \in V$, and all $(x \ra{\gamma} y)\in \Gamma_1$}.}$$
\end{Definition}
\mdef{ Suppose that $\Gamma$ is finite and $\rho$ is a finite dimensional representation of $\C(\G)$ on $V$. Standard techniques, as e.g. in   \cite[page 12]{Gould}, prove that we can find an inner product in $V$ with respect to which $\rho$ is unitary. Standard tecniques, as in \cite{Gould}, also prove that $\C(\Gamma)$ is semisimple. This will be addressed elsewhere.}

\begin{Example}\label{rre}
The right regular representation of $\C(\Gamma)$ has
$\C(\Gamma)$ as underlying vector space. The action is by right multiplication. This representation is unitary if we
pick the inner product in $\C(\Gamma)$ that renders different arrows
orthonormal. 
\end{Example}

\begin{Example}[Object regular representation]\label{obre} The object regular representation of $\C(\Gamma)$ has the free vector space $\C \Gamma_0$ on $\Gamma_0$ as underlying vector space. The action is such that if $x \in\Gamma_0$ and $(a \ra{\gamma} b)\in {\G_1}$, then:
 $$x . (a \ra{\gamma} b)=\delta(a,x)b. $$
This representation is unitary if we choose $\Gamma_0$ to be an orthonormal basis of $\C \Gamma_0$.

{If $x \in \Gamma_0$, then $\C\pi_0(\Gamma,x)$ is a sub-representation of $\C \Gamma_0$. Here  $\C\pi_0(\Gamma,x)$ is the connected component of $\Gamma_0$ to each $x$ belongs. As in the case of the quantum-double of a finite group algebra \cite[Chapter IV]{Gould}, the latter representation can be modified in order to include a representation of the automorphism group of $x$.}
\end{Example}

\subsection{The virtual braid group and the  welded braid group}
\label{ss:vbg}

Virtual knot theory is discussed e.g. in \cite{kamada,kauffman,manturov}. In this paper all manifolds, {homeomorphisms}, isotopies, immersions, embeddings, etc, are assumed to be piecewise linear \cite{BZ,Kw}, unless otherwise specified. 

{Let $N$ be an oriented 1-dimensional manifold, possibly with boundary. An immersion $Q$ of $N$ in $\R\times [0,1]$ is the image $Q=\phi(N)$ of an immersion map $\phi\colon N \to \R \times [0,1]$; {i.e. $\phi$ is piecewise linear and locally injective.}} The set of multiple points $M(Q)$ of $Q$ is $M(Q)=\{z \in Q\colon \#\phi^{-1}(z)\ge 2\}$, where $\#$ denotes set cardinality. A double-point (also called a double intersection) is a multiple point such that $\#\phi^{-1}(z)=2$.


{The interval $[0,1]$ is  oriented in the positive direction. We put $p\colon \R \times [0,1] \to [0,1]$ to be  $p(s,t)=t$.}

\begin{Definition}[Virtual braid diagram]
A virtual braid diagram, of degree $n$, is the image $Q$ of an immersion  {map $\phi\colon \sqcup_{i=1}^n I_i  \to \R \times [0,1]$, where $I_i=[0,1]$.
We put $\phi_i\colon  I_i \to \R\times [0,1]$ to denote the restriction of $\phi$ to $I_i$. The component $Q_i\subset Q$ of the immersion $Q$ is by definition given by  $\phi(I_i)$, where $i=1,\dots,n$.}

We impose that:
{\bf (1)} multiple points {of $Q$} are all transversal \cite[\S 5.2]{RS} double intersections between different components of the immersion,
{\bf (2)} each double point of $Q$ is assigned  a label from `classical crossings'
(positive and negative) and `virtual crossings', as below: 
$$
\xymatrix@R=1pt{\\ \\\textrm{positive crossing}}\hskip-1cm
\xymatrix{ &\ar[dr]|\hole &\ar[dl] \\ && }\xymatrix@R=1pt{\\\\,} \qquad \qquad
\xymatrix@R=1pt{\\ \\\textrm{negative crossing}} \hskip-1cm
\xymatrix{ &\ar[dr] &\ar[dl]|\hole \\ && }\xymatrix@R=1pt{\\\\,}\qquad \qquad
\xymatrix@R=1pt{\\ \\\textrm{virtual crossing}}\hskip-1cm
\xymatrix{ &\ar[dr] &\ar[dl] \\ && }\xymatrix@R=1pt{\\\\;}
$$
{\bf (3)} for each $i\in \{1,\dots,n\}$, the map $p\circ \phi_i \colon [0,1] \to [0,1]$ is an orientation reversion homeomorphism; {\bf (4)}  $Q \cap (\R \times \{1\})=\{1,\dots,n\}\times \{1\}$ and  $Q \cap (\R \times \{0\})=\{1,\dots,n\}\times \{0\}$, {\bf (5)} for each $i\in \{1,\dots, n\}$, {$\phi_i(1)=(i,0)$}; {\bf (6)} the restriction of  $p\colon \R \times [0,1] \to [0,1]$ to the set of double points is injective (i.e. double points occur at different heights).

We consider two virtual virtual braid diagrams $Q$ and $Q'$ equivalent if they can be deformed one into the other via an ambient isotopy $t\mapsto f_t$ of $\R \times [0,1]$, fixing the boundary, and such that $f_t(Q)$ is a virtual braid diagram for each $t \in [0,1]$.
\end{Definition}

\mdef\label{diagtoperm} A virtual braid diagram $Q$ {yields} a bijection $U_Q\colon \{1,\dots,n\} \to \{1,\dots, n\}$. Our convention is that $U_Q$ sends $j$ to the unique $i$ such that $\phi_i(0)=j\times \{1\}$. (Note that $\phi_i(1)=(i,0)$.)  {See also diagram} \eqref{diagtoPERM}, below.

\begin{Definition}[The monoid ${\rm M V}{{[n]}}$ of virtual braid diagrams]Given a positive integer $n$, we have a monoid ${\rm M V}{{[n]}}$ of  virtual braid diagrams in degree $n$. Given $M$ and $M'$ in ${\rm M V}{{[n]}}$,  the multiplication $MM'$ is the vertical juxtaposition of $M$ and $M'$, where $M$ stays on top of $M'$, {followed by rescalling in the height direction.} The identity of the monoid ${\rm M V}{{[n]}}$ is given by the equivalence class of the diagram $I{{[n]}}$, below: $$\xymatrix@R=1pt{\\ I{{[n]}} =}\In\xymatrix@R=1pt{\\\\\,\,.}$$
(Here, and also below, we put $(i,1)=i$, where $i=1,2,\dots,n$.)
\end{Definition}

The monoid  ${\rm M V}{{[n]}}$ is freely generated \cite{kamada} by the equivalence classes of the diagrams displayed below:
\begin{equation}\label{mongens}
\hskip-0.5cm{\xymatrix@R=1pt{\\S^+_a{{[n]}}=}}
\Spic
\xymatrix@R=1pt{ \\\textrm{, } \quad S^-_a{{[n]}}=} 
\Smpic
\xymatrix@R=1pt{ \\ \textrm{,} \quad V_a{{[n]}}=}
\Vpic\xymatrix@R=1pt{\\\\.}
\end{equation}
Here $a\in \{1,2,\dots,n-1\}$.
From now on, we will frequently not display some of the vertical strands which do not interact with the rest of the diagram, when drawing compositions of generators of the monoid  ${\rm M V}{{[n]}}$. Also note that given our convention for the multiplication in  ${\rm M V}{{[n]}}$, diagrams are read from top to bottom.

\begin{Definition} \label{de:WBG}
The {$n$-strand} welded braid group
$\WBG_n$  {\cite{kamada,KLam,fenn_et_al}}
is the monoid (easily proven to be a group) obtained from the monoid ${\rm M V}{{[n]}}$ of virtual braid diagrams, {by imposing the relations \eqref{eq:SRI}--\eqref{eq:F}, below.   Here $a,b\in \{1,\dots, n-1\}$ for \eqref{eq:SRI} and \eqref{eq:Loc}, and $a \in \{1,\dots,n-2\}$ for \eqref{eq:R3}---\eqref{eq:F}.} {Some relations are written algebraically and also diagramatically, since there will be several diagrammatic calculations later on.}
\begin{multline} \label{eq:SRI}
\mbox{\bf Classical and virtual Reidemeister II moves (RII$^{\pm}$ and VII):}\\
V_a{{[n]}}\,\, V_a{{[n]}}\stackrel{{\rm VII}}{=} {I[n]}, \qquad S^+_a{{[n]}}\,\, S^-_a{{[n]}}\stackrel{{\rm RII^+}}{=}{I[n]} \quad \textrm{ and }\qquad S^-_a{{[n]}} \,\, S^+_a{{[n]}}\stackrel{{\rm RII^-}}={I[n]}. 
\end{multline}
\begin{multline}  \label{eq:Loc} 
{\mbox{\bf Locality,  if $|a-b|\ge 2$ then:}}\\
\begin{split}
&V_a{{[n]}}\,\, V_b{{[n]}}=V_b{{[n]}}\,\, V_a{{[n]}},  &&V_a{{[n]}}\,\, S_b^+{{[n]}}=S_b^+{{[n]}}\,\, V_a{{[n]}},
 &V_a{{[n]}}\,\, S_b^-{{[n]}}=S_b^-{{[n]}}\,\, V_a{{[n]}},  
 \\ &S_a^+{{[n]}}\,\, S_b^-{{[n]}}=S_b^-{{[n]}}\,\, S_a^+{{[n]}}, && S_a^-{{[n]}}\,\, S_b^-{{[n]}}=S_b^-{{[n]}}\,\, S_a^-{{[n]}}, & S_a^+{{[n]}}\,\, S_b^+{{[n]}}=S_b^+{{[n]}}\,\, S_a^+{{[n]}}.
 \end{split}
\end{multline}

\begin{equation}\label{eq:R3}
\mbox{\bf Reidemeister III move (RIII):}
\;\;\;\;\;\;    {S_a^+}{{[n]}}\,\, {S_{(a+1)}^+}{{[n]}}\,\, {S_a^+}{{[n]}}\stackrel{\rm RIII}{=}{S_{(a+1)}^+}{{[n]}}\,\, {S_a^+}{{[n]}}\,\, {S_{(a+1)}^+}{{[n]}}.    
\end{equation}
$$\xymatrix{\\ &\qquad \textrm{{Graphically:}} \quad}
\xymatrix@R=15pt@C=15pt{ &a\ar[dr]|\hole & a+1\ar[dl] & a+2\ar[d]  \\
  &a\ar[d] & a+1\ar[dr]|\hole & a+2\ar[dl] \\
  &a\ar[dr]|\hole & a+1\ar[dl] & a+2\ar[d] \\
  &a & a+1 & a+2}
\xymatrix@R=1pt{ \\ \\ \\ \\ \quad \stackrel{\rm RIII}{=}}
\xymatrix@R=15pt@C=15pt{ &a\ar[d] & a+1\ar[dr]|\hole & a+2\ar[dl]  \\
  &a\ar[dr]|\hole & a+1\ar[dl] & a+2\ar[d] \\
  &a\ar[d] & a+1\ar[dr]|\hole & a+2\ar[dl] \\
  &a & a+1 & a+2}\xymatrix@R=1pt{\\\\\\\\\\\,\,\,\,\,.}
$$


\begin{equation}\label{VIII}
\mbox{\bf Virtual Reidemeister III  move (VIII):} \quad
\;\;\;\;\;\;    \hspace{-0.2in}
  V_a{{[n]}}\,\, V_{(a+1)}{{[n]}}\,\, V_a{{[n]}}\stackrel{\rm VIII}{=}V_{(a+1)}{{[n]}}\,\, V_a{{[n]}}\,\, V_{(a+1)}{{[n]}}     .
\end{equation}
$$\xymatrix{\\ &\qquad \textrm{{Graphically:}} \quad}\xymatrix@R=15pt@C=15pt{ &a\ar[dr]& a+1\ar[dl] & a+2\ar[d]  \\
                                                                     &a\ar[d] & a+1\ar[dr] & a+2\ar[dl] \\
                                                                     &a\ar[dr] & a+1\ar[dl] & a+2\ar[d] \\
                                                                     &a & a+1 & a+2}
                                                                     \xymatrix@R=1pt{ \\ \\ \\ \quad \stackrel{\rm VIII}{=}}
                                                                     \xymatrix@R=15pt@C=15pt{ &a\ar[d] & a+1\ar[dr] & a+2\ar[dl]  \\
                                                                     &a\ar[dr] & a+1\ar[dl] & a+2\ar[d] \\
                                                                     &a\ar[d] & a+1\ar[dr] & a+2\ar[dl] \\
                                                                     &a & a+1 & a+2}\xymatrix@R=1pt{\\\\\\\\\\\,\,\,.} $$

\begin{equation} \label{eq:VRIII}
\mbox{\bf Mixed Reidemeister III move (MIII):} \quad
\;\;\;    V_a{{[n]}}\,\, V_{(a+1)}{{[n]}}\,\, {S_a^+}{{[n]}}\stackrel{\rm MIII}{=}{S_{(a+1)}^+}{{[n]}}\,\, V_a{{[n]}}\,\, V_{(a+1)}{{[n]}} .   
\end{equation}
$$ \xymatrix{\\ &\qquad \textrm{{Graphically:}}} 
    \xymatrix@R=15pt@C=15pt
{ &a\ar[dr] & a+1\ar[dl] & a+2\ar[d]  \\
  &a\ar[d] & a+1\ar[dr] & a+2\ar[dl] \\
  &a\ar[dr]|\hole & a+1\ar[dl] & a+2\ar[d] \\
  &a & a+1 & a+2}
\xymatrix@R=1pt{ \\ \\ \\ \quad \stackrel{\rm MIII}{=}}
\xymatrix@R=15pt@C=15pt
{ &a\ar[d] & a+1\ar[dr]|\hole & a+2\ar[dl]  \\
   &a\ar[dr] & a+1\ar[dl] & a+2\ar[d] \\
   &a\ar[d] & a+1\ar[dr] & a+2\ar[dl] \\
   &a & a+1 & a+2}\xymatrix@R=1pt{\\\\\\\\\\\,\,\,.}
$$

%


\begin{equation}\label{eq:F} \hspace{-.13in}
 \mbox{\bf {Welded} Reidemeister III move (WIII):} \quad   V_a{{[n]}}\,\, {S_{(a+1)}^+}{{[n]}}\,\, {S_a^+}{{[n]}}\stackrel{\rm WIII}{=}{S_{(a+1)}^+}{{[n]}}\,\, {S_a^+} {{[n]}}\,\, V_{(a+1)}{{[n]}}.
\end{equation}
 $$\xymatrix{\\ &\qquad \textrm{{Graphically:}}}\xymatrix@R=15pt@C=15pt{ &a\ar[dr] & a+1\ar[dl] & a+2\ar[d]  \\
    &a\ar[d] & a+1\ar[dr]|\hole & a+2\ar[dl] \\
    &a\ar[dr]|\hole & a+1\ar[dl] & a+2\ar[d] \\
    &a & a+1 & a+2}
  \xymatrix@R=1pt{ \\ \\ \\ \quad \stackrel{\rm WIII}{=}}
  \xymatrix@R=15pt@C=15pt{ &a\ar[d] & a+1\ar[dr]|\hole & a+2\ar[dl]  \\
    &a\ar[dr]|\hole & a+1\ar[dl] & a+2\ar[d] \\
    &a\ar[d] & a+1\ar[dr] & a+2\ar[dl] \\
    &a & a+1 & a+2}\xymatrix@R=1pt{\\\\\\\\\\\,\,\,.}
$$
\end{Definition}
\mdef\label{twoperpectives} It is easy to see that $\WBG_n$ is isomorphic to the group formally generated by the symbols $S^+_a{{[n]}}$ and $V_a{{[n]}}$, where $a=1,\dots, n-1$, with relations \eqref{eq:SRI} to \eqref{eq:F}, desconsidering all relations involving $S^-_a{{[n]}}$, $a=1,\dots, n-1$. This is the point of view taken in
\cite{paul_et_al,baez_et_al,damiani,fenn_et_al}. We will {use the   two perspectives} (presentation of $\WBG_n$ as a group, versus presentation of $\WBG_n$ as a monoid) in this paper.

\begin{Definition}[Virtual braid group] The group with the same generators
  as $\WBG_n$,
  and the same relations, except for the {welded} Reidemeister III (WIII) move \eqref{eq:F},
  is called the virtual braid group $\VBG_n$. 
\end{Definition}
 \begin{Definition}[Braid group] The group with generators $S^\pm_a{{[n]}}$, $a=1,\dots, n-1$, and all relations in $\WBG_n$ not involving the {$V_n[a]$,}  $a=1,\dots, n-1$, is called the braid group $\BG_n$ in $n$-strands. 
\end{Definition}
\noindent Clearly, inclusion of generators provides an epimorphism $\VBG_n\to \WBG_n$.

\smallskip
\mdef\label{deftab} {Consider $\Sigma_n$  in \peq{sym},} with the usual presentation with generators $V_{a}{{[n]}}$, where $a=1,\dots,n-1$, and relations: (i) $V_a{{[n]}}\, V_{(a+1)}{{[n]}}\, V_a{{[n]}}\stackrel{}{=}V_{(a+1)}{{[n]}}\, V_a{{[n]}}\, V_{(a+1)}{{[n]}}$; (ii)
$V_a{{[n]}}\ V_b{{[n]}}=V_b{{[n]}}\, V_a{{[n]}}$ if $|a-b|\ge 2$; and (iii) $V_a{{[n]}}^2=\id$. And then, concretely, $V_a(n)=t_{a,a+1}^n\in \Sigma_n$, the transposition  exchanging $a$ and $a+1$.
\smallskip

\mdef\label{proj} {We have epimorphisms $U\colon \WBG_n,\VBG_n \to \Sigma_n$ sending both $S_a^\pm{{[n]}}$ and $V_a{{[n]}}$ to $V_a{{[n]}}$. Given $B$ in  $\WBG_n$ or $\VBG_n$, then $U(B)$, denoted $U_B$, is called the underlying permutation of $B$.  Looking at elements of  $\WBG_n$ and $\VBG_n$ as equivalence classes $[Q]$ of virtual braid diagrams $Q$, it holds that $U_{[Q]}=U_Q$ in \peq{diagtoperm}.}

 \begin{Remark}[Tournant Dangereux]\label{tornant}
%
  {The {welded} Reidemeister
 III relation \eqref{eq:F} is not equivalent to:}
\begin{equation}\label{rev}
{S_a^+}{{[n]}}\,
{S_{(a+1)}^+}{{[n]}}\,
 V_a{{[n]}}
=V_{(a+1)}{{[n]}}\, {S_a^+} {{[n]}} \, {S_{(a+1)}^+}{{[n]}}.
\end{equation}
{(We have just reversed the order of the factors in \eqref{eq:F}.) The relation in \eqref{rev} is equivalent to the ``Forbidden move $F_2$ in \cite{KLam}''.}
{The group obtained from  $\WBG_n$ by imposing \eqref{rev}
is called the ``unrestricted virtual braid group'' \cite{KLam,BBD}. We note that virtual knot theory in the presence of the {welded} Reidemeister III move and its reverse \eqref{rev} is trivial; see e.g. \cite{KT,Nelson:2001}; however non-trivial linking phenomena  may still occur \cite{NF}.} 
\end{Remark}

\begin{Definition}[{The monoidal categories ${\WB}$ and ${\VB}$}]\label{wbcat} 
We have strict monoidal categories {${\WB}$ and ${\VB}$}, the welded braid category and the virtual braid category. Objects in  ${\WB}$ and ${\VB}$ are given by non-negative integers. The sets of morphism $m\to m'$ are non-empty if, and only if, $m=m'$. Also put $\hom_{{\WB}}(m,m)=\WBG_m$ and $\hom_{{\VB}}(m,m)=\VBG_m$. On objects, the tensor product  is  $n\otimes n'=n+n'$. Given morphisms $B\colon n \to n$ and $B'\colon n'\to n'$, in ${\WB}$ or ${\VB}$, their tensor product $(B\otimes B')\colon (n+n') \to (n+n')$  is derived from  the horizontal juxtaposition $Q_B Q_B'$ of virtual braid diagrams $Q_B$ and $Q_{B'}$, for $B$ and for $B'$, moving crossings up and down, if necessary, in order that all crossings appear at different heights.

\noindent By \eqref{eq:Loc}, the tensor product of morphisms is well defined and we  have  monoidal categories ${\WB}$ and ${\VB}$.
\end{Definition}


\mdef \label{catpres} {The monoidal categories  {${\WB}$ and ${\VB}$}  can be presented by generators and relations, as in \cite[\S XII]{Kassel}. The generators are $S^\pm[2]$ and  $V[2]$. For ${\VB}$ the only relations are the ones in \eqref{eq:SRI}, for $n=2$, and \eqref{eq:R3}--\eqref{eq:VRIII}, for $n=3$, where $\SSS{2}{3} = 1_1 \otimes \SSS{1}{2}$, $\SSS{1}{3} =  \SSS{1}{2} \otimes 1_1$ and so on. (We note that the locality relation   \eqref{eq:Loc} is automatically satisfied in a monoidal category.) In order to present ${\WB}$, we must also add relation \eqref{eq:F}.}

\subsection{{The loop braid group: 
    isomorphism with the welded braid group}} \label{motion}

Papers addressing the loop braid group, {or close relatives}, include
\cite{baez_et_al,damiani,GS,brendle_hatcher}.
We essentially follow  \cite{damiani}.

{In this subsection, our convention for the 3-disk is: $ D^3=\{z\in \R^2\colon ||z||\leq 2\} \times [-1,1]$.}
Let $n\in \Z^+$. Let $C_n \subset D^3  $ 
be a disjoint union of unlinked circles in  {$D^3\cap (\R^2\times\{0\})$,}
oriented counterclockwise.
For definiteness put
$S^1_j =(\{z \in {\R^2}\colon||z||=1/(3n)\}+j/n)\times \{0\}$
and
$C_n=\bigcup_{j=1}^n S^1_j$, {which is therefore oriented.} 

\mdef
{Let $\Homeo(D^3,C_n)$ be the group  of orientation preserving
homeomorphisms $D^3 \to D^3$, which restrict to an orientation preserving homeomorphism $C_n  \to C_n$,
and which are the
identity over $\d D^3$. The group operation in $\Homeo(D^3,C_n)$ is composition.
Two homeomorphisms $f,g\colon \Homeo(D^3,C_n)$ are said to be
{pair-isotopic}, if there exists a map
$H\colon D^3 \times [0,1] \to D^3$ such that
$(z,t) \in D^3 \mapsto H(z,t,s)\in  D^3$
is in
$\Homeo(D^3,C_n)$, for each $s\in [0,1]$.} 
%
%
%

\begin{Definition}
{Let $\MCG(D^3,C_n)$ (the ``mapping class group of $(D^3,C_n)$'')
be the group of homeomorphisms $f\in \Homeo(D^3,C_n)$,
considered up to pair isotopy. The group law in $\MCG(D^3,C_n)$ is
induced by the composition in $\Homeo(D^3,C_n)$, which descends to the
quotient under pair isotopy. In this paper, $\MCG(D^3,C_n)$ is also
denoted $\LBG_n$ and called the ``loop braid group {(in $n$ circles)}''.}  
\end{Definition}
As we shall illustrate shortly,
 $\MCG(D^3,C_n)$ can be thought of as a group of
braidings of loop world-sheets in {(3+1)-dimensions}, generalising the ordinary braid group
regarded as a group of braidings of point world-lines in {(2+1)-dimensions}; {see \cite{baez_et_al,damiani}}.
Hence the term `loop braid group' and the  notation $\LBG_n$.

For some manipulations, and for example representation theoretically,
the given realisation of $\LBG_n$ is difficult to work with. 
Next we prepare 
a useful presentation.

\begin{Theorem} [{Baez et al  \cite{baez_et_al}, Damiani \cite{damiani}}]
\label{main0}
There exists an isomorphism:
$$
T\colon B \in  \WBG_n \longmapsto T_B \in  \MCG(D^3,C_n) =\LBG_n.
$$
{This isomorphism sends {a group generator \peq{twoperpectives}} $g\in \WBG_n$ of the form  $S_a^+{{[n]}}$ or $V_a{{[n]}}$,  
to the pair-isotopy class of the  homeomorphism
$f^g=\phi_{t=1}^g\colon (D^3,C_n) \to (D^3,C_n)$ at the end of the isotopies $\Phi^g=\big(t\in [0,1]\mapsto \phi_t^g\in \Homeo(D^3,\emptyset)\big)$,  
of $D^3$, relative to {its} boundary, indicated in the Fig. \ref{ISotopy}, below
(only a few virtual braid strands and circles are displayed, but the general picture should be clear).  
}
$$\centerline{\relabelbox 
\epsfysize 6.25cm 
\epsfbox{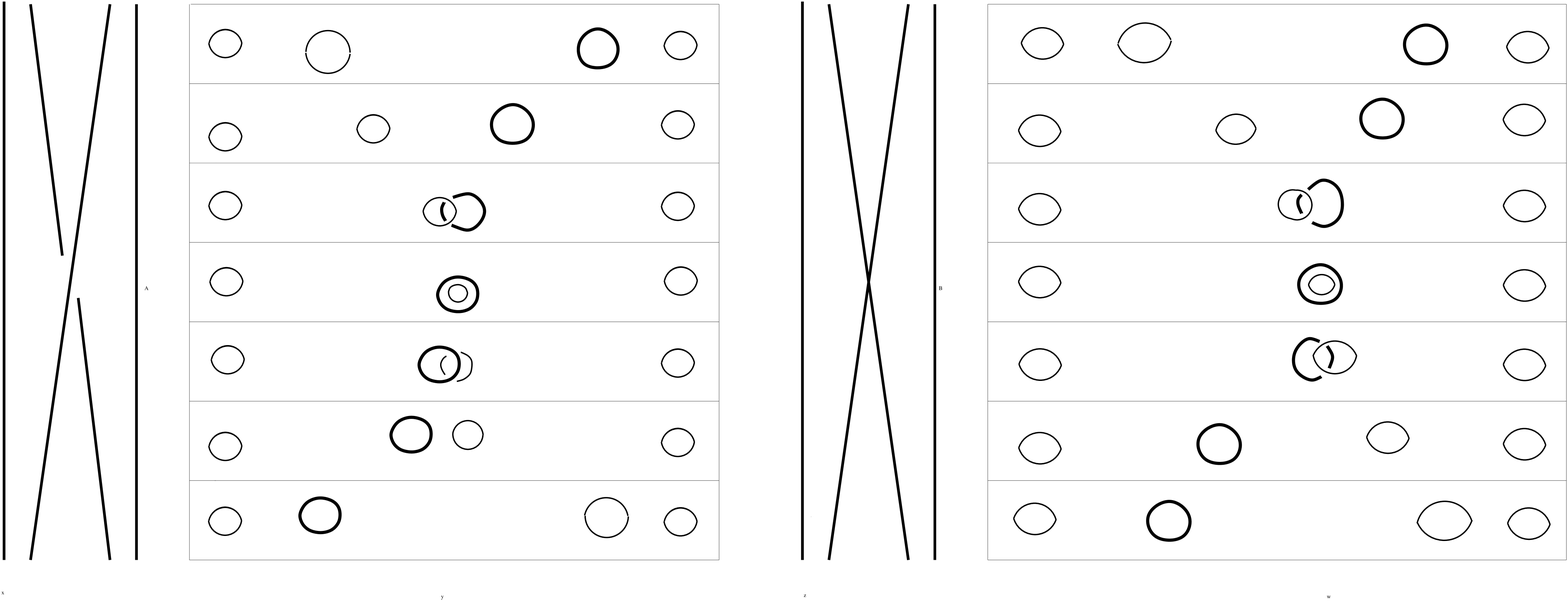}
\relabel{x}{$g=S^+_a{{[n]}}$}
\relabel{y}{$\Phi^g$}
\relabel{z}{$g=V_a{{[n]}}$}
\relabel{w}{$\Phi^g$}
\endrelabelbox}$$
\vskip-1cm
\begin{figure}[H]
 {\caption{{{A generator $g \in \WBG_{n}$  gives rise to an isotopy $\Phi^{g}=\big(t\in [0,1]\mapsto \phi_{\,t}^{\,g}\in \Homeo(D^3,\emptyset)\big)$ of $D^3$.}\label{ISotopy}}}}
 \end{figure} 
\end{Theorem}
\noindent {We note that the convention for loop braiding in Fig. \ref{ISotopy} is opposite to that in Equations \eqref{lp}--\eqref{lp3} in \S\ref{physmot}.}

\medskip

{This theorem was stated in \cite{baez_et_al} for smooth motion groups. The interpretation of smooth motion groups in  terms of topological mapping class groups appears in \cite{damiani}. In order to connect the two points of view, we can use results in \cite{GS}, to pass from mapping class groups to motion groups,  and \cite{Wil}, to go from the continuous setting to the  smooth setting.}
{In \cite{brendle_hatcher} it is proven that $\WBG_n$ has another topological realisation as the fundamental group of the configuration space of rings in $D^3$.}

{Several other groups are isomorphic to $\WBG_n$, e.g. the group of conjugating automorphisms of a free group \cite{Sav} -- {we will go back to this in \S\ref{bandh}.} For further discussion see \cite{damiani}.}

%
\section{Bikoids and welded bikoids (W-\biker s)}\label{M-bk}

\subsection{The wreath product groupoid $\gwrr{n}=\gwr{n}$ of a groupoid $\Gamma$} \label{thegr}


Let $\Gamma = \Gam$ be a groupoid.
For $n \in \Z^+$,
we write $(x_1,\dots,x_n)=\underline{x} \in \Gamma_0^n$ and
$(\gamma_1,\dots,\gamma_n)=\underline{\gamma} \in \Gamma_1^n$.
The {\em product groupoid }
$\Gamma^n$  is  $\Gamma^n=(\Gamma_1^n , \Gamma_0^n , \underline{\sigma} ,
 \underline{\tau} ,  \underline{\iota}, \underline{\star} )$,
with  product composition.
Morphisms in $\Gamma^n$ can be represented in a variety of ways, as indicated below:
$$\big(\underline{\gamma}\colon \underline{x} \to \underline{y}\big)=
\big(\underline{x} \ra{\underline{\gamma}} \underline{y} \big)
\;\;=\;\; \big( (x_1,\dots, x_n) \ra{(\gamma_1,\dots, \gamma_n)} (y_1,\dots, y_n)\big)
\;\;=\;\; \big ( (x_1 \ra{\gamma_1} y_1) \tn \dots \tn (x_n \ra{\gamma_n} y_n)\big).
$$

\mdef \label{pa:gf} Cf. \peq{sym} for our convention on the symmetric group $\Sigma_n$.
We have a left-action of  $\Sigma_n$ on $\Gamma^n$
by groupoid functors, obtained by permuting components.
Namely if
$f\in \Sigma_n$ and $({\underline{\gamma}}\colon\underline{x} \to \underline{z}) \in \Gamma_1^n$,
we put:
\begin{multline*}
  f \trr \big ( (x_1,x_2,\dots x_n)
  \ra{(\gamma_1,\gamma_2,\dots, \gamma_n)}(z_1,z_2,\dots z_n)\big) =\big ( f \trr (x_1,x_2,\dots x_n)
  \ra{ f \trr (\gamma_1,\gamma_2,\dots, \gamma_n)} f \trr (z_1,z_2,\dots z_n) \big)\\
  \doteq\big ( (x_{f(1)},x_{f(2)},\dots x_{f(n)})
  \ra{(\gamma_{f(1)},\gamma_{f(2)},\dots, \gamma_{f(n)})}(z_{f(1)},z_{f(2)},\dots z_{f(n)})\big).
\end{multline*}
\begin{Definition}
The `wreath product groupoid' 
$\gwrr{n} \; {=} \; \Gamma^n \rtimes \Sigma_n$,
with respect to the action in \peq{pa:gf}, is defined as follows. The set of objects is  $\Gamma_0^n$ and the set of morphisms is  $\Gamma_1^n \times \Sigma_n$.
Arrows have the form:
$$
\Big (\big(\sigma{(\gamma_1)},\dots, \sigma{(\gamma_n)}\big)
\ra{\big((\gamma_1, \dots, \gamma_n),f\big)}
\big(\tau{(\gamma_{f^{-1}(1)}),\dots, \tau(\gamma_{f^{-1}(n)})}\big)\Big )
=\big(\underline{\sigma}(\underline{\gamma})
\ra{(\underline{\gamma},f)}
 f^{-1} \trr\big(\underline{\tau}(\underline{\gamma}) \big).
$$
And given another arrow in $\gwrr{n}$, namely
$\big(\underline{\sigma}(\underline{\phi})
\ra{(\underline{\phi},g)} \,\,   g^{-1} \trr \underline{\tau}(\underline{\phi})\big)$,
with
$\underline{\sigma}(\underline{\phi})=  f^{-1} \trr\underline{\tau}(\underline{\gamma})$, 
the composition is:
\begin{equation}\label{compGN}
\Big(\underline{\sigma}(\underline{\gamma})\ra{(\underline{\gamma},f)}
 f^{-1}\trr \underline{\tau}(\underline{\gamma})\Big) 
\star 
\Big ( f^{-1}\trr\underline{\tau}(\underline{\gamma})
\ra{(\underline{\phi},g)} g^{-1}\trr \underline{\tau}(\underline{\phi}) \Big)
=\Big(\underline{\sigma}(\underline{\gamma})
\ra{(\underline{\gamma}\,\star \, (f\trr \underline{\phi}), f.g)} g^{-1}\trr\underline{\tau}(\underline{\phi})\Big).
\end{equation}
(These apparently  awkward conventions are justified by the graphical calculus in \S \ref{ss:msgc}.)
\end{Definition}

\subsection{Graphical calculus for  $\Gamma^n \rtimes \Sigma_n=\gwrr{n}$ }\label{ss:msgc}

{Recall \peq{diagtoperm} and \peq{proj}. 
Any welded (or virtual) braid $[Q]$, and any virtual braid diagram $Q$, gives rise to a permutation $U_Q$. Hence permutations can
be represented by virtual braids diagrams.
Thus we can use a graphical calculus to represent morphisms in
$\Gamma^n \rtimes \Sigma_n$, cf. \cite{MarWooLev}. 
We explain the graphical calculus  for $n=3$ (the other cases extend in the obvious way).}
  {This provides visual help for some   calculations in this paper.}
%

{Firstly note that bijections $f\colon \{1,2,3\} \to \{1,2,3\}$
can be represented by virtual braid diagrams as  below:}

\begin{minipage}[t]{0.85\textwidth}
$$\centerline{\relabelbox 
\epsfysize 2cm 
\epsfbox{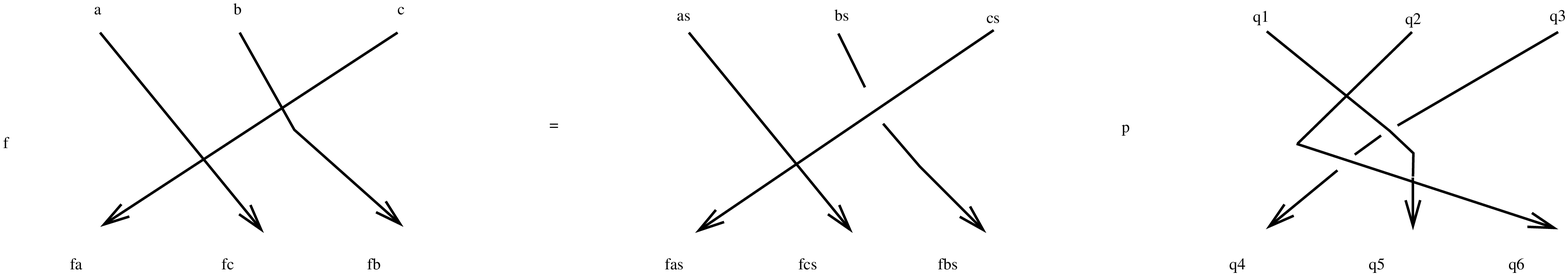}
\relabel{a}{$\scriptstyle{1}$}
\relabel{b}{$\scriptstyle{2}$}
\relabel{c}{$\scriptstyle{3}$}
\relabel{fa}{$\scriptstyle{f(3)=1}$}
\relabel{fb}{$\scriptstyle{f(2)=3}$}
\relabel{fc}{$\scriptstyle{f(1)=2}$}
\relabel{as}{$\scriptstyle{1}$}
\relabel{bs}{$\scriptstyle{2}$}
\relabel{cs}{$\scriptstyle{3}$}
\relabel{fas}{$\scriptstyle{f(3)=1}$}
\relabel{fbs}{$\scriptstyle{f(2)=3}$}
\relabel{fcs}{$\scriptstyle{f(1)=2}$}
\relabel{f}{$f=$}
\relabel{=}{$=$}
\relabel{p}{$=$}
\relabel{q1}{$\scriptstyle{1}$}
\relabel{q2}{$\scriptstyle{2}$}
\relabel{q3}{$\scriptstyle{3}$}
\relabel{q4}{$\scriptstyle{f(3)=1}$}
\relabel{q6}{$\scriptstyle{f(2)=3}$}
\relabel{q5}{$\scriptstyle{f(1)=2}$}
\endrelabelbox}
\hspace{.5in}
$$
\end{minipage}
\begin{minipage}[t]{0.12\textwidth}
\begin{equation}\label{diagtoPERM} 
\end{equation}
\end{minipage}

\smallskip
\noindent Here only initial and end-points of strands matter. 
In this diagrammatic presentation, the product $f.g$ of two permutations
is given by the vertical juxtaposition where $f$ stays on top of $g$. 
{(In order for this to work,  we must use convention \peq{sym} for the symmetric group.)}

\smallskip

The morphisms of $\Gamma^3\rtimes \Sigma_3$ can be represented
by diagrams as  in Fig. 
\ref{permc}.
{Only initial and end-points of strands matter, encoding
permutation components of morphisms in $\Gamma^3 \rtimes \Sigma_3$. Hence, we can switch positive crossings with negative or virtual crossings without changing the morphism in  $\Gamma^3\rtimes \Sigma_3$ {that} we are describing. }



\begin{figure}[H]
\centerline{\relabelbox 
\epsfysize 2.9cm 
\epsfbox{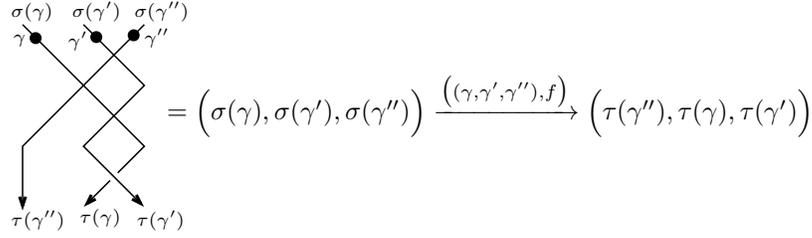}
\relabel{a}{$\scriptstyle{\sigma(\gamma)}$}
\relabel{b}{$\scriptstyle{\sigma(\gamma')}$}
\relabel{c}{$\scriptstyle{\sigma(\gamma'')}$}
\relabel{fa}{$\scriptstyle{\tau(\gamma'')}$}
\relabel{fb}{$\scriptstyle{\tau(\gamma')}$}
\relabel{fc}{$\scriptstyle{\tau(\gamma)}$}
\relabel{g}{$\scriptstyle{\gamma}$}
\relabel{gg}{$\scriptstyle{\gamma'}$}
\relabel{ggg}{$\scriptstyle{\gamma''}$}
\relabel{M}{$=\Big(\sigma(\g),\sigma(\g'),\sigma(\g'')\Big)\ra{\big ( (\gamma,\gamma',\gamma''),f\big) } \Big(\tau(\g''),\tau(\gamma),\tau(\g')\Big)$}
\endrelabelbox}
\caption{\label{permc} A morphism in $\Gamma^3\rtimes \Sigma_3$. Here $f\in \Sigma_3$ is such that $f(1)=2,f(2)=3$ and $f(3)=1$.  }
\end{figure}
%
%


\noindent
For $\gamma = \mbox{id}$ we may omit this and the blob ``$\bullet$'' entirely.

{Assume that two morphisms in 
$\Gamma^3 \rtimes \Sigma_3$ are composable. (In the case below this means
$\sigma(\phi)=\tau(\gamma'')$, $\sigma(\phi')=\tau(\gamma')$,
and $\sigma(\phi'')=\tau(\gamma)$.) Graphically their composition is via vertical juxtaposition, as below:}

\begin{figure}[H]
\centerline{\relabelbox 
\epsfysize 4.5cm 
\epsfbox{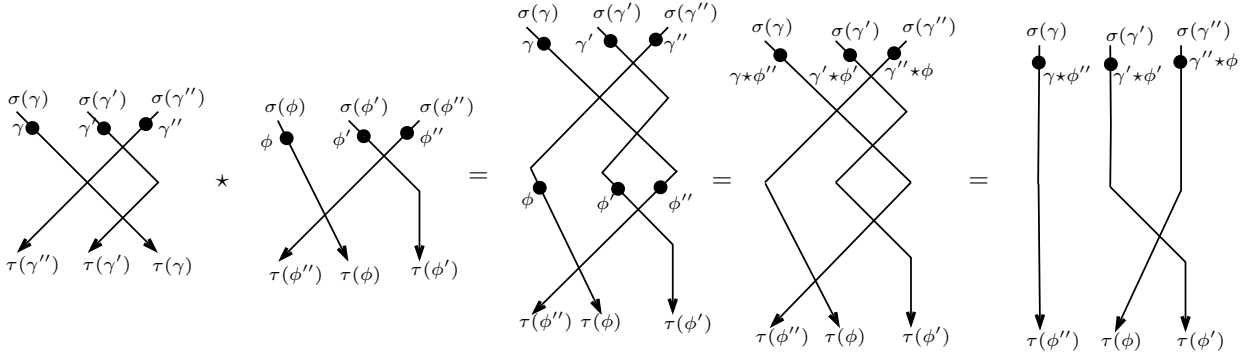}
\relabel{a}{$\scriptstyle{\sigma(\gamma)}$}
\relabel{b}{$\scriptstyle{\sigma(\gamma')}$}
\relabel{c}{$\scriptstyle{\sigma(\gamma'')}$}
\relabel{aaa}{$\scriptstyle{\sigma(\gamma)}$}
\relabel{bbb}{$\scriptstyle{\sigma(\gamma')}$}
\relabel{ccc}{$\scriptstyle{\sigma(\gamma'')}$}
\relabel{aaaa}{$\scriptstyle{\sigma(\gamma)}$}
\relabel{bbbb}{$\scriptstyle{\sigma(\gamma')}$}
\relabel{cccc}{$\scriptstyle{\sigma(\gamma'')}$}
\relabel{aaaas}{$\scriptstyle{\sigma(\gamma)}$}
\relabel{bbbbs}{$\scriptstyle{\sigma(\gamma')}$}
\relabel{ccccs}{$\scriptstyle{\sigma(\gamma'')}$}
\relabel{fa}{$\scriptstyle{\tau(\gamma'')}$}
\relabel{fb}{$\scriptstyle{\tau(\gamma')}$}
 \relabel{fc}{$\scriptstyle{\tau(\gamma)}$}
 \relabel{aa}{$\scriptstyle{\sigma(\phi)}$}
 \relabel{bb}{$\scriptstyle{\sigma(\phi')}$}
 \relabel{cc}{$\scriptstyle{\sigma(\phi'')}$}
\relabel{g}{$\scriptstyle{\gamma}$}
\relabel{gg}{$\scriptstyle{\gamma'}$}
\relabel{ggg}{$\scriptstyle{\gamma''}$}
\relabel{i}{$\scriptstyle{\phi}$}
\relabel{ii}{$\scriptstyle{\phi'}$}
\relabel{iii}{$\scriptstyle{\phi''}$}
\relabel{h}{$\scriptstyle{\gamma}$}
\relabel{hh}{$\scriptstyle{\gamma'}$}
\relabel{hhh}{$\scriptstyle{\gamma''}$}
\relabel{j}{$\scriptstyle{\gamma\star\phi''}$}
\relabel{jj}{$\scriptstyle{\gamma'\star \phi'}$}
\relabel{jjj}{$\scriptstyle{\gamma''\star \phi}$}
\relabel{js}{$\scriptstyle{\gamma\star\phi''}$}
\relabel{jjs}{$\scriptstyle{\gamma'\star \phi'}$}
\relabel{jjjs}{$\scriptstyle{\gamma''\star \phi}$}
\relabel{faa}{$\scriptstyle{\tau(\phi'')}$}
\relabel{fbb}{$\scriptstyle{\tau(\phi)}$}
\relabel{fcc}{$\scriptstyle{\tau(\phi')}$}
\relabel{faaa}{$\scriptstyle{\tau(\phi'')}$}
\relabel{fbbb}{$\scriptstyle{\tau(\phi)}$}
\relabel{fccc}{$\scriptstyle{\tau(\phi')}$}
\relabel{faaaa}{$\scriptstyle{\tau(\phi'')}$}
\relabel{fbbbb}{$\scriptstyle{\tau(\phi)}$}
\relabel{fcccc}{$\scriptstyle{\tau(\phi')}$}
\relabel{faaaas}{$\scriptstyle{\tau(\phi'')}$}
\relabel{fbbbbs}{$\scriptstyle{\tau(\phi)}$}
\relabel{fccccs}{$\scriptstyle{\tau(\phi')}$}
\relabel{f}{$\scriptstyle{\phi}$}
\relabel{ff}{$\scriptstyle{\phi'}$}
\relabel{fff}{$\scriptstyle{\phi''}$}
\relabel{bull}{$\star$}
\relabel{=}{$=$}
\relabel{ig}{$=$}
\relabel{igg}{$=$}
\endrelabelbox}
\caption{{Graphical notation for composing morphisms in $\Gamma^3\rtimes \Sigma_3$.} 
}
\label{composing}
\end{figure}
\noindent Hence,  blobs in a diagram can be moved freely along a strand, including pass a crossing, as long as no other blob gets in the way. We can compose the arrows colouring any two contiguous blobs   in the same strand.

\subsection{Conventions and notation for biracks and welded biracks (W-biracks) }\label{c:biracks}
References on quandles, biquandles and  biracks  include:
\cite{fenn_barth,fenn_biquandle,nelson,unsolved,C1,C2,C3,nelson_bilinear,reduce}.

\begin{Definition}[Birack, W-birack]
  \label{de:birack}
A {\em birack}  $(X,/,\backslash)$ is  a set $X$ with two maps $X\times X \to X$, denoted
$(x,y) \mapsto x \backslash  y $, called ``$x$ under $y$'', and
$(x,y) \mapsto y  / x $
called ``$y$ over $x$''.
  The conditions to be satisfied are:
\begin{enumerate}
\item Fixing $a \in X$, the maps 
$f_a, f^a : X \rightarrow X$ given by 
 $f_a : x  \mapsto x/a$ and
  $f^a : x  \mapsto x\backslash a$ each are invertible.  
\item\label{cond2} The map
  $S\colon  X \times X \rightarrow  X \times X$ given by 
$(x,y)  \mapsto (y/x,x\backslash y)$
  (called \cite{fenn_biquandle} a ``switch'')
 is invertible.
\item Putting
  $S_1=S \times \id_X \colon X \times X \times X \to X \times X \times X$ and 
 $S_2=\id_X  \times S $ 
   then 
               $S_1\circ S_2 \circ S_1 = S_2\circ S_1 \circ S_2 $.
\end{enumerate}


A birack is ``welded'', 
or a {\em W-birack},
if for each $x,y,z$ the following   hold:
\begin{align} \label{eq:wba}
  &(z/x)/y=(z/y)/x,
  & y\backslash z=y \backslash (z/x)), && x\backslash z=x \backslash (z/y) . 
\end{align}
{The last equation is of course redundant. It is left since it facilitates 
some of the diagrammatic proofs later.}
\end{Definition}


\smallskip

\mdef
 \label{qq}
Let $G$ be a  group.  
The  operations  $h/g=h$ and $g\backslash h=h^{-1} gh$
make $G$ a W-birack
(also a quandle \cite{fenn_biquandle}).

\smallskip



\mdef For any birack $X$, we may visualize 
the maps $S$ and $t_{1,2}: X^2 \rightarrow X^2$ given by 
$t_{1,2}: (x,y)\mapsto (y,x) $
as sending $(x,y)\in X \times X$ to
the pairs appearing at the  bottom  of the  diagrams:
$$
\xymatrix@R=10pt@C=15pt{&x\ar[dr]|\hole
  &y\ar[dl]\\ & y/x & x \backslash y}
\xymatrix@R=1pt{ \\ \\ \quad \textrm{ and }}
\xymatrix@R=15pt@C=15pt{&x\ar[dr] &y\ar[dl]\\ & y & x}\xymatrix@R=1pt{\\\\.}
$$
In particular   $S_1 \circ S_2 \circ S_1 = S_2 \circ S_1 \circ S_2 $ becomes
identification of the bottom rows here,  for each $x,y,z \in X$:
\begin{equation} \label{eq:stybe}
\hskip-0.8cm\xymatrix@R=14pt@C=8pt{ &x\ar[dr]|\hole
  & y\ar[dl] & z\ar[d]
  \\
  &y/x\ar[d] & x\backslash y\ar[dr]|\hole & z\ar[dl]
  \\
  &y/x\ar[dr]|\hole & z/(x\backslash y)\ar[dl]
  & (x\backslash y)\backslash z\ar[d]
  \\
  & (z/(x\backslash y))/(y/x) &(y/x)\backslash (z/(x\backslash y))
  & (x\backslash y)\backslash z}
\xymatrix@R=1pt{ \\ \\ \\ \\ \\\quad =}\hskip-0.5cm
\xymatrix@R=14pt@C=8pt{ &x\ar[d] & y\ar[dr]|\hole & z\ar[dl]  \\
  &x\ar[dr]|\hole & z/y\ar[dl] & y\backslash z\ar[d] \\
  &(z/y)/x\ar[d] & x\backslash(z/y)\ar[dr]|\hole & y\backslash z\ar[dl] \\
  &(z/y)/x & (y\backslash z)/(x\backslash(z/y)) & 
  (x\backslash(z/y))\backslash (y\backslash z)}\xymatrix@R=1pt{\\\\\\\\\\.}
\end{equation} 


The welded birack axiom (\ref{eq:wba}) says that 
the bottom
colours of the  diagrams below coincide,  if  $x,y,z \in X$:
\begin{equation} \label{eq:webi}
\hskip-1cm\xymatrix@R=14pt@C=10pt{ 
&x\ar[dr] & y\ar[dl] & z\ar[d]  \\
&y\ar[d] & x \ar[dr]|\hole & z\ar[dl] \\
&y\ar[dr]|\hole & z/ x\ar[dl] & x\backslash z\ar[d] \\
&(z/x) /y & y \backslash (z/x) & x\backslash z}
\xymatrix@R=1pt{ \\ \\ \\ \\ \quad =}
\xymatrix@R=12pt@C=10pt{ &x\ar[d] & y\ar[dr]|\hole & z\ar[dl]  \\
  &x\ar[dr]|\hole & z/y\ar[dl] & y\backslash z\ar[d] \\
  &(z/y)/x\ar[d] & x\backslash(z/y)\ar[dr] & y\backslash z\ar[dl] \\
  &(z/y)/x & y\backslash z & x \backslash (z/y)}\xymatrix@R=1pt{\\\\\\\\\\.}
\end{equation}
Note that  
the bottom colours  
{in \eqref{eq:easy}} coincide identically, given any triple $(x,y,z) \in X^3$:
\begin{equation} \label{eq:easy}
\hskip-1cm\xymatrix@R=12pt@C=18pt{ 
&x\ar[dr] & y\ar[dl] & z\ar[d]  \\
&y\ar[d] & x \ar[dr] & z\ar[dl] \\
&y\ar[dr]|\hole & z \ar[dl] & x\ar[d] \\
&z/y & y \backslash z &  x}
\xymatrix@R=1pt{ \\ \\ \\ \\ \quad =}
\xymatrix@R=10pt@C=18pt{ &x\ar[d] & y\ar[dr]|\hole & z\ar[dl]  \\
  &x\ar[dr] & z/y\ar[dl] & y\backslash z\ar[d] \\
  &z/y \ar[d] & x \ar[dr] & y\backslash z\ar[dl] \\
  &z/y  & y\backslash z & x }\xymatrix@R=1pt{\\\\\\\\\\.}
\end{equation}

{(See \S\ref{ss:vbg}.) The following is standard
\cite{fenn_biquandle,nelson,unsolved,fenn_barth}. In this paper we will extend this to bikoids.}

\begin{Lemma}
  \label{birackact}
If $(X,/,\backslash)$ is a birack, then  
$X^n=X \times \dots \times X$
has a {right-action}   ``$\trl$''
of $\VBG_n$
given by:
\begin{align*}
(x_1,\dots,x_{a-1}, x_a,x_{a+1},x_{a+2},\dots,x_n) \; \trl \; S_a^+{{[n]}} \;\;
  &= \;\; (x_1,\dots,x_{a-1}, x_{a+1}/x_a,x_a\backslash x_{a+1},x_{a+2},\dots x_n),
\\
(x_1,\dots,x_{a-1}, x_a,x_{a+1},x_{a+2},\dots,x_n) \;\trl\; V_a{{[n]}} \;\;
    &=\;\; (x_1,\dots,x_{a-1}, x_{a+1},x_a,x_{a+2},\dots x_n).
\end{align*}
\noindent
This action descends to an action of the welded braid group $\WBG_n$
if, and only if,   $(X,/,\backslash)$ is welded. 
\end{Lemma}
\proof{We follow the formulation in \peq{twoperpectives} for $\VBG_n$. Firstly $(-) \trl S^+_a{{[n]}}$ is invertible by item \ref{cond2} of Def. \ref{de:birack}.
Then note that 
$\underline{x} \trl S_a^+{{[n]}}$ and $\underline{x} \trl V_a{{[n]}}$, where $\underline{x}=(x_1,\dots,x_{a-1}, x_a,x_{a+1},x_{a+2},\dots,x_n)$,
are the bottom colours of the following diagrams:
$$
\xymatrix@R=15pt@C=0pt{
  & x_1\ar[d]
  &\dots\ar[d]<1.2ex> \ar[d]<-1.2ex> \ar[d]<0ex>
  &x_{a-1}\ar[d] & x_a \ar[drrr]|\hole
  &&& x_{a+1}\ar[dlll] &x_{a+2}\ar[d]
  & \dots\ar[d]<1.2ex> \ar[d]<-1.2ex> \ar[d]<0ex> & \ar[d] x_n
  \\ 
  & x_1 &\dots &x_{a-1}& x_{a+1}/x_a
  &&& x_{a}\backslash x_{a+1} &x_{a+2}& \dots
  & x_n
}
\xymatrix@R=1pt{\\ \\
  \quad \textrm{ and } \qquad } \hskip-0.4cm
\xymatrix@R=15pt@C=0pt{ & x_1\ar[d]
  &\dots\ar[d]<1.2ex> \ar[d]<-1.2ex> \ar[d]<0ex>
  &x_{a-1}\ar[d] & x_a \ar[drrr] &&& x_{a+1}\ar[dlll]
  &x_{a+2}\ar[d]& \dots\ar[d]<1.2ex> \ar[d]<-1.2ex> \ar[d]<0ex>
  & \ar[d] x_n
  \\ 
& x_1 &\dots &x_{a-1}& x_{a+1} &&& x_{a} &x_{a+2}& \dots & x_n
}\xymatrix@R=1pt{\\\\.}
$$
The  check of RIII \eqref{eq:R3} from Def. \ref{de:WBG} boils down to (\ref{eq:stybe});
VRIII \eqref{eq:VRIII} boils down to (\ref{eq:easy}). {Equations \eqref{eq:SRI},  \eqref{eq:Loc} and \eqref{VIII}} are trivial. Equation \eqref{eq:F} follows from \eqref{eq:webi}, which happens if, and only if, $(X,/,\backslash)$ is welded.  
\qed
}

\medskip

\mdef\label{lin} {{By extending linearly}, given a birack, we hence have a right representation $\trl$ of $\VBG_n$ on the vector space $\kappa(X^n)\cong \kappa(X)^{n \otimes}.$ Here $\kappa$ is a field. This descends to a representation of $\WBG_n$ if $(X,/,\backslash)$  is welded.}


\medskip
{ We will be particularly interested in biracks for which the reverse form of the {welded} Reidemeister III move in \eqref{rev}  does not hold for the action $\trl$ (Lem. \ref{birackact}) of $\WBG_n$ on $X^n$. These are called {\em essential biracks}.}

\begin{Definition}\label{essential1} {(We follow \cite{fenn_barth}.)}  A W-birack $(X,/,\backslash)$ is called {\it essential} if it \underline{does not hold that}:
\begin{align*}
\forall x,y,z \in X^3\colon & (x\backslash y) \backslash z=  (x\backslash z) \backslash y, &\textrm{ and } &&y/x=y/(x \backslash z), &&\textrm{ and } && z/(x \backslash y)=z/x.
\end{align*}
\end{Definition}
 \mdef\label{whyess}{Essential W-biracks are the most interesting  when it comes to applications to {$\WBG_n$}. This is because otherwise  $\trl$ in Lem. \ref{birackact} will descend to a representation of the unrestricted virtual braid group; see Rem. \ref{tornant}.}

%
%

\subsection{\Biker s: definition}\label{s:biker}


\begin{Definition}[\Biker]
  \label{de:biker}
A \biker\  $(\Gamma,X^+)$ is a groupoid
$\Gamma = \Gam$,
together with
set maps $L, R \colon{\Gamma_0\times \Gamma_0}\to \Gamma_1$,
such that the following hold.
Firstly,
$\sigma(L(x,y))=x$, $\sigma(R(x,y))=y$. Secondly, $L$ and $R$ define a birack 
$(\Gamma_0,/,\backslash)$, called the {\em underlying birack of   $(\Gamma,X^+)$}, by
$x \backslash y = \tau(L(x,y))$ and $y/x = \tau( R(x,y))$.  
That is, the {source} and {target} of $L(x,y)$ and $R(x,y)$ are as indicated below:
$$(x\ra{L(x,y)} x \backslash y) \textrm{ and }
(y \ra{R(x,y)} y / x).$$

Next define the following morphisms of the wreath product $\Gamma^2 \rtimes \Sigma_2=\gwrr{2}$ {(recall \peq{deftab}, \S \ref{thegr} and \S\ref{ss:msgc})}:
\begin{align}
Y^+(x,y)  &= \; {\Big ( x\tn y \ra{\big( (L(x,y),R(x,y)),1_{\Sigma_2}\big)} x\backslash y\tn y/x \Big)},\label{DEFY}\\
X^+(x,y) &=     Y^+(x,y) \star \big(x\backslash y \tn y/x 
\ra{ \big((\id_{x\backslash y},\id_{y/x}), {t_{1,2}^2} \big)} y/x \tn x\backslash y) \label{DEFX1}\\
&=\big( {x \tn y}\ra{\big( (L(x,y),R(x,y)),{t_{1,2}^2}\big)}  y/x \tn x\backslash y\big)  .\label{DEFX}
\end{align}
Note  $\underline{\sigma}(X^+(x,y))={(x,y)=x \tn y}$.
We use the notations (recall the diagrammatic calculus for $\gwrr{2}$ in \S\ref{ss:msgc}):
\begin{equation}\label{ab}
\xymatrix@R=1pt{\\\\ X^+(x,y)=}
\xymatrix{& x  \ar[dr]|\hole
  \ar[dr]|<<<<<{L(x,y)\,\bullet\quad \quad\,\,\,}|\hole
  & y  \ar[dl]\ar[dl]|<<<<<{\quad \quad \,\, \bullet R(x,y)}\\
& y/ x &x\backslash y
}\xymatrix@R=1pt{\\ \\\\\qquad=}
\xymatrix{& x  \ar[dr]
  \ar[dr]|<<<<<{L(x,y)\,\bullet\quad \quad\,\,\,}|\hole
  & y  \ar[dl]\ar[dl]|<<<<<{\quad \quad \,\, \bullet R(x,y)}\\
& y/ x &x\backslash y
}\xymatrix@R=1pt{\\\\\quad.}
\end{equation}

We impose that for  each $x,y,z \in \Gamma_0$ the equation below holds in 
$\Gamma^3 \rtimes \Sigma_3{=\gwrr{3}}$:

\begin{equation} \label{eq:r3}
\hskip-1cm
\xymatrix@R=25pt@C=8pt{ &
  x\ar[dr]|<<<<<<<<<{L(x,y)\,\,\,\bullet\qquad\,\,\,\,\,}|\hole\ar[dr]|\hole
  & y\ar[dl]\ar[dl]|<<<<<<<<{\qquad\,\,\,\,\,\,\, \bullet\,\,\ \,R(x,y)} & z\ar[d]  \\
  &y/x\ar[d] & x\backslash y\ar[dr]|<<<<<<{L(x\backslash
    y,z)\,\bullet\qquad\,\,\quad}|\hole\ar[dr]|\hole
  & z\ar[dl]|<<<<<<<{\qquad \,\,\,\,\,\,\,\,\, \bullet\,R(x\backslash y,z)}\ar[dl] \\
  &y/x\ar[dr]|<<<<<<<<<{L(y /x,z/(x\backslash
    y))\,\,\bullet\qquad\qquad\quad\,\,\,\, \,\,}|\hole \ar[dr]|\hole
  & z/(x\backslash y)\ar[dl]|<<<<<<<{\qquad \qquad\, \,\,\,
    \quad \,\bullet \, R(y /x,z/(x\backslash y))}\ar[dl] & (x\backslash y)\backslash z\ar[d] \\
  & (z/(x\backslash y))/(y/x) &(y/x)\backslash (z/(x\backslash y))
  & (x\backslash y)\backslash z}
\xymatrix@R=1pt{ \\ \\ \\ \\ \,\, =}
\xymatrix@R=25pt@C=10pt{ &x\ar[d] &
  y\ar[dr]|<<<<<<<<<{L(y,z)\,\,\,\,\bullet\qquad\,\,\,\,\,\,}|\hole\ar[dr]|\hole
  & z\ar[dl]|<<<<<<<<{\qquad\,\, \,\,\,\bullet \,\,\,R(y,z)}\ar[dl]  \\
  &x\ar[dr]|<<<<<<<{L(x,z/y)\,\,\bullet\qquad\,\,\,\,\quad}|\hole\ar[dr]|\hole
  & z/y\ar[dl]|<<<<<<<{\qquad \,\,\,\,\,\,\,\,\,
    \bullet\,\,R(x, z /y)}|\hole\ar[dl] & y\backslash z\ar[d] \\
  &(z/y)/x\ar[d]
  & x\backslash(z/y)\ar[dr]|<<<<<<{L(x \backslash(z/y),y \backslash
    z)\,\,\bullet\qquad\qquad\quad\,\,\,\, \,\,}|\hole\ar[dr]|\hole
  & y\backslash z\ar[dl]\ar[dl]|<<<<<<<<{\qquad \,\qquad\,
    \,\,\,\,\,\,\,\quad
    \bullet \,\,\,\,R(x \backslash(z/y),y \backslash z)} \\
  &(z/y)/x & (y\backslash z)/(x\backslash(z/y))
  & (x\backslash(z/y))\backslash (y\backslash z)}\hskip-0.5cm\xymatrix@R=1pt{\\\\\\\\\\\\\\\\.}
\end{equation}

\noindent 
{Recall that  undecorated lines carry identities in $\G$.
 Hence Equation (\ref{eq:r3}) means that for each $x,y,z \in \Gamma_0$:}
\begin{equation}\label{compbikoids}
\begin{split}
L(x,y)\star L(x \backslash y,z)&
    =L(x,z/y)  \star L(x \backslash(z/y),y \backslash z),
  \\
R(x,y)\star L(y /x,z/(x\backslash y))
  &=L(y,z)\star R(x \backslash(z/y),y \backslash z),
  \\
R(x\backslash y,z) \star R(y /x,z/(x\backslash y)) &= R(y,z) \star R(x,z /y).
\end{split}
\end{equation}
\end{Definition}

\begin{Remark}[Holonomy arrows] The $L(x,y)$ and $R(x,y)$ arrows are called {\em holonomy arrows}.
{In the context of finite group and finite 2-group topological field
theory, they encode Aharonov-Bohm phases {\cite{LL,Bais1,Bais2}} arising from flat connection holonomy and flat 2-connection 2-dimensional holonomy obtained when point-particles move in 2-dimensional space {and loop-particles} move in
3-dimensional space -- see \S\ref{motivation1} and \S\ref{phys}.}
\end{Remark}
\begin{Remark}We can define a \biker\ $(\Gamma,X^+)$ as being given by a birack structure on $\Gamma_0$, together with morphisms  $(L(x,y)\colon x\to x \backslash y) \textrm{ and }
(R(x,y)\colon y \to y / x),$ in $\G$, for each $x,y\in \Gamma_0$, satisfying {Equation \eqref{compbikoids}. The wreath groupoid approach chosen, and its graphical calculus, greatly 
facilitates the proofs to come.}
\end{Remark}

\mdef\label{refnow1} Let $(\Gamma,X^+)$ be a \biker. 
Consider the underlying birack $(\Gamma_0,/, \backslash)$ of
$(\Gamma,X^+)$.
In particular the map
$  (x,y)\in \Gamma_0\times \Gamma_0 
  \mapsto (y/x, x \backslash y)\in \Gamma_0\times \Gamma_0$
is bijective (Def. \ref{de:birack}). We put (the inverse is taken in $\Gamma^2\rtimes \Sigma_2$):
$$
X^{-}(y/x, x \backslash y)
=\Big ( (x,y)\ra{X^+(x,y)} (y/x, x \backslash y)\Big)^{-1}
=\big ( (y/x, x \backslash y)\ra{X^-(y / x,x \backslash y)} (x,y) \big) \in \Mor(\Gamma^2\rtimes \Sigma_2) .
$$
We use the following diagrammatic notations for $X^{-}(y/x, x \backslash y)$:
$$
\xymatrix@R=1pt{\\\\ X^-(y/x,x\backslash y)=}
\xymatrix{& y/x  \ar[dr]\ar[dr]|<<<<{\overline{R(x,y)}\,\bullet\quad \quad\,\,\,} & x\backslash y  \ar[dl]|\hole\ar[dl]|<<<<{\quad \quad \,\, \bullet \overline{L(x,y)}}|\hole\\
  & x & y}\xymatrix@R=1pt{\\\\ \\\qquad=}
\xymatrix{& y/x  \ar[dr]\ar[dr]|<<<<{\overline{R(x,y)}\,\bullet\quad \quad\,\,\,}|\hole & x\backslash y  \ar[dl]\ar[dl]|<<<<{\quad \quad \,\, \bullet \overline{L(x,y)}}\\
& x & y}\xymatrix@R=1pt{\\\\\,\,\,.}$$

\mdef \label{uselater}Thus
$\overline{L(x,y)}\star L(x,y)=\id_{x\backslash y}$ and
$L(x,y)\star \overline{L(x,y)}=\id_x$. And the same for $R(x,y)$ and $\overline{R(x,y)}$.



\begin{Definition}[Welded \biker] \label{de:wbiker}
{(Recall Def. \ref{de:birack} and Equation \eqref{eq:webi}.)}
A \biker\ 
$(\Gamma,X^+)$ is called welded, or a {\em \wwbiker}, if $(\Gamma_0,/,\backslash)$
is a welded birack, and for each $x,y,z \in \Gamma_0$ it holds that:
\begin{equation} \label{eq:wr1}
\hskip-2cm
\xymatrix@R=23pt@C=15pt{
  &x\ar[dr]|<<<<<{\id\,\,\bullet\quad}|\hole\ar[dr]
  & y\ar[dl]\ar[dl]|<<<<<{\quad \bullet \id\,\,\,} & z\ar[d]\ar[d]|{\quad\bullet \id\,\,}  \\
  &y \ar[d]\ar[d]|{\quad\bullet \id\,\,}
  & x \ar[dr]|<<<<{L(x,z)\,\,\bullet\quad\quad\,\,\,\,\,}|\hole\ar[dr]|\hole & z\ar[dl]|<<<<<{\quad\,\, \,\,\,\,\,\,\,\, \bullet\,R(x,z)}\ar[dl] \\
  &y \ar[dr]|<<<<<<{L(y ,z/ x)\,\,\bullet\quad\quad\quad\,\,\,\,}|\hole \ar[dr]|\hole
  & z/x \ar[dl]|<<<<<{\quad \quad\, \,\, \quad \bullet \, R(y,z / x )}\ar[dl]
  & x\backslash  z\ar[d]\ar[d]|{\quad\bullet \id\,\,} \\
  & (z/ x) / y
  &y \backslash (z/x)
  & x\backslash z }
\xymatrix@R=1pt{ \\ \\ \\ \\ \quad =}
\xymatrix@R=23pt@C=10pt{ &x\ar[d]\ar[d]|{\quad\bullet \id\,\,}
  & y\ar[dr]|<<<<<{L(y,z)\,\,\bullet\qquad\,\,\,}|\hole\ar[dr]|\hole
  & z\ar[dl]|<<<<<<{\qquad\,\,\, \bullet \,R(y,z)}\ar[dl]  \\
  &x\ar[dr]|<<<<<{L(x,z/y)\,\,\,\bullet\qquad\,\,\,\,\quad}|\hole\ar[dr]|\hole
  & z/y\ar[dl]|<<<<<{\qquad \,\,\,\,\,\,\,\,\, \bullet\,R(x, z /y)}|\hole\ar[dl]
  & y\backslash z\ar[d]\ar[d]|{\quad\bullet \id\,\,} \\
  &(z/y)/x\ar[d]\ar[d]|{\quad\bullet \,\id\,}
  & x\backslash(z/y)\ar[dr]|<<<{\,\,\,\,\,\,\id\,\,\bullet\quad\, \,\,\,\,}\ar[dr]
  & y\backslash z\ar[dl]\ar[dl]|<<<<{\,\quad \bullet \,\id\,\,\,\,} \\
  &(z/y)/x & y \backslash z & x\backslash (z/y)}\xymatrix@R=1pt{\\\\\\\\\\\\\\\,\,.}
\end{equation}
Since  $(\Gamma_0,/,\backslash)$  is a W-birack, the
 equation above means that for each $x,y,z \in \Gamma_0$ {we have that:} 
\begin{align}\label{weldedtoR}
  &L(x,z)=L(x,z/y),
  & L(y,z)=L(y,z/x), && R(x,z)\star R(y,z/x)=R(y,z)\star R(x,z/y).
\end{align}
\end{Definition}



\begin{Proposition}[Finite group \wwbiker] \label{groupalg}
Let $G$ be a group. 
We have a \wwbiker\ structure $X^+_G$ in  the groupoid
$\AUT(G)$ of {Example}  \ref{autg}. The underlying birack $(\Gamma_0 , /,\backslash) = (G,/,\backslash)$ is given by  $y/x =y$ and $ x \backslash y =y^{-1} x y $, cf. \peq{qq}.
{The} holonomy morphisms are {as follows} (below we put $\overline{y}=y^{-1}$, where $y \in G$):
\begin{equation}\label{fgwb}\xymatrix@R=0pt{ \\ X^+_G(x,y) =}
\xymatrix{& x \ar[dr]^{\,\,\,\bullet\, 1_G}|\hole &y\ar[dl]_{\overline{y}\,\,\bullet\,\,\,\,}\\
  &y & y^{-1}xy      }   \xymatrix@R=1pt{\\\\.}
\ignore{{
\quad \quad
\begin{CD}\\\\\\ \ &&&& \textrm{ thus }\quad
  \begin{cases}y/x =y \\ \\                                                                                               x \backslash y =y^{-1}xy\end{cases}\end{CD}
}}
\end{equation}
\end{Proposition}

\begin{proof}
{Recalling \peq{qq}, $(G,/,\backslash)$ is a W-birack. Equation \eqref{eq:r3} -- equivalently  \eqref{compbikoids} -- is true since:}
\begin{equation}\label{dIA1}
\xymatrix@R=20pt{& x\ar[dr]|\hole\ar[dr]|<<<<<{\overline{y}\, \,\bullet\,\,\,\,\,\,}|\hole & y \ar[dl]\ar[dl]|<<<<<{\,\,\,\,\,\,\,\,\,\bullet\,\,1_G} &z\ar[d]\ar[d]|{\,\,\,\,\,\,\, \bullet 1_G}\\
& y\ar[d]\ar[d]|{\,\,\,\,\,\,\, \bullet 1_G} & y^{-1}xy\ar[dr]|\hole\ar[dr]|<<<<{\overline{z} \,\,\,\,\bullet\,\,\,\,\,}|\hole &z\ar[dl]\ar[dl]|<<<<<{\quad	\,\,\bullet\,\,1_G}\\    
& y\ar[dr]|\hole\ar[dr]|<<<<<{\overline{z}\,\, \bullet\,\,\,\,}|\hole & z\ar[dl]\ar[dl]|<<<<<{\quad \,\,\bullet\,\,1_G} &z^{-1} y^{-1}xyz\ar[d]\ar[d]|{\,\,\,\,\,\,\, \bullet 1_G}\\                         
& z & z^{-1}yz &z^{-1}y^{-1}xyz
}\xymatrix@R=1pt{\\ \\ \\ \\ \quad =}\xymatrix@R=20pt{& x\ar[dr]\ar[dr]|<<<<<<<<{\overline{z}\,\overline{y}\,\, \,\,\bullet\,\,\,\,\,\,\,\,\,\,\,}|\hole & y \ar[dl]\ar[dl]|<<<<<<<{\,\,\,\,\,\,\,\,\bullet\,\,\,\overline{z}\,\,\,\,} &z\ar[d]\ar[d]|{\,\,\,\,\,\,\, \bullet 1_G}\\
& z^{-1}yz\ar[d]\ar[d] & z^{-1}y^{-1}xyz\ar[dr]&z\ar[dl]\\    
& z^{-1}yz\ar[dr] & z\ar[dl] &z^{-1} y^{-1}xyz\ar[d]\\                         
& z & z^{-1}yz &z^{-1}y^{-1}xyz
 } \xymatrix@R=1pt{\\\\\\\\\\\\\ .}
 \end{equation}
(Recall our convention for composition in $\AUT(G)$; see {Example}  \ref{autg}.) Whereas:
\begin{equation}\label{dIA2}
\xymatrix@C=20pt@R=20pt{& x \ar[d]\ar[d]|{\,\,\,\,\,\,\, \bullet 1_G} & y \ar[dr]|\hole\ar[dr]|<<<<<{\,\,\,\,\overline{z}\,\,\,\, \bullet\,\,\,\,\,\,\,\,\,\,}|\hole &z\ar[dl]\ar[dl]|<<<<<{\quad \,\bullet\,\,1_G}
\\
& x \ar[dr]|\hole\ar[dr]|<<<<{\overline{z}\, \bullet\,\,\,\,}|\hole & z\ar[dl]\ar[dl]|<<<<<{\quad \bullet\,1_G} &z^{-1}yz\ar[d]\ar[d]|{\,\,\,\,\,\,\, \bullet 1_G}\\    
& z\ar[d]\ar[d]|{\,\,\,\,\,\,\, \bullet 1_G} & z^{-1}xz\ar[dr]|\hole\ar[dr]|<<<<<{\overline{z^{-1}yz}\, \,\,\,\bullet\,\,\,\quad\quad  }|\hole &z^{-1}yz\ar[dl]\ar[dl]|<<<<{\quad\bullet\, 1_G}
\\                         
& z & z^{-1}yz& z^{-1}y^{-1} xyz 
}
\xymatrix@R=1pt{\\ \\ \\ \\\\ \quad =}
\xymatrix@C=5pt@R=20pt{& x \ar[d]\ar[d]|<<<<{ \overline{z^{-1}yz}\,\overline{z} \,\bullet\,\,\,\,\,\,\,\quad\quad} & y \ar[dr]\ar[dr]|<<<<<{\,\,\,\,\overline{z}\,\, \bullet\,\,\,\,\,\,\,\,\,}|\hole &z\ar[dl]\ar[dl]|<<<<<{\quad \,\,\bullet\,\,1_G}
\\
& z^{-1}y^{-1} xyz \ar[dr] & z\ar[dl] &z^{-1}yz\ar[d]\\    
& z\ar[d] & z^{-1}y^{-1} xyz\ar[dr] &z^{-1}yz \ar[dl]
\\                         
& z & z^{-1}yz& z^{-1}y^{-1} xyz 
}\xymatrix@R=1pt{\\\\.}
\end{equation}
{Since $\overline{z^{-1}yz}\,\overline{z}=z^{-1}y^{-1} z z^{-1}=\overline{z}\,\overline{y}$,
the  morphisms in $\gwrr{3}$ in the right-hand-sides of \eqref{dIA1} and \eqref{dIA2} coincide.
That the bikoid $X^+_G$ is welded is proved similarly; we will prove this 
in a more
general context in \S\ref{bs2}.}
\end{proof}

\noindent The W-\biker\ $(\AUT(G),X^+_G)$ is closely related to topological gauge theory in $D^2$; see \S\ref{motivation1}.

\smallskip

The proof of the following result is an easy exercise:
\begin{Proposition}\label{bigbirack}
Let  $(\Gamma,X^+)$ be a \biker. Then ${\G_1}$, the set of morphisms of $\Gamma$, is a birack with:
\begin{align*}
\big(x'\ra{\g} x) \backslash (y'\ra{\phi} y)&=(x'\ra{\gamma} x)\star (x \ra{L(x,y)} x\backslash y)=(x'\ra{\gamma\star L(x,y)} x\backslash y),\\
\big(y'\ra{\phi} y) / (x'\ra{\g} x)&=(y'\ra{\phi} y)\star (y \ra{R(x,y)} y / x)=(y'\ra{\phi\star R(x,y)} y / x).
 \end{align*}
Furthermore $(\Gamma,X^+)$ is a welded \biker\ if, and only if, $({\G_1},/,\backslash)$ is a welded birack.
\end{Proposition}
\mdef To $({\G_1},/,\backslash)$ we call the {\em upper birack} of  $(\Gamma,X^+)$. 

 \smallskip
\mdef\label{birackact2} {By combining Prop. \ref{bigbirack} with Lem. \ref{birackact}, we hence have a  representation $\trll$ of $\VBG_n$ on $\C\Gamma_1^n\cong (\C\Gamma_1)^{n \otimes}$, derived from the upper birack of a bikoid {$(\Gamma,X^+)$}. Explicit formulae are in  \peq{upper-rep}. {The representation $\trll$ of $\VBG_n$ descends to a representation of $\WBG_n$ when $(\Gamma,X^+)$ is welded.}  This $\trll$ will be  generalised in \S\ref{representations}.

\smallskip
{(Recall \peq{whyess}).  We will be particularly interested in essential W-bikoids. These are the W-bikoids for which the representation  $\trll$ does not descend from $\WBG_n$ to the unrestricted virtual braid group in Rem. \ref{tornant}.} 

\begin{Definition}\label{essential4}A W-bikoid\ is called {\it essential} if its upper birack is essential (as in  Def. \ref{essential1}).
\end{Definition}

 \begin{Example}\label{Gisessential} {A quick calculation shows that if $G$ is non-abelian then the W-bikoid $X^+_G$ in Prop. \ref{groupalg} is an essential W-bikoid. The lower birack of $X^+_G$, displayed in \peq{qq}, may be  non-essential when $G$ is non-abelian, if it happens that for all  $x,y,z \in G$, it holds that $(xy)z(xy)^{-1}=(yx)z(yx)^{-1}.$}
 \end{Example}

%

\subsection{\Biker s from abelian $gr$-groups} \label{abeliangrgroups}

There will be heavy  use of semidirect products in this document. Here are our conventions.

\mdef\label{sdc1} Let the group $G$ left-act by automorphisms on the group  $E$. Such action is denoted by:
$$(g,e) \in G \times E \mapsto g \trr e \in E. $$
Our convention for the semidirect product $G \ltimes E=  G\ltimes_\trr E$ is:
$$(g,a)(h,b)=(gh,a\,\,g \trr b), \textrm{ for each } g,h \in G \textrm{ and } a,b \in E. $$
Hence inverses in $G \ltimes_\trr E$ are: $(g,e)^{-1}=(g^{-1},g^{-1} \trr e^{-1}).$ 

\smallskip

\mdef\label{sdc2}  Let $g \in G$ and $e \in E$. We will frequently put $g=(g,1_E)$, where $1_E$ is the identity of $E$, and $e=(1_G,e)$. Hence $(g,e)=e\,\,g=g\,\,(g^{-1}\trr e) $. Also recall that given left-actions $\t$ of $G$ and 
of $E$ on a set $X$, then $(g,e) \t x\doteq e \t ( g\t x)$ is an action of $G\ltimes_{\trr} E$ on $X$ if, and only if, for each $g\in G$, $e \in E$ and $x \in X$: 
\begin{equation}\label{sdirectcomp}
e \t ( g\t x)=g \t \,\, \big(( g^{-1} \trr e) \t x\big).
\end{equation}

\newcommand{\agg}{abelian $gr$-group}

\begin{Definition}[Abelian $gr$-group]\label{abgrg} An abelian $gr$-group is given by pair $(G,A)$, or more correctly a triple $(G,A,\trr)$, where $G$ and $A$ are groups, with $A$ abelian, and $\trr$ is a left-action  of $G$ on $A$ by automorphisms.
Morphisms  $\phi\colon (G,A) \to (H,B)$ of abelian $gr$-groups are
  defined as pairs of group maps
  $\phi_1\colon G \to H$ and $\phi_2\colon A \to B$ preserving  group actions, namely $\phi_2(g \trr a)=\phi_1(g) \trr \phi_2(a)$, for each $g \in G$ and $a \in A$. The set of morphisms of abelian $gr$-groups between  $(G,A)$ and $(H,B)$ is denoted $\hom_{gr}\big((G,A),(H,B)\big)$.
\end{Definition}

\begin{Example}\label{gr1} {E.g., for $p$ a prime and $m \in \Z^+$,  put $G={\rm GL}(m,\Z_p)$, the group of invertible {$m\times m$} matrices in the field $(\Z_p,+,\times)$, and $A=(\Z_p^m,+)$. The action is by matrix multiplication.}
 \end{Example}
\smallskip
Let $(G,A)$ be an \agg, and $G \ltimes A=G \ltimes_\trr A$. Cf. Def. \ref{autg},
let
$\TRANS(G,A)=\AUT(G\ltimes A)$, the action groupoid  (Def. \ref{de:ag}) of the conjugation action of
$G \ltimes A$ on itself.
Thus arrows of
$\TRANS(G , A)$ 
 are: 
$$
\big((g,a)\ra{(w,k)} (w,k)(g,a)(w,k)^{-1}\big)
  = \big(wgw^{-1},k+w \trr a-(wgw^{-1}) \trr a)\big),
      \textrm{ where }  g,w \in G \textrm{ and } a, k \in A.
$$

\begin{Theorem}[Abelian $gr$-group \biker s]\label{AbelianGRgroups}
We have a \wwbiker\  $( \TRANS(G,A) , X^+_{gr})$,
given by:
\begin{equation}\label{kau-mar}
\begin{split}
\xymatrix@R=1pt{\\\\ X^+_{gr}\Big((z,a),(w,b)\Big)=}\hskip-1cm
\xymatrix{& (z,a)  \ar[dr]|\hole
  \ar[dr]|<<<<<<<<<{(w^{-1},0_A)\,\,\,\,\,\bullet\,\,\,
    \qquad\quad\,\,\,}|\hole
  & (w,b)  \ar[dl]\ar[dl]|<<<<<<<<<{\qquad \qquad \,\,\,\, \,\,\bullet \,\,\,(1_G , -{w}^{-1}\trr a)}\\
  &(w,a+b-w^{-1} \trr a)  & \big(w^{-1}zw,w^{-1} \trr a\big)}\xymatrix@R=1pt{\\\\\\.}
\end{split}
\end{equation}
{Hence the underlying birack of $X^+_{gr}$ is such that:}
\begin{align}\label{kau-mar2}
&(z,a) 	\backslash (w,b)=(w^{-1}zw,w^{-1} \trr a) & \textrm{and}  &&(w,b)/(z,a)=(w,a+b-w^{-1} \trr a).
\end{align}

\end{Theorem}
\noindent {We  omit the proof, as we will treat a more general case
(\biker s derived from crossed modules) in {Thm. \ref{bik-proof}.}}

\begin{Proposition}{(Cf. Def. \ref{essential4}.)}
 The W-\biker\ $X^+_{gr}$ in \eqref{kau-mar} is essential if, and only if, \underline{it does not hold that}:
\begin{align*}\label{essiff}
{\forall 
y,z \in G, \forall a \in A:\quad}
 &{yz=zy}, &&{y^{-1}\trr a=(y^{-1}z^{-1})\trr a,} &\textrm{and }&&{z^{-1}\trr a =(z^{-1}y^{-1})\trr a}.
\end{align*}
{Hence if $G$ is non-abelian, or if the action of $G$ on $A$ is non-trivial, then the W-\biker\ $X^+_{gr}$  in \eqref{kau-mar} is essential.}
\end{Proposition}
\begin{proof}This follows by easy calculations.\end{proof}

\begin{Example}\label{gr2} It is  possible that  $X^+_{gr}$ \ in \eqref{kau-mar} is essential when $G$ is abelian. E.g. take $G=\Z_2=(\{1,-1\},\times)$ and $A=\Z_3=(\{0,1,-1\},+)$. Hence $\Z_2$ acts on $\Z_3$  as $x\trr a=xa$. The associated W-\biker\ is essential.  
\end{Example}

\subsubsection{{Algebraic topological} interpretation of abelian $gr$-group
  \wwbiker s: `balloons and hoops'}\label{bandh}

{We freely use \S \ref{motion} and Thm. \ref{main0}. We use} the same notation for $B \in \WBG_n$ and its image $T_B\in \LBG_n$. 

Recall \cite{martins_kauffman,satoh} that the underlying birack \peq{qq} of the finite group \biker\
$(\AUT(G),X^+_G)$ of Prop. \ref{groupalg} essentially computes the   cardinality of the set of group maps from the knot group of a welded knot into $G$.

{Let $n\in \Z^+$. The underlying birack \eqref{kau-mar2} of the abelian $gr$-group \biker\  $X^+_{gr}=( \TRANS(G,A) , X^+_{gr})$ of Thm.  \ref{AbelianGRgroups} is strongly related to the action of $\pi_1(D^3\setminus C_n,*)$ on $\pi_2(D^3\setminus C_n,*)$. {We now explain this.}}


 {Consider the base point $*=(0,0,1)\in \d(D^3)$ for $D^3$.
Hence  any homeomorphism $f\colon (D^3,C_n) \to (D^3,C_n)$ fixes $*$. So do the isotopies we consider.}
Homotopically, $(D^3\setminus C_n,*)$
is a wedge product of $n$ circles and $n$ 2-spheres.
The fundamental group $\pi_1(D^3\setminus C_n,*)$
is the free product $\vee_{i=1}^n \Z_i$, where
$\Z_i=\{m_i^k\}_{k \in \Z} \cong (\Z,+)$.
Here $m_i\in \pi_1(D^3 \setminus C_n,*)$
{is  
associated with the oriented unknotted circle {$S^1_i$,} in the usual way; see Fig. \ref{K}.}
\begin{figure}[H]
\centerline{\relabelbox 
\epsfysize 3.5cm 
\epsfbox{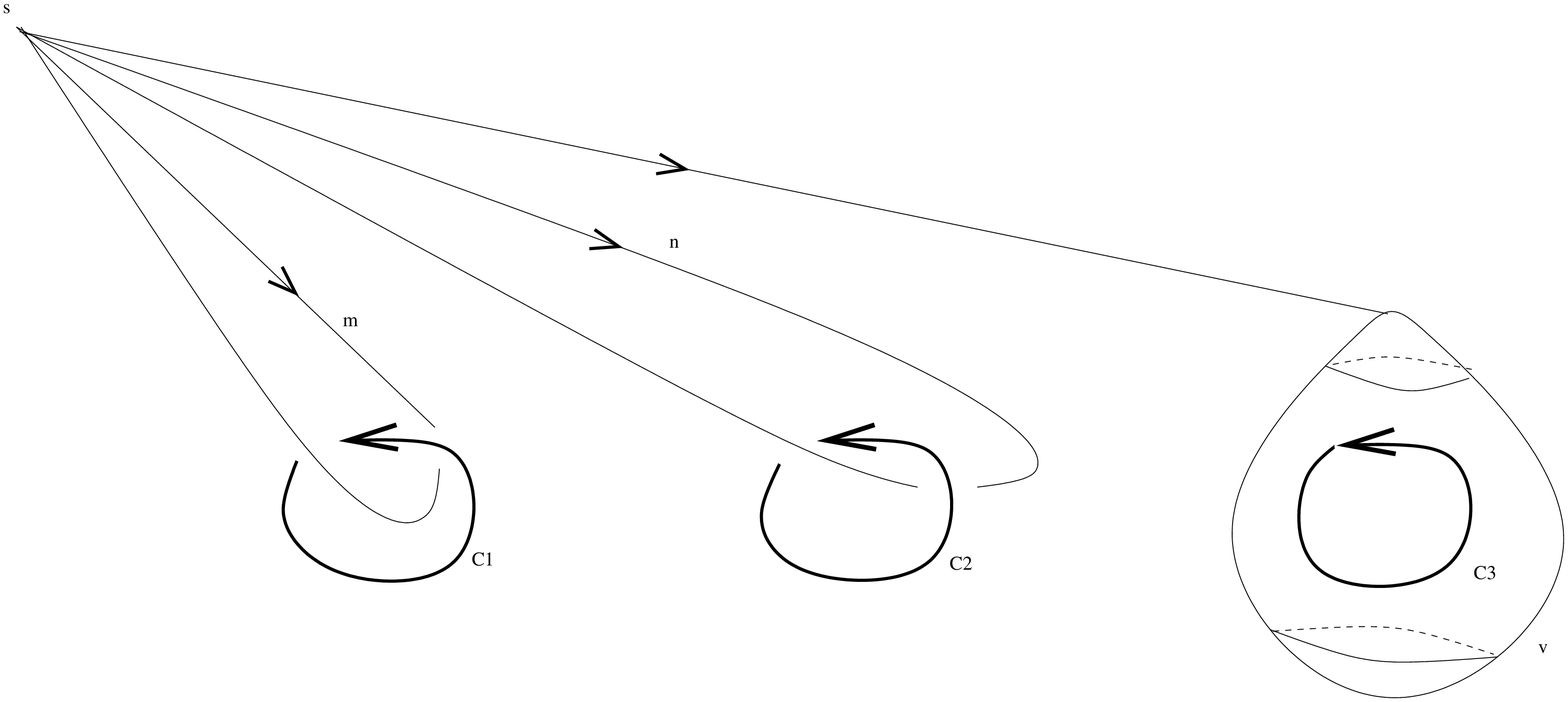}
\relabel{s}{$\scriptstyle{*}$}
\relabel{m}{$\scriptstyle{m_i}$}
\relabel{n}{$\scriptstyle{m_j}$}
\relabel{C1}{$\scriptstyle{S^1_i}$}
\relabel{C2}{$\scriptstyle{S^1_j}$}
\relabel{C3}{$\scriptstyle{S^1_k}$}
\relabel{v}{$\scriptstyle{v_k}$}
\endrelabelbox}
\caption{Elements  {$m_i,m_j\in \pi_1(D^3\setminus C_n,*)$, $v_k\in \pi_2(D^3\setminus C_n,*)$ given by the unknotted circles $S^1_i,S^1_j$,$S^1_k$.}\label{K}}
\end{figure}
\noindent{(This is  the usual Wirtinger presentation of the fundamental group of a knot complement, see e.g. \cite{Rolfsen}.)} The group-algebra of
$\pi_1(D^3\setminus C_n,*)$,
with coefficients in $\Z$, is thus isomorphic to the algebra of Laurent polynomials
$\Z\{m_1,m_1^{-1},\dots,m_n,m_n^{-1}\}$, in the non-commuting variables
$m_1,\dots,m_n$.

The second homotopy group $\pi_2(D^3\setminus C_n,*)$
is isomorphic to the second homology group of the universal cover of $D^3\setminus C_n$, which recall is homotopic to the  wedge product of $n$ 1-spheres $S^1$ and $n$ 2-spheres $S^2$. Therefore  $\pi_2(D^3\setminus C_n,*)$ is isomorphic to  the free
$\Z\{m_1,m_1^{-1},\dots,m_n,m_n^{-1}\}$--module
on the variables $v_1,\dots,v_n$; {\cite[Example 4.27]{hatcher} treats a particular case of this.}
{Geometrically,
each $v_k \in \pi_2(D^3\setminus C_n,*)$ traces a `balloon'
encircling the circle $S^1_k$; see Fig. \ref{K}.}
{(A similar `balloons and hoops' picture appears in \cite{Ballons}.)}

Given a {path}-connected pointed space $(X,*)$, the abelian $gr$-group
made out of $\pi_1(X,*)$ acting on $\pi_2(X,*)$, in the usual way,
is denoted by $\pi_{1,2}(X,*)$. The previous two paragraphs imply {(given the freeness of $\pi_1(D^3\setminus C_n,*)$ and the freeness of $\pi_2(D^3\setminus C_n,*)$ as a $\Z(\pi_1(D^3\setminus C_n,*))$-module)} {the following}:

\begin{Lemma}\label{thekey}{
  Let $(G,A)$ be an abelian $gr$-group.
  Abelian $gr$-group maps
  $\phi=(\phi_1,\phi_2)\colon \pi_{1,2}(D^3\setminus C_n,*)
  \rightarrow (G,A)$
  are in  bijective correspondence with sequences of the form:
  $\big((g_1,a_1), \dots, (g_n,a_n)\big)$, where $g_i \in G$ and $a_i \in A$, for each $i=1,\dots,n$. 
The bijection is defined  by   $g_i=\phi_1(m_i)$ and $a_i=\phi_2(v_i)$, for $i=1,\dots,n$.}
\end{Lemma}

The loop braid group $\LBG_n$ in \S\ref{motion}
has a natural left-action $\t$ in $\pi_{1,2}(D^3\setminus C_n,*)$ by abelian $gr$-group maps. {In another language, we have a group map from  $\LBG_n$ to the automorphism group of the abelian $gr$-group $\pi_{1,2}(D^3\setminus C_n,*)$. This extends the embedding of the loop braid group into the automorphism group of a free group -- the bit consisting of basis conjugating automorphisms -- as explained in \cite[\S 4]{damiani} and \cite{Sav}.}

Let us give details. {A homeomorphism} $f\colon (D^3, C_n) \to (D^3, C_n)$, considered up to pair-isotopy, functorially gives an isomorphism $f_*\colon \pi_{1,2}(D^3\setminus C_n,*) \to \pi_{1,2}(D^3\setminus C_n,*).$ 
And then given $m\in \pi_{1}(D^3\setminus C_n,*)$ and $v\in \pi_{2}(D^3\setminus C_n,*)$, we put $f\t m=f_*(m)$ and $f\t v=f_*(v)$.
{Hence} $\LBG_n$ right-acts in
$\hom_{gr}\big( \pi_{1,2}(D^3\setminus C_n,*),(G,A)\big)$, namely given $f \in \LBG_n$, we send $\phi\in \hom_{gr}\big( \pi_{1,2}(D^3\setminus C_n,*),(G,A)\big) $ to $\phi\circ f_* \doteq\phi \l f.$ 

Looking at {Fig. \ref{ISotopy} in Thm. \ref{main0},} we can see that  $S_i^+{{[n]}},V_i{{[n]}}\in \LBG_n$  act in $\pi_{1,2}(D^3\setminus C_n,*)$ as:
\begin{align*}
& V_i{{[n]}} \t m_j=m_j, \textrm{ if } j\notin {i,i+1}, 
&V_i{{[n]}}\t m_{(i+1)}=m_i,
&&  V_i{{[n]}}\t m_{i}=m_{(i+1)}, \\
& V_i{{[n]}} \t v_j=v_j, \textrm{ if } j\notin {i,i+1}, 
&  V_i{{[n]}} \t v_{(i+1)}=v_i,
&&  V_i{{[n]}} \t v_{i}=v_{(i+1)},
\\
& S^+_i{{[n]}} \t m_j=m_j, \textrm{ if } j\notin {i,i+1},
&  S^+_i{{[n]}} \t m_{(i+1)} =m_i,
&& S^+_i{{[n]}}\t m_{i} =m_i^{-1}m_{i+1}{m_i},\\
& S^+_i{{[n]}} \t v_j =v_j, \textrm{ if } j\notin {i,i+1},
& &&  S^+_i{{[n]}}\t v_{i}=m_i^{-1}\trr v_{(i+1)}.
\end{align*}

{The hardest action to address is $S^+_i{{[n]}} \t v_{(i+1)}$.
Cf. the left bit of Fig. \ref{ISotopy}, when the circle $S^i_i$ goes inside
the  circle $S^1_{i+1}$,  it `drags' the balloon $v_{(i+1)}$ with it. 
Since the isotopy connecting the top and bottom of the left part of Fig \ref{ISotopy}  can be made local, the sum $v_i +v_{i+1}$ in $\pi_2(D^3\setminus C_n,*)$, which can be visualised as
a large balloon encircling the circles $S_i^1$ and $S_{i+1}^1$, remains stable
during  the isotopy, since the isotopy can be made local enough so that it happens well inside the balloon $v_i +v_{i+1}$. In a more precise language, this means:}
$$ S^+_{i}{{[n]}}\t (v_i+v_{(i+1)})=v_i+v_{(i+1)}.$$
Since
$S^+_{i}{{[n]}}\t (v_i+v_{(i+1)})=S^+_{i}{{[n]}}\t v_i  +S^+_{i}{{[n]}}\t v_{(i+1)}$ and $ S^+_i{{[n]}}\t v_{i}=m_i^{-1}\trr v_{(i+1)}$, it follows that:
$$
 S^+_i{{[n]}} \t v_{(i+1)}=v_i+v_{(i+1)}-m_i^{-1}\trr v_{(i+1)}.
$$

\noindent{Comparing with {Equation} \eqref{kau-mar}, this implies, by noting that the $V_a{{[n]}}$ and the $S^+_a{{[n]}}$ generate $\LBG_n$, that:}
\begin{Theorem}[Topological interpretation of abelian $gr$-group W-\biker s]
Let $(G,A)$ be an abelian $gr$-group.
Consider the right-action  $\trl$ of the loop braid group $\LBG_n\cong \WBG_n$
on $(G\times A)^n$ derived (Lem. \ref{birackact}) from  the underlying
birack \eqref{kau-mar2} of the welded \biker\  $X_{gr}^+$ of \eqref{kau-mar}.
Looking at the elements in $(G\times A)^n$ as being \agg\ maps
$\phi\colon \pi_{1,2}(D^3\setminus C_n,*)
\rightarrow (G,A)$ -- by Lem. \ref{thekey} --
{then given a $B\in \LBG_n$}  it holds that
{$\phi\trl B=\phi\circ B_*$,} where $B_*\colon \pi_{1,2}(D^3\setminus C_n,*)\to {\pi_{1,2}}(D^3\setminus C_n,*)$ is the induced map on homotopy groups.
\end{Theorem}

\section{Virtual and welded braid group representations from \biker s}\label{representations}Throughout this section we will fix a groupoid $\Gamma$ and a \biker\ $(\Gamma,X^+)$; see} Def. \ref{de:biker}. 
%

\subsection{Representations of $\WBG_n$ derived from \biker s}\label{fromR}
{Recall the groupoid algebra $\C(\G)$ in Def. \ref{groupoid_algebra}.} We define the following element $\Rc\in \C(\Gamma) \tn \C(\Gamma)$:
\begin{equation}\label{def2ofR}
\Rc =\sum_{x,y \in \Gamma_0}(x \ra{L(x,y)} x\backslash y) \tn (y \ra{R(x,y)} y/x).
\end{equation}
\begin{Lemma}\label{propsofR}
The element $\Rc$ is invertible and its inverse is {(recall the notation introduced in \peq{invnot}):}
\begin{equation}\label{def2ofRm}
\Rc^{-1} =\sum_{a,b \in \Gamma_0}(a\backslash b \ra{\overline{L(a,b)}} a) \tn ( b/a  \ra{\overline{R(a,b)}} b).
\end{equation}\label{propR}
Moreover $\Rc$ satisfies the relation below:
\begin{equation}\label{prop2R}
\Rc_{12}\Rc_{13}\Rc_{23}=\Rc_{23}\Rc_{13}\Rc_{12}, \textrm{ in }  \C(\Gamma)\tn\C(\Gamma)\tn \C(\Gamma).
\end{equation}
Here $\Rc_{12}=\Rc \tn \id_{\C(\Gamma)}$, $\Rc_{23}=  \id_{\C(\Gamma)}\tn \Rc$ and $\Rc_{13}=\sum_{x,y \in \Gamma_0}(x \ra{L(x,y)} x\backslash y) \tn \id_{\C(\Gamma)}\tn (y \ra{R(x,y)} y/x). $
Furthermore, if  the \biker\ $(\Gamma,X^+)$ is welded, then in $\C(\Gamma)\tn\C(\Gamma)\tn \C(\Gamma)$:
\begin{equation}\label{prop2Rw}
\Rc_{13}\Rc_{23}=\Rc_{23}\Rc_{13}. 
\end{equation}
\end{Lemma}
{Note that Equation \eqref{prop2R} is satisfied by R-matrices in quasi-triangular bialgebras;  see e.g. \cite[Thm.VIII.2.4]{Kassel}.}
\begin{proof}
 We have:
 \begin{align*}
  \Rc \Rc^{-1}&=\sum_{x,y,a,b \in \Gamma_0}\Big(
  (x \ra{L(x,y)} x\backslash y) \tn (y \ra{R(x,y)} y/x)\Big) \Big((a\backslash b \ra{\overline{L(a,b)}} a) \tn ( b/a  \ra{\overline{R(a,b)}} b)\Big)\\ 
  &=\sum_{x,y,a,b \in \Gamma_0}
  \Big( (x \ra{L(x,y)} x\backslash y) (a\backslash b \ra{\overline{L(a,b)}} a) \Big) \tn \Big(
  (y \ra{R(x,y)} y/x)(b/a  \ra{\overline{R(a,b)}} b)\Big)\\
  &\stackrel{(\#1)}{=}\sum_{\substack{x,y,a,b \in \Gamma_0\\
  x\backslash y=a \backslash b\, \wedge\, y/x =b/a 
  }}
  \Big( (x \ra{L(x,y)} x\backslash y) (a\backslash b \ra{\overline{L(a,b)}} a) \Big) \tn \Big(
  (y \ra{R(x,y)} y/x)(b/a  \ra{\overline{R(a,b)}} b)\Big)\\
    &\stackrel{(\#2)}{=}\sum_{\substack{x,y,a,b \in \Gamma_0\\ 
    x=a\, \wedge\, y=b}}
  \Big( (x \ra{L(x,y)} x\backslash y) (a\backslash b \ra{\overline{L(a,b)}} a) \Big) \tn \Big(
  (y \ra{R(x,y)} y/x)(b/a  \ra{\overline{R(a,b)}} b)\Big)\\
    &\stackrel{(\#3)}
    {=}\sum_{x,y \Gamma_0}
  \big( x \ra{\id_x} x \big) \tn \big( y\ra{\id_y} y\big)=1_{\C(\G)} \tn 1_{\C(\G)}.
 \end{align*}
Here $(\#1)$ follows from the definition of the product \eqref{pga} in $\C(\G)$, and $(\#2)$ from the fact that {$(x,y) \mapsto ( y / x,x\backslash y)$} is a bijection,  by definition of biracks. Step $(\# 3)$ follows from \peq{uselater}.

That $\Rc \Rc^{-1}=1_{\C(\G)} \tn 1_{\C(\G)}$ is proven analogously.

Now note: 
\begin{align*}
 \Rc_{12}&\Rc_{13}\Rc_{23}=\smash{\sum_{x,y, a,z, u,v\in \G_0}}
 \Big( (x \ra{L(x,y)} x\backslash y) \tn (y \ra{R(x,y)} y/x) \tn 1_{\C{(\G)}}\Big)\\
  &\phantom{-----------------}\Big( (a \ra{L(a,z)} a\backslash z)  \tn  1_{\C{(\G)}}\ \tn (z \ra{R(a,z)} z/a) \Big)\\
 &\phantom{--------------------}\Big( 1_{\C{(\G)}}\tn (u \ra{L(u,v)} u\backslash v) \tn (v \ra{R(u,v)} v/u) \Big)\\
&=\smash{ \sum_{{x,y, a,z, u,v\in \G_0}}}
 \Big( 
 (x \ra{L(x,y)} x\backslash y) (a \ra{L(a,z)} a\backslash z)  \Big) \tn\Big( (y \ra{R(x,y)} y/x) (u \ra{L(u,v)} u\backslash v) \Big) \\&\phantom{-----------------------}
\tn\Big( (z \ra{R(a,z)} z/a) (v \ra{R(u,v)} v/u) \Big)
 \\&\stackrel{(\#1)}{=} \smash{\sum_{\substack{x,y,a,z, u,v\in \G_0\\ 
a=x \backslash y\,\wedge\, u= y/x\, \wedge \, v=z/a
}}}
 \Big( 
 (x \ra{L(x,y)} x\backslash y) (a \ra{L(a,z)} a\backslash z   )\Big)
 \tn\Big( (y \ra{R(x,y)} y/x) (u \ra{L(u,v)} u\backslash v) \Big)\\ &\phantom{---------------------------}\tn\Big( (z \ra{R(a,z)} z/a) (v \ra{R(u,v)} v/u) \Big)\\
 &\stackrel{(\#2)}{=}\sum_{x,y,z  \in \Gamma}\Big (x \ra{L(x,y)\star L(x \backslash y,z)} (x\backslash y)\backslash z\Big) \tn \Big(y \ra {R(x,y)\star L(y /x,z/(x\backslash y))}(y/x)\backslash (z/(x\backslash y)) \Big)\\ &\phantom{-----------------------} \tn \Big (z \ra{R(x\backslash y,z) \star R(y /x,z/(x\backslash y))}(z/(x\backslash y))/(y/x)\Big)\\
 &\stackrel{(\#3)}{=}\sum_{x,y,z  \in \Gamma}\Big (x \ra{L(x,z/y)  \star L(x \backslash(z/y),y \backslash z)}  (x\backslash(z/y))\backslash (y\backslash z)  \Big)
 \tn \
 \Big(y \ra {L(y,z)\star R(x \backslash(z/y),y \backslash z)   } (y\backslash z)/(x\backslash(z/y)) \Big)\\ &\phantom{---------------------------} \tn 
 \Big (z \ra{R(y,z) \star R(x,z /y)
} {(z/y)/x}\Big)\\
&
{=}\Rc_{23}\Rc_{13}\Rc_{12}.
\end{align*}
Here $(\#1)$ and $(\#2)$ follow from  {\eqref{pga}}. On the other hand $(\#3)$ follows from 
\eqref{eq:stybe}, {(\ref{eq:r3})} and \eqref{compbikoids}.  

That $\Rc_{13}\Rc_{23}=\Rc_{23}\Rc_{13}$ if the \biker\ is welded follows analogously, by using {\eqref{eq:wba}, \eqref{eq:webi}, \eqref{eq:wr1}  and \eqref{weldedtoR}.}
\end{proof}

We can hence can define  representations of the virtual braid group $\VBG_n$ derived from a \biker\ $(\G,X^+)$. {The following is one of our main results.}
\begin{Theorem}\label{rep12} Let $(\G,X^+)$ be a bikoid.
 Let $V$ be a right-representation of the groupoid algebra $\C(\G)$; \S\ref{ga}. Let $n \in \Z^+$. We have a representation $\trl^*$ of $\VBG_n$ on $V^{n \tn}$, which on the  generators of $\VBG_n$ has the form:
\begin{equation}\label{act2c}(v_1 \tn \dots\tn v_{a-1}  \tn v_a \tn v_{a+1}  \tn v_{a+2}\tn  \dots \tn v_n)\trl^* V_a{{[n]}}
=v_1 \tn \dots \tn  v_{a-1}
  \tn v_{a+1}\tn v_a \tn v_{a+2} \tn \dots \tn v_n,
\end{equation}
 \begin{multline}\label{act2b}
(v_1 \tn \dots\tn v_{a-1}  \tn v_a \tn v_{a+1}  \tn v_{a+2}\tn  \dots \tn v_n)\trl^* S^+_a{{[n]}}\\
=\displaystyle\sum_{x,y \in \Gamma_0}v_1 \tn \dots \tn  v_{a-1}
  \tn v_{a+1}.\big (y \ra{R(x,y)} y /x\big)\tn v_a.
    \big (x \ra{L(x,y)} x\backslash y \big) \tn v_{a+2} \tn \dots \tn v_n.
\end{multline}
 Moreover, $\trl^*$ descends to a representation of  $\WBG_n$ on $V^{n \tn}$ if $(\Gamma,X^+)$ is welded. 

  These are unitary representations if $V$ is a unitary representation of the groupoid algebra; see Def. \ref{unirep}.
 \end{Theorem}

 \mdef\label{uncat} {If $V$ is the object-regular representation (see {Example}  \ref{obre}) of $\C(\Gamma)$, then $\trl^*$ coincides with $\trl$ in  \peq{lin}.}

\mdef\label{uncat1} {If $V$ is the right-regular representation (see {Example}  \ref{rre}) of $\C(\Gamma)$, then $\trl^*$ coincides with $\trll$ in  \peq{birackact2}.}
\begin{proof}
 {(We are using the definition in \peq{twoperpectives} for $\WBG_n$)}. First of all note that $(-)\trl^* S^+_a{{[n]}}$ and $(-)\trl^* V_a{{[n]}}$ {each are} invertible maps $V^{n\tn} \to V^{n\tn}$. 
 In order to prove $\trl^*$ is a representation, we must check that the relations in  Def. \ref{de:WBG} for the virtual braid group, not involving the $S^-_a{{[n]}}$, hold. The Reidemeister III move follows from \eqref{prop2R}, whereas Virtual and Mixed Reidemeister III moves, as well as Locality, follow trivially. If $(\Gamma,X^+)$ is welded, then  \eqref{prop2Rw} holds, from which it follows the {Welded} Reidemeister III move.

Let us now prove that the representation of $\VBG_n$ is unitary if $V$ is unitary; see Def. \ref{unirep}. It suffices proving   unitarity for $S_a^+{{[n]}}$. This follows easily from the following calculation, which encodes unitarity for $n=2$. If  $u,v,u',v' \in V$, we have:
\begin{align*}
 &\left \langle \sum_{x,y \in \Gamma_0}u'.\big (y \ra{R(x,y)} y /x\big)\tn u.\big (x \ra{L(x,y)} x\backslash y \big), \sum_{z,w \in \Gamma_0}v'.\big (w \ra{R(z,w)} w /z\big)\tn v.\big (z \ra{L(z,w)} z\backslash w \big)\right  \rangle\\
&\stackrel{(\#1)}=  \sum_{z,w \in \Gamma_0}\sum_{x,y \in \Gamma_0} \left\langle u'\tn u, v'.\big  (w \ra{R(z,w)} w /z\big) \star (y/x \ra{\overline{R(x,y)}} y) \tn v.(z \ra{L(z,w)} z\backslash w ) \star (x \backslash y \ra{\overline{L(x,y)} } x\big)\right \rangle\\
&\stackrel{(\#2)}=\langle u'\tn u, v'\tn v\rangle.
\end{align*}
In $(\# 1)$ we have used Def. \ref{unirep} and the notation of \peq{refnow1} and \peq{uselater}. Step $(\#2)$ follows from  (in $\C(\G) \tn \C(\G)$):
\begin{align*}\sum_{z,w \in \Gamma_0}\sum_{x,y \in \Gamma_0}& \big  (w \ra{R(z,w)} w /z\big) \star (y/x \ra{\overline{R(x,y)}} y) \tn (z \ra{L(z,w)} z\backslash w ) \star (x \backslash y \ra{\overline{L(x,y)} } x\big)\\&\stackrel{(\# 1)}{=}
\sum_{z,w \in \Gamma_0}\sum_{\substack{x,y \in \Gamma_0\\ w /z=y/x\\ z\backslash w =x \backslash y  }} \big  (w \ra{R(z,w)} w /z\big) \star (y/x \ra{\overline{R(x,y)}} y) \tn (z \ra{L(z,w)} z\backslash w ) \star (x \backslash y \ra{\overline{L(x,y)} } x\big) \\
&\stackrel{(\# 2)}{=}
\sum_{z,w \in \Gamma_0}\big  (w \ra{R(z,w)} w /z\big) \star (w/z \ra{\overline{R(z,w)}} {w}) \tn (z \ra{L(z,w)} z\backslash w ) \star (z \backslash w \ra{\overline{L(z,w)} } z\big) \\
&\stackrel{(\# 3)}=\sum_{z,w \Gamma_0}( w \ra{\id_w} w )\tn( z \ra{\id_z} z)=1_{\C(\Gamma)}\tn 1_{\C(\Gamma)}.
\end{align*}
Note that step $(\# 1)$ follows from \eqref{pga}. Step $(\#2)$ follows from the fact that the map $(z,w) \mapsto (w / z, z \backslash w)$ is bijective, by Def. \ref{de:birack}. Step $(\#3)$ follows from the description in \peq{refnow1} and  \peq{uselater}.
 This finishes the proof.
\end{proof}

The representation in Thm. \ref{rep12} generalises. Pick different representations $V_1,\dots, V_n$ of $\C(\Gamma)$, and consider $(-)\trl^*S^+_a{{[n]}}$ and $(-)\trl^*V_a{{[n]}}$ in (\ref{act2c}) and (\ref{act2b}) to act on $v_1\tn \dots \tn v_n \in V_1 \tn \dots \tn V_n$. In general, this does not  define a representation of $\VBG_n$,  since $(-)\trl^*S^+_a{{[n]}}$ and $(-)\trl^*V_a{{[n]}}$ each send $V_{1}\tn \dots\tn  V_a\tn V_{a+1} \tn \dots \tn V_n$ to  $V_{1}\tn \dots\tn  V_{a+1}\tn V_{a} \tn \dots \tn V_n$,
but it still satisfies the $\VBG_n$ relations. {We give details in the next subsection.}

\subsection{Representations of the categories $\G$-$\WBG_n$ and ${\cal U}\G$-$\WBG_n$ from bikoids}\label{fromM}

The results in this section are stated for W-\biker s and representations of $\WBG_n$; see {Def. \ref{de:WBG} and Def. \ref{de:wbiker}.} All results apply, {\it mutatis mutandis}, to \biker s and the virtual braid group $\VBG_n$. {We work over $\C$.  }

{Throughout the subsection, we fix $n\in \Z^+$. Given a pair of vector spaces $U$ and $V$, we put $L(U,V)$ to denote the vector space of linear maps $U \to V$.} 
Given vector spaces $U_1,\dots,U_n,V_1,\dots,V_n$, linear maps $g_1\colon U_1 \to V_1,\dots,g_n\colon U_n \to V_n$, and $f \in\Sigma_n$, we put:
\begin{equation}\label{fact}
f\trr \big(U_1 \tn \dots \tn U_n \ra{g_1 \tn \dots \tn g_n}  V_1 \tn \dots \tn V_n\big)=\big(U_{f(1)} \tn \dots \tn U_{f(n)} \ra{g_{f(1)} \tn \dots \tn g_{f(n)}}  V_{f(1)} \tn \dots \tn V_{f(n)}\big).
\end{equation}
{This yields a linear map $f \trr (-)$,  from $L(\otimes_{i=1}^n U_i,\otimes_{i=1}^n V_i)$  to  $L(\otimes_{i=1}^n U_{f(i)},\otimes_{i=1}^n V_{f(i)})$.}


\begin{Definition}[{Categories} ${\cal V}^n$ and ${\cal UV}^n$] 
We define the following category ${\cal V}^n$. Objects are given by sequences $\U=(U_1,\dots, U_n)$ of vector spaces. Given objects $\U$ and $\V=(V_1,\dots, V_n)$, the morphisms $\U \to \V$ are given by linear maps  $F\colon \otimes_{i=1}^n U_i \to \otimes_{i=1}^n V_i$. Given objects $\U,\V,\W$, {as well as} morphisms $F\colon \U \to \V$ and $G \colon \V \to \W$, their composition in ${\cal V}^n$ is  given by the following  composition of linear maps (see \peq{coc}):
$$\big(\U \ra{F} \V\ra{G} { \W}\big)=\big( \U \ra{G \circ F }\W\big).$$
{The category ${\cal UV}^n$ is defined analogously. Objects of ${\cal UV}^n$ are  given by sequences $\U=(U_1,\dots, U_n)$ of vector spaces, provided with inner products, and the linear maps $F\colon \otimes_{i=1}^n U_i \to \otimes_{i=1}^n V_i$ are required to be unitary.}
\end{Definition}

\noindent {The symmetric group $\Sigma_n$ left-acts in the categories ${\cal V}^n$ and  ${\cal UV}^n$ by functors, by means of Equation \eqref{fact}.}

\begin{Definition}[{Categories} $\G$-$\WBG_n$ and ${\cal U}\G$-$\WBG_n$] Let $\Gamma$ be a groupoid.
We have categories $\G$-$\WBG_n$ and ${\cal U}\G$-$\WBG_n$.   Objects of $\G$-$\WBG_n$ are given by sequences {$\uR=(R_1,\dots, R_n)$} of representations of $\C(\Gamma)$; cf. \S\ref{ga}. Given objects $\underline{R}$ and $\underline{S}$, morphisms $ \underline{R} \ra{B} \underline{S}$  are given by welded braids $B \in \WBG_n$ such that $U_B^{-1} \trr \underline{R}=\underline{S}$, where $U_B$ is the underlying permutation of $B$; see \peq{proj}. The composition in $\G$-$\WBG_n$  has the form below:
$${\big(\uR\ra{B} U_B^{-1}\trr \uR \ra{B'} U_{B'}^{-1}\trr  U_B^{-1} \trr\uR\big)= \big(\uR\ra{BB'} U_{B'}^{-1}\trr  U_B^{-1}\trr \uR\big)=\big(\uR\ra{BB'} U_{BB'}^{-1}\trr \uR\big).}$$

 {The category ${\cal U}\G$-$\WBG_n$ is similarly defined, the only difference being that  the $R_1,\dots,R_n$ are required to be unitary representations of $\C(\Gamma)$.}
\end{Definition}
\noindent We can see morphisms of $\G$-$\WBG_n$ as welded braids, whose strands are coloured with representation of $\C(\Gamma)$.
%
\begin{Theorem}\label{uf}
Let $(\Gamma,X^+)$ be a W-\biker. We have a functor ${\cal F}\colon\Gamma$-$\WBG_n\to {\cal V}^n.$ 
{The functor $\F$ sends an $n$-tuple of representations $\uR=(R_1,\dots,R_n)$ of $\C(\Gamma)$ to  $U(\uR)=(U(R_1),\dots,U(R_n))$, where $U(R_i)$ is the underlying vector space of $R_i$. }
{On morphisms, $\F$ is uniquely specified by its value on the morphisms of the form $\uR\ra{B}U_B^{-1} \trr 
\uR$, where $B$ is a generator of {the monoid} $\WBG_n$, as in Equation \eqref{mongens}. {In these particular cases, $\F$ has the form in shown in Equations \eqref{a1}, \eqref{a2} and \eqref{a3}}, below:}
\begin{multline}\label{a1}
\F\left( (U_1,\dots , U_a , U_{a+1}, \dots, U_n) \ra{S^+_a{{[n]}}}  
(U_1, \dots , U_{a+1} , U_{a} , \dots, U_n)
\right)(v_1 \tn \dots\tn v_a \tn v_{a+1}  \tn \dots \tn v_n)\\=
\displaystyle\sum_{x,y \in \Gamma_0}v_1 \tn \dots \tn  v_{a-1}
  \tn v_{a+1}.\big (y \ra{R(x,y)} y /x\big)\tn v_a.
    \big (x \ra{L(x,y)} x\backslash y \big) \tn v_{a+2} \tn \dots \tn v_n,
\end{multline}
\begin{multline}\label{a2}
\F\left((U_1 , \dots , U_a , U_{a+1} , \dots, U_n) \ra{S^-_a{{[n]}}}  
(U_1,\dots, U_{a+1} , U_{a} , \dots,U_n)
\right)(v_1 \tn \dots\tn v_a \tn v_{a+1}  \tn \dots \tn v_n)\\=
\displaystyle\sum_{x,y \in \Gamma_0}v_1 \tn \dots \tn  v_{a-1}
  \tn v_{a+1}.\big (x\backslash y \ra{\overline{L(x,y)}} x \big)
\tn v_{a}.\big (y/x  \ra{\overline{R(x,y)}}y \big)
 \tn v_{a+2} \tn \dots \tn v_n,
\end{multline}
\begin{multline}\label{a3}
\F\left((U_1 , \dots , U_a , U_{a+1} , \dots, U_n) \ra{V_a{{[n]}}}  
(U_1 , \dots , U_{a+1} , U_{a} , \dots, U_n)
\right)(v_1 \tn \dots\tn v_a \tn v_{a+1}  \tn \dots \tn v_n)\\=
v_1 \tn \dots \tn  v_{a-1}
  \tn v_{a+1}
\tn v_{a}
 \tn v_{a+2} \tn \dots \tn v_n.
\end{multline}

\noindent Equations \eqref{a1}--\eqref{a3} also gives a functor ${\cal UF}\colon{\cal U}\Gamma$-$\WBG_n\to {\cal UV}^n,$ sending  $\uR=(R_1,\dots,R_n)$ to $U(\uR)=(U(R_1),\dots, U(R_n))$, where $U(R_i)$ is the underlying inner-product space of the unitary representation $R_i$. 
\end{Theorem}
\begin{proof}Follows from the calculations in the proof of Thm. \ref{rep12}. We can  also combine \peq{catpres} and Lem. \ref{propsofR}.
\end{proof}

\section{\Biker s from crossed modules} \label{ss:4}
We already have, from \S\ref{fromR} and \S\ref{fromM}, machines for constructing
representations of the welded braid group $\WBG_n$, hence \S\ref{motion} of the loop braid group $\LBG_n$, from
W-\biker s. The task now is to  manufacture 
{W-\biker s} 
from  crossed modules. These will be our main examples of W-\biker s. {They  are related to Aharonov-Bohm {like} effects \cite{LL,EN,Bais1} inherent to moving loop-particles  in topological higher gauge theory in $D^3$, {see \S\ref{phys}.}}

\smallskip
Let the group $G$ have a left-action $(g,e) \in G \times E \mapsto g \trr e \in E$ on the group $E$  by automorphisms. In order to not surcharge our notation with {brackets}, let us convention that actions   have priority over group multiplication, meaning for instance that we put $g\trr a\,\, g\trr b \,\,c$ for $(g\trr a)\,\,(g\trr b) \,\,c$, where $a,b,c \in E$ and $g,h \in G$.

\smallskip

{The following simple observation will be useful later on in \S\ref{bs2}.}

\mdef\label{transport} Let $\G=\Gam$ and $\G'=\Gamp$ be groupoids. Let us be given a groupoid inclusion $g\colon \G\to \G'=(g_0\colon \G_0 \to \G_0'\,,\, g_1\colon \G_1 \to \G_1')$. Hence $g_0$ and $g_1$ are injective and preserve all structure maps in $\G$ and $\G'$. Suppose in addition that $g_0$ is surjective. If we have a \biker\ (or W-\biker) $X^+_\G=X^+$ on $\Gamma$ then $X^+_\G$ can be {\em transported} to a \biker\ (or W-\biker)  $X^+_{\G'}$ in $\G'$, in the obvious way. {Explicitly we put:} $$\xymatrix@R=1pt{\\\\ X^+_{\G'}(x',y')=}
\xymatrix{& x'  \ar[dr]|\hole
  \ar[dr]|<<<<<<<<<<<{g_1\big(L\big(g_0^{-1}(x'),g_0^{-1}(y')\big)\big)\,\,\,\bullet\qquad\quad \qquad\qquad \qquad}|\hole
  & y'  \ar[dl]\ar[dl]|<<<<<<<<<<<<{\quad \quad \qquad \qquad\,\,\,\quad \quad \,\,\bullet g_1\big (R\big(g_0^{-1}(x'),g_0^{-1}(y')\big)\big)}\\
& g_0(g_0^{-1}(y')/ g_0^{-1}(x')) &g_0\Big(g_0^{-1}(x')\backslash g_0^{-1}(y')\Big)
}\xymatrix@R=1pt{\\\\\\\quad.}
$$

\subsection{Crossed modules ${\cal G}$}

\begin{Definition}[Crossed module]\label{xmod}Let $E$ and $G$ be groups. A crossed module   of groups
${\cal G}=(\d\colon E \to G,\trr)$
(see
\cite{martins_tams,brown_higgins_sivera,brown_hha,baez_lauda})
is  a group map $\d\colon E \to G$, and a left action
of $G$ on $E$ by automorphisms, such that the  relations below, called
Peiffer relations hold for each $g \in G$ and $e,e' \in E$: 
\begin{align}\label{PR}
 &\textrm{1st Peiffer relation: } \d(g \trr e)=g \d(e) g^{-1},&
\textrm{2nd Peiffer relation: } \d(e) \trr e'=ee'e^{-1}.
\end{align}
\end{Definition}

\begin{Lemma}\label{xmodprop}Let  ${\cal G}=(\d\colon E \to G,\trr)$ be a crossed module.
a) The group $A=\ker(\partial)\subset E$ is central in $E$,  and in particular $A$ is abelian. b)  The group $A$ is closed under the action of $G$, meaning that if $K \in A$ and $g \in G$, then $ g \trr K \in A$.  
  c) Given any $g,h \in G$, $e \in E$, and $f \in A$ it holds that:
$(g \d(e) h) \trr f=(gh) \trr f$.
 d) The action of $G$ on $A$ descends to an action of 
${\rm coker}(\partial)$ on $A$.
\end{Lemma}

\begin{proof}
 Assertions a), c), d) follow from the 2nd Peiffer relation, and  b)  from the 1st Peiffer relation. 
\end{proof}

\mdef\label{grt} A crossed module ${\cal G}$ hence gives rise to
two abelian $gr$-groups: $(G,A)$ and  
$({\rm coker} (\partial),A)$; see Def. \ref{abgrg}.

\mdef\label{grtox} {Conversely, if $(G,A,\trr)$ is an abelian $gr$-group, then we have a crossed module $(A \stackrel{a \mapsto 1_G}{\longrightarrow} G,\trr)$. Looking at {the abelian $gr$-groups in Examples \ref{gr1} and \ref{gr2}, this gives many}  examples of crossed modules of finite groups.} 

\smallskip

\mdef {Let $G$ be a group and ${\rm Aut}(G)$ be its group of automorphisms, acting on $G$ in the obvious way. Let $g \in G \mapsto {\rm Ad}_g \in {\rm Aut}(G)$ be such that ${\rm Ad}_g(x)=gxg^{-1}$. This defines a crossed module $({\rm Ad}\colon G \to {\rm Aut}(G),\trr)$.}


%
\subsection{{Groupoids $\TRANS(T^2_R(\Gc))$ and $\TRANS(S^2(\Gc))$, and the associated W-bikoids}}\label{refBIK}
Throughout this subsection, we fix a crossed module $\Gc=(\d\colon E \to G,\trr)$, and  an element $R \in E$; {see \peq{defR}.} 
{We will explicitly construct the W-\biker s $X^+_R$ and $X^+_{gr^*}$ of Equations \eqref{compact} in \eqref{addX2}, appearing at the end of the physics motivation section \S\ref{phys}. We  firstly construct $X^+_{gr^*}$, and them obtain $X^+_R$ by transporting  $X^+_{gr^*}$ along a groupoid isomorphism; see \peq{transport}. The ideas are simple, but require several algebraic preliminaries.} 
{A slightly depleted version of the W-bikoids $X^+_{gr^*}$ and $X^+_R$ is the W-bikoid $X^+_{gr}$ in Equation \eqref{kau-mar}.}

\smallskip
 Our convention for commutators in a group $G$ is $[p,q]=pqp^{-1}q^{-1}$, where $p,q \in G$.
\subsubsection{The groups $T^2_R(\Gc)$  and $S^2(\Gc)$}

\begin{Definition} Fix an $R\in E$. We let $T^2_R(\Gc)=\big\{(g,\d(R),e) \in G \times G \times E\,\,\big |\,\, \d(e)=[\d(R),g]\big\}$.
\end{Definition}
 Explicit calculations swiftly prove the following lemma. (Cf. \peq{2f} and the second equation of \eqref{compos}.)
\begin{Lemma}
The set $T^2_R(\Gc)$ is a group with the operation
$(g,\d(R),e)(g',\d(R),e')=(gg',\d(R), e \,\, g \trr e') $. Hence inverses in  $T^2_R(\Gc)$ take the form $(g,\d(R),e)^{-1}=(g^{-1},\d(R),g^{-1} \trr {e^{-1}}).$ {The unit is $(1_G,\d(R),1_E)$.}
\end{Lemma}
\begin{proof}
 Associativity is immediate. The most important bit is to show that $\d(e\,\,g \trr e')=[\d(R),gg']$. Note:
\begin{align*}
\d(e\,\,g^{-1} \trr e')=\d(e)\,\,g^{-1} \d( e') g=[\d(R),g]\,g[\d(R),g']g^{-1}=[\d(R),gg'],
\end{align*}
where we have used  the 1st {Peiffer relation} in Def. \ref{xmod}.
\end{proof}

\begin{Definition} {(Recall the conventions in \peq{sdc1}.)} Let $S^2(\Gc)=G \ltimes _\trr A$.  Here  $A=\ker(\partial)$, which is closed under the action of $G$ on $E$; see Lem. \ref{xmodprop}. Hence the product in $S^2(\Gc)$ is
$(g,K)\,\, (h,L)=(gh, K\,\, g \trr L). $
\end{Definition}

\mdef\label{thisgen} {We have a group isomorphism  $\phi_{1_E}\colon (g,J) \in S^2(\Gc) \mapsto   (g,\d(1_E),J) \in T^2_{1_E}(\Gc)$. This  generalises.}

 \begin{Lemma}\label{defpsiphi}Let $R \in E$. We have a group isomorphism $\psi_R\colon T^2_R(\Gc) \to S^2(\Gc)$. It is defined as: $$(g,\d(R),e) \stackrel{\psi_R}{\longmapsto}  (g,R^{-1} \,\,e \,\,g \trr R).$$  Its inverse is given by $\phi_R\colon S^2(\Gc) \to T^2_R(\Gc)$, where $\phi_R(g,J)=(g, \d(R), R \,\,J\,\,g \trr R^{-1})$.
\end{Lemma}
\begin{proof}
Let $(g,\d(R),e) \in T^2_R(\Gc)$. Then $\d(e)=[\d(R),g]$, from which it follows, by using $\d(g \trr R)=g\,\d(R)\,g^{-1}$,  that $R^{-1}\, e\, g \trr R\in A=\ker(\partial)$. Thence $  (g,R^{-1} \,e\, g \trr R)\in S^2(\Gc)$.  Analogously, if $J$ in $\ker(\partial)$, it holds that  $\d(R \,\,J\,\,g \trr R^{-1})=[\d(R),g]$. (We have used  the 1st {Peiffer rule in \eqref{PR}}.) Hence {$\phi_R(g,J)\in T^2_R(\Gc)$, if $g \in G$}.

Given $(g,\d(R),e),(g',\d(R),e') \in T^2_R(\Gc)$, we have:
\begin{align*}
\psi_R\big ((g,\d(R),e)(g',\d(R),e') \big )&= \psi_R\big ( gg',\d(R), e \,\, g \trr e')=\Big(gg',R^{-1}\,\,  e \,\, g \trr e'\,\, (gg')\trr R\Big)\\
  &=(gg',  R^{-1} \,\,e\,\,g \trr R\,\,\, g \trr R^{-1} \,\, g\trr e'\,\,(gg') \trr R)= (g, R^{-1} \,\,e\,\,g \trr R) \,(g', R^{-1} \,\,e'\,\,g' \trr R)\\
&=\psi_R\big (g,\d(R),e\big) \,\psi_R(g',\d(R),e').\end{align*}
  
  To finalise, note that clearly $\psi_R\circ \phi_R=\id_{S^2(\Gc)}$ and  $\phi_R\circ \psi_R=\id_{T^2_R(\Gc)}$. 
\end{proof}

\mdef\label{defTheta} We have a homomorphism $\Theta\colon S^2(\Gc) \to S^2(\Gc)$ defined as: $(g,J)\stackrel{\Theta}{\mapsto} (g,1_A)$. Note $\Theta \circ \Theta=\Theta$. 

\noindent By combining with the previous lemma,  
{or by a direct calculation,} it follows  that:
\begin{Lemma}\label{defThetaR}
 We have a group homomorphism $\Theta_R\colon T^2_R(\Gc)\to  T^2_R(\Gc)$, which has the explicit form:
\begin{equation}\label{dEFT}(g,\d(R),e) \in T^2_R(\Gc) \stackrel{\Theta_R}{\longmapsto} (g,\d(R),R\,\ g\trr R^{-1}) \in T^2_R(\Gc) .  
\end{equation}

Furthermore  $\Theta_R\circ \Theta_R=\Theta_R$ and $\Theta_R= \phi_R \circ \Theta \circ \psi_R$.
\end{Lemma}

\mdef\label{defbeta} We will make strong use of the map of sets $\beta\colon S^2(\Gc) \to S^2(\Gc)$, defined as \smash{$(g,K)\stackrel{\beta}{\mapsto} (1_G,K^{-1})$.} 
Note that $\beta(g,K)=(g,1_A)(g,K)^{-1}$, in other words: $\beta(g,K)=\Theta(g,K)\,\, (g,K)^{-1}$, for each $(g,K) \in S^2(\Gc)$.

\begin{Definition}\label{defbetaR}
Analogoulsy, if $(g, \d(R),e) \in T^2_R(\Gc)$, we  put $\beta_R(g, \d(R),e)\in T^2_R(\Gc)$ to be:
\begin{equation}\label{betaversions}
\begin{split}
\beta_R(g, \d(R),e)&=\Theta_R(g,\d(R),e)\,\,(g,\d(R),e)^{-1}\\ &=(1_G,\d(R), R\,\, g \trr R^{-1}\,\, e^{-1})\in T^2_R(\Gc).
\end{split}
\end{equation}
\end{Definition}
\mdef\label{betasame} By construction, or by a direct calculation, it follows that: $\beta_R=\phi_R \circ \beta \circ \psi_R$.
\begin{Lemma}\label{def:t}We have left-actions $\t$ of $E$ on $T^2_R(\Gc)$ and on $S^2_R(\Gc)$. They are such that: 
\begin{align}
a \t (g,\d(R),e)&=(\d(a)g,\d(R),  R\,\,a\,\, R^{-1} \,\, e\,\, a^{-1}),\label{tact1}\\
a \t (g,K)&=(\d(a)g,K). \label{tact2}
\end{align}
Here $a \in E$.
Morever $\psi_R\colon T^2_R(\Gc) \to S^2(\Gc)$, hence also $\phi_R\colon S^2(\Gc) \to T^2_R(\Gc)$, preserves $E$ actions.
\end{Lemma}
\noindent {(The action in Equation \eqref{tact1} was mentioned in  \peq{lastc}. It arises from the gauge transformation rule for the 2-dimensional holonomy of a 2-connection \eqref{leftrightconj}.) These actions are, in general, not actions by automorphisms.}
\begin{proof}
Firstly, if $(g,\d(R),e) \in T^2_R(\Gc)$, then an easy calculation shows that $\d(R\,\,a\,\, R^{-1} \,\, e\,\, a^{-1})=[\d(R),\d(a) g]$. 
 Let $a,b \in E$. Clearly $a\t b\t (-)=(ab)\t (-)$, both for $\t$ in \eqref{tact1} and $\t$ in \eqref{tact2}. Finally note that:
\begin{align*}
 \psi_R\big (a \t (g,\d(R),e) \big)&=\psi_R\big (\d(a)g,\d(R),  R\,\,a\,\, R^{-1} \,\, e\,\, a^{-1} \big)=\big(\d(a)g,  R^{-1}\,\,  R\,a\, R^{-1} \,\, e\,\, a^{-1} \,\, (\d(a)g) \trr R\big) \\
&\stackrel{(\#1)}{=}\big(\d(a)g, R^{-1} R\,\, a \, R^{-1} \,\, e\,\, a^{-1} \,\,  a\,\, g \trr R \,\, a^{-1}\big) \\
&=\big(\d(a)g, a\, R^{-1} \,\, e\,\,   g \trr R \,\, a^{-1}\big),
\end{align*}
where $(\#1)$ follows from the 2nd Peiffer rule (see Def. \ref{xmod}). On the other hand:
\begin{align*}
 a \t\psi_R\big (g,\d(R),e \big)&=a \t  (g,R^{-1} \,\, e\,\, g \trr R )=(\d(a) g,R^{-1}\,\, e \,\, g \trr R) \\&\stackrel{(\#2)}{=}(\d(a) g,a \,\, R^{-1}\,\, e \,\, g \trr R\,\, a^{-1}).
\end{align*}
Here step $(\#2)$ follows from the fact that $\ker (\partial)=A$ is central in $E$ (Lem. \ref{xmodprop}), since $R^{-1}\,\, e \,\, g \trr R\in A$.
\end{proof}

%

\subsubsection{The action groupoids $\TRANS(S^2(\Gc))$ and  $\TRANS(T^2_R(\Gc))$}\label{agrpoids}
\mdef Consider the actions $\trr'$ of $T^2_R(\Gc)$ and of $S^2(\Gc)$ on $E$, by automorphisms, such that:
\begin{align*}
(g,\d(R),e) \trr' a&=g \trr a, \textrm{ where } a \in E \textrm{ and } (g,\d(R),e)\in T^2_R(\Gc),\\
(g,K) \trr' a&=g \trr a, \textrm{ where } a \in E \textrm{ and } (g,K)\in S^2(\Gc). 
\end{align*}

\smallskip

\mdef\label{comp1} Note that $\psi_R( g,\d(R),e) \trr' a=(g,\d(R),e) \trr' a,$ where  $a \in E \textrm{ and } (g,\d(R),e)\in T^2_R(\Gc)$.

\smallskip

\mdef\label{rrrn} Let ${\td}$ denote the left-actions of $T^2_R(\Gc)$ and of $S^2(\Gc)$ on themselves by conjugation. 

\begin{Lemma}\label{refnow}
We have a left-action $\t$ of the group $T^2_R(\Gc)\ltimes_{\trr'}E$ (cf. \peq{sdc1}) on $T^2_R(\Gc)$. It has the form:
\begin{equation}\label{tT2}
\begin{split}
(g,\d(R),e,a) \t (h,\d(R),f)&= a \t \big((g,\d(R),e){\td} (h,\d(R),f)\big) \\&=(\d(a)ghg^{-1},R\, a\, R^{-1}\,\,  e\,\, g \trr f\,\, (ghg^{-1}) \trr e^{-1}\,\, a^{-1}\big).
\end{split}
\end{equation}
\end{Lemma}
\begin{proof}
Cf. Lem. \ref{def:t} and \peq{sdc2}. We solely need to prove that, given $a \in E$ and $(g,\d(R),e)\in T^2_R(\Gc)$, then: $$a \t \big( (g,\d(R),e){\td} (h,\d(R),f)\big)= (g,\d(R),e){\td} \big( (g^{-1} \trr a) \t (h,\d(R),e)\big).$$
 We have:
\begin{align*}
a \t \big((g,\d(R),e){\td} (h,\d(R),f)\big)&=a \t (ghg^{-1}, \d(R),e\,\, g \trr f\,\, (ghg^{-1}) \trr e^{-1}) \\&=(\d(a)ghg^{-1},\d(R),R\, a\, R^{-1}\,\,  e\,\, g \trr f\,\, (ghg^{-1}) \trr e^{-1}\,\, a^{-1}\big).
\end{align*}
On the other hand:
\begin{align*}
(g,\d(R),e){\td}\big( (g^{-1} \trr a) \t  (h,\d(R),f)\big)&=(g,\d(R),e) {\td}  (g^{-1}\d(a)gh, R\,\ g^{-1} \trr a\,\,  R^{-1} \,f \,\, g^{-1} \trr a^{-1})\\
&=\big(\d(a)g h g^{-1}, e\,\, g \trr R \,\,a \,\,g\trr R^{-1} \,\, g \trr f\,\, a^{-1}\,\,\,(\d(a) ghg^{-1})\trr e^{-1}\big)\\
&\stackrel{(\#1)}{=}\big(\d(a)g h g^{-1}, e\,\, g \trr R \,\,a \,\,g\trr R^{-1} \,\, g \trr f\,\, (ghg^{-1})\trr e^{-1} \,\,a^{-1}\big)\\
&\stackrel{(\#2)}{=}(\d(a)ghg^{-1},R\, a\, R^{-1}\,\,  e\,\, g \trr f\,\, (ghg^{-1}) \trr e^{-1}\,\, a^{-1}\big).
\end{align*}
Here $(\#1)$ follows  by the 2nd. Peiffer rule {\eqref{PR}.} Step $(\#2)$ also follows from the 2nd Peiffer rule, since:
\begin{align*}
 e\,\, g \trr R \,\,a \,\,g\trr R^{-1}
&\stackrel{(\# 1)}=\d(e) \trr \big ( g \trr R \,\,a \,\,g\trr R^{-1} \big)\,\, e \stackrel{(\# 2)}=\big(\d(R) g \d(R)^{-1} g^{-1} \big)  \trr \big ( g \trr R \,\,a \,\,g\trr R^{-1} \big)\,\, e 
\\&\stackrel{(\# 3)}=(\d(R)g\d(R^{-1})) \trr \big ( R \,\,g^{-1} \trr a \,\, R^{-1} \big)\,\, e  \\
&\stackrel{(\# 4)}=R\,\, g\trr \big (R^{-1} R \,\,g^{-1} \trr a \,\, R^{-1} R\big) \,\,R^{-1}\,\, e  =R\,\, a \, R^{-1}\, e.
\end{align*}
In $(\#2)$   we used   $\d(e)=\d(R) g \d(R)^{-1} g^{-1}$. The 2nd Peiffer rule was used {in $(\# 1)$ and $(\# 4)$.}
\end{proof}
\begin{Definition}\label{comp4} Consider the group $S^2(\Gc)\ltimes_{\trr'}E=G \ltimes_\trr (A \times E)$, where $\trr$ is the product action of $G$ on  $A \times E$. By \peq{comp1} and Lem. \ref{defpsiphi}, we can define a group isomorphism $\psi_R^*\colon T^2_R(\Gc)\ltimes_{\trr'}E\to S^2(\Gc)\ltimes_{\trr'}E$, as: $$\psi_R^*(g,\d(R),e,a)=(\psi_R(g,\d(R),e),a), \textrm{ where } (g,\d(R),e)\in T^2_R(\Gc) \textrm{ and } a \in E.$$
\end{Definition}
\mdef\label{prstar} The inverse $\phi_R^*\colon  S^2(\Gc)\ltimes_{\trr'}E \to T^2_R(\Gc)\ltimes_{\trr'}E $ of $\psi_R^*$ is such that $\phi_R^*(g,K,a)=(\phi_R(g,K),a)$.
\begin{Lemma}\label{compat}
We have a left-action $\t$ of $S^2(\Gc)\ltimes_{\trr'}E=G \ltimes_\trr (A \times E)$, on $S^2(\Gc)$. It has the form:
\begin{equation}\label{tS2}
(g,J,a) \t (h,K)= a \t \big((g,J){\td} (h,K)\big) =(\d(a)ghg^{-1}, J\,\,g\trr K\,\, (ghg)^{-1} \trr J^{-1}\big).
\end{equation}
Moreover, given ${(g,J,a)} \in S^2(\Gc)\ltimes_{\trr'}E$ and  $(h,K)\in S^2(E)$, it holds that: 
$$\phi_R\big(   {(g,J,a)}\t (h,K)\big)=\phi_R^*  {(g,J,a)} \t \phi_R\big( h,K\big) .$$
\end{Lemma}
\noindent Note that $\t$ is not, in general, an action by automorphisms.
\begin{proof}
That we have an action follows as in the proof of Lem. \ref{refnow}, or by using Lem. \ref{refnow} in the particular case $R=1_E$, given \peq{thisgen}. The other bits follows from concatenation of previous formulae. Here is a proof:
\begin{align*}
 \phi_R\big(   (g,J,{a})\t (h,K)\big)&=\phi_R\big ( \d({a})ghg^{-1}, J\,\,g\trr K\,\, (ghg)^{-1} \trr J^{-1}\big)\\ &=(\d({a}) ghg^{-1},{\d(R)}, R \,\, J\,\, g \trr K\,\, (ghg^{-1}) \trr J^{-1}\,\, (\d({a})ghg^{-1}) \trr R^{-1}\big)
\\ &=(\d({a}) ghg^{-1},{\d(R)}, R \,\, J\,\, g \trr K\,\, (ghg^{-1}) \trr J^{-1}\,\,{a}\,\, (ghg^{-1}) \trr R^{-1}\,\, {a}^{-1}\big).
\end{align*}
In the last step we used the 2nd {Peiffer relation} (see Def. \ref{xmod}). On the other hand:
\begin{align*}
\phi_R^*  (g,J,{a}) \t &\phi_R\big( h,K\big)=(g,\d(R),R\,\, J\,\, g \trr R^{-1},{a}) \t (h,\d(R), R\,\, K\,\, h\trr R^{-1})\\
&=(\d({a}) ghg^{-1},{\d(R)}, R\,\, {a}R^{-1} \,\, R\,\, J\,\, g \trr R^{-1}\,\, g \trr (R\,\, K\,\, h\trr R^{-1})\,\, (ghg^{-1}) \trr (R\,\, J\,\, g \trr R^{-1})^{-1}\,\, {a}^{-1})\\
&\stackrel{(\#1)}{=}(\d({a}) ghg^{-1},{\d(R)}, R\,\, {a}\,\, J \,\, g \trr K\,\, (ghg^{-1}) \trr J^{-1} \,\,(ghg^{-1})\trr R^{-1}\,\, {a}^{-1})\\
&\stackrel{(\#2)}{=}(\d({a}) ghg^{-1},{\d(R)}, R\,\, J \,\, g \trr K\,\, (ghg^{-1}) \trr J^{-1} \,\,{a}\,\,(ghg^{-1})\trr R^{-1}\,\, {a}^{-1}).
\end{align*}
Here $(\#1)$ follows by cancelling pairs of a group element product with  its inverse. And $(\#2)$ follows from the fact that $A$ is central in $E$ and $J,K\in A$; recall Lem. \ref{xmodprop}.
\end{proof}

\begin{Definition}\label{ts2} Cf. Def. \ref{de:ag}; $\TRANS(S^2(\Gc))$ is the action groupoid of the action $\t$ of $S^2(\Gc)\ltimes_\tp E$ on $S^2(\Gc)$. \end{Definition}

\begin{Definition}\label{tt2}
We let $\TRANS(T^2_R(\Gc))$ be the action groupoid of the action $\t$ of $T^2_R(\Gc)\ltimes_\tp E$ on $T^2_R(\Gc).$
\end{Definition}

 Arrows of  $\TRANS(S^2(\Gc))$ and of $\TRANS(T^2_R(\Gc))$ hence have the form below; see \eqref{tT2} and \eqref{tS2}:
\begin{align*}  
\Big ( (h,\d(R),f) \ra{(g,\d(R),e,a)}  (g,\d(R),e,a) \t {(h,\d(R),f)} \Big) \textrm{ and }
\Big ( (h,K) \ra{(g,J,{a})}  (g,J,{a}) \t (h,K)\Big).
\end{align*}
Here $g,h \in G$; $e,a,f \in E$ and $J,K \in A$. For an interpretation in terms of 2-fluxes of loop-particles see \S\ref{phys}.

\begin{Lemma}\label{iso} We have an isomorphism $\Phi_R\colon \TRANS(S^2(\Gc))\to \TRANS(T^2_R(\Gc))$, of groupoids.  On objects $\Phi_R$ is given by $\phi_R,$ see Lem. \ref{defpsiphi}. On morphisms  $\Phi_R$ takes the form {(see \peq{prstar}):}
$$\Phi_R\big ( (h,K) \ra{(g,J,{a})} (g,J,{a}) \t (h,K)\big)= \big(\phi_R( h,K) \ra{\phi_R^*(g,J,{a})} \phi_R^*(g,J,{a}) \t \phi_R(h,K)\big).$$
\end{Lemma}
\begin{proof}
 Follows by combining Lem. \ref{defpsiphi}, Lem \ref{compat} and  Def.  \ref{comp4}. 
\end{proof}
  \subsubsection{A W-\biker\ structure on  $\TRANS(S^2(\Gc))$}\label{bs2}
Let $\Gc=(\d\colon E \to G,\trr)$ be a crossed module. We put $\overline{g}$ to denote $g^{-1}$. {Recall \eqref{tS2}, \peq{sdc2} and Def. \ref{ts2}.}

\begin{Theorem}\label{bik-proof}
We have  a W-\biker\ structure $X^+_{gr^*}$ in the groupoid
$\TRANS(S^2(\Gc))$, such that:
\begin{equation}\label{X2}
\xymatrix@R=1pt{\\\\ X^+_{gr^*}\Big((z,J),(w,K)\Big)=}\hskip-1cm
\xymatrix{& (z,J)  \ar[dr]|\hole \ar[dr]|<<<<<{\overline{w}\,\,\,\,\bullet
    \quad\,\,}|\hole
  & (w,K)  \ar[dl]\ar[dl]|<<<<<{\quad\, \, \,\bullet \,\,\,\overline{w}\trr \overline{J}}\\
  & (\overline{w}\trr\overline{J})\t (w,K) & \overline{w} \t (z,J)  }
\raisebox{-.31in}{$
, 
\textrm{ where }\begin{cases} \overline{w} = (w^{-1},1_A,1_E) \\ 
                \overline{w}\trr \overline{J} =    (1_G,w^{-1} \trr J^{-1},1_E)    
\end{cases}
$}\xymatrix@R=1pc{\\ .}
\end{equation}
{(Note that $ \overline{w}\trr \overline{J}=\big(\beta\big(\Theta(w,K)^{-1} \t (z,J) \big),1_E\big)$; see \peq{defbeta}.)  Equation {\eqref{X2}} can  be written as (recall \peq{rrrn}):}
\begin{equation}\label{X2p}
\xymatrix@R=1pt{\\\\ X^+_{gr^*}\Big((z,J),(w,K)\Big)=}\hskip-1cm\xymatrix@R=48pt{& (z,J)  \ar[dr]|\hole \ar[dr]|<<<<<<<<<{(\Theta(w,K)^{-1},1_E)\,\,\,\,\bullet\,\,
    \quad\qquad\qquad\,\,\,\,}|\hole
  & (w,K)  \ar[dl]\ar[dl]|<<<<<{\quad \quad \quad \quad \quad \quad \quad \,\,\,\,\,\,\,\,\,\, \,\,\,\,\bullet\,\,\,\, (\Theta(w,K)^{-1} \t (z,J)^{-1},1_E)}
  |<<<<<<<<<<<<<<<{\quad \quad\quad \quad \quad \quad \quad \quad \,\,\,\,\,\,\,\,\,\,\,\,\,\,\,\,\,\,\bullet\,\,\,\,\,\,\,\,\,\, \big(\Theta\big(\Theta(w,K)^{-1} \t (z,J)\big),1_E\big)}
  \\
  & \beta\big(\Theta(w,K)^{-1} \t (z,J)\big)\t (w,K) & \Theta(w,K)^{-1} \t (z,J)  }\xymatrix@R=1pt{\\\\\\\ .}
\end{equation}

Explicitly, the underlying W-birack of $X^+_{gr^*}$ is such that (where $z,w \in G$ and $K,L \in A=\ker(\partial)$):
\begin{equation}\label{XR2}
\begin{split}
(z,J)\backslash (w,K)&=\overline{w} \t (z,J)=(w^{-1},1_A,1_E) \t (z,J)=\big({w^{-1}zw},w^{-1} \trr J\big),
 \\
(w,K) / (z,J)&=(\overline{w}\trr\overline{J})\t (w,K)=(1_G, w^{-1} \trr J^{-1},1_E)\t (w,K)= (w, w^{-1} \trr J^{-1} \,\, K \,\, J).
\end{split}
\end{equation}
And, as the diagram in \eqref{X2} indicates, the holonomy morphisms in the W-bikoid $X^+_{{gr^*}}$ in \eqref{X2} are:
\begin{align*}
L\big ((z,J),(w,K))&=\Big((z,J)\ra{(w^{-1},1_A,1_E)} (w^{-1},1_A,1_E) \t (z,J) \Big),\\ 
R\big ((z,J),(w,K))&=\Big((w,K)\ra{(1_G,w^{-1} \trr J^{-1},1_E) }(1_G,w^{-1} \trr J^{-1},1_E) \t (w,K)\Big).
\end{align*}
\end{Theorem}
\begin{Remark}The form \eqref{X2p} for the \biker \ $X^+_{gr^*}$ is useful (see \S\ref{mainbiker}) for transporting \peq{transport}  $X^+_{gr^*}$ in \eqref{X2} to  $\TRANS(T^2_R(\Gc))$. Equation \eqref{X2p} for \ $X^+_{gr^*}$  also makes the  interpretation of the associated representations \S\ref{representations} of the loop braid group as observables in topological higher gauge theory in $D^3$ more transparent; see \S\ref{phys}.
\end{Remark}
\begin{Remark}A more compact formulation of $X^+_{gr^*}$ in \eqref{X2}, in additive notation, can be found in \eqref{addX2}.
\end{Remark}

\noindent\begin{proof}
That  \eqref{X2} and \eqref{X2p} coincide follows from \peq{defTheta} and \peq{defbeta}. 

Let us prove that $X^+_{gr^*}$ in \eqref{X2} is a \biker. We freely use Defs. \ref{de:birack} and \ref{de:biker}.
First of all note that the maps:
\begin{align*}
{f^{(w,K)}}\colon (z,J) \in S^2(\Gc) &\longmapsto (z,J)\backslash (w,K)=(w^{-1},1_A,1_E) \t (z,J)=\big(w^{-1}zw,w^{-1} \trr J\big)\in S^2(\Gc), \\
{f_{(z,J)}}\colon (w,K) \in S^2(\Gc) & \longmapsto   (w,K)/(z,J) =(1_G,w^{-1} \trr J^{-1},1_E) \t (w,K)=  (w, w^{-1} \trr J^{-1} \,\, K \,\, J)\in S^2(\Gc) 
\end{align*}
each are bijections.
Inverses are given by:
\begin{align*}
\overline{{f^{(w,K)}}}\colon  (z',J') \in S^2(\Gc) &\longmapsto (w,1_A,1_E) \t (z',J')=(wz'w^{-1},w \trr J')\in S^2(\Gc), \\
\overline{{f_{(z,J)}}}\colon (w',K') \in S^2(\Gc) &\longmapsto  (1_G,{w'}^{-1} \trr J,1_E)  \t (w',K')=(w',{w'}^{-1} \trr J\,\, K'\,\,{J^{-1}}) \in S^2(\Gc).
\end{align*}

On the other hand, the map below is also invertible: $$S\colon \big((z,J),(w,K)\big) \in S^2(\Gc)\times S^2(\Gc) \mapsto \big((w,K)/(z,J),(z,J)\backslash (w,K)\big) \in S^2(\Gc)\times S^2(\Gc), $$
and its inverse is:
$$\big((w',K'),(z',J')\big)  \mapsto \big((w',1_A,1_E)\t (z',J'),(1_G,J',1_E)\t (w',K')\big)=\big((w'z'{w'}^{-1}, w'\trr J'),(w',J'\,\,K'\,\,w'\trr {J'}^{-1})\big).
$$

That  $\big((z,J),(w,K)\big) \in S^2(\Gc)\times S^2(\Gc) \mapsto \big((w,K)/(z,J),(z,J)\backslash (w,K)\big) \in S^2(\Gc)\times S^2(\Gc) $ in {Equation} \eqref{XR2} is a birack, and that {$X^+_{gr^*}$ in \eqref{X2}} is a \biker,  follows from comparing the diagrams \eqref{d1} and \eqref{d2} in 
$\TRANS\big(T^2_R(\Gc)\big)^3 \rtimes \Sigma_3$, where $(z,J),(w,K)$ and $(t,L)$ are general elements of $S^2(\Gc)$.
\begin{equation}\label{d1}
\hskip-2cm
\xymatrix@R=25pt@C=5pt{ 
  & (z,J) \ar[dr]|<<<<<<<{\overline{w}\,\,\,\bullet\quad }|\hole\ar[dr]|\hole
  & (w,K)\ar[dl]\ar[dl]|<<<<<<{\quad \,\,\bullet \,\,\,\overline{w}\trr \overline{J}} & (t,L)\ar[d]  \\
& (\bw \trr \bJ) \t (w,L)\ar[d] 
  &     \bw \t (z,J)
  \ar[dr]|<<<<{\,\,\overline{t}\,\,
   \,\,\, \bullet\quad}|\hole\ar[dr]|\hole
  & (t,L)\ar[dl]|<<<<<<<{\qquad \,\,\,\,
   \,\, \bullet\,\, \,\,(\overline{t}\,\overline{w}) \trr \overline{J}}\ar[dl] \\
& (\bw \trr \bJ) \t (w,K)
  \ar[dr]|<<<<<{\overline{t}\,\,\,\,\,\,\bullet\,\,\,\,}|\hole \ar[dr]|\hole
  & \big((\bt \bw) \trr \bJ\big)\t (t,L) \ar[dl]|<<<<<<{\qquad \quad\,\,\,\,\,\, \,\, \qquad \bullet \,\,\,\,\,\overline{t}\trr \bJ\,\,\overline{t}\trr \bK\,\,(\overline{t}\,\overline{w}) \trr J}\ar[dl] & (\bt\bw)\t (z,J)\ar[d] \\
  &(\overline{t}\trr \bJ\,\,\overline{t}\trr \bK)\t (t,L)  
& (\bt \,\, (\bw\trr \bJ)\trr (w,L)& (\bt\bw)\t (z,J)}\qquad\quad\xymatrix@R=1pt{
\bw=(w^{-1},1_A,1_E)\\
\bw \trr \bJ=(1_G,w^{-1} \trr J^{-1},1_E)\\
\overline{t}=(t^{-1},1_A,1_E)\\
(\bt\,\bw)\trr \bJ=(1_G,(t^{-1}w^{-1})\trr J^{-1},1_E)\\
\overline{t}\trr \bJ=(1_G, t^{-1}\trr J^{-1},1_E)\\
\overline{t}\trr \bK=(1_G, t^{-1} \trr K^{-1},1_E)
\\
(\overline{t}\,\overline{w}) \trr J=(1_G, (t^{-1}w^{-1})\trr J,1_E)
}\xymatrix@R=1pt{\\\\\\\\\\\\ .}
\end{equation}  
{(Recall the convention in Def. \ref{de:ag} for the composition in $\TRANS\big(T^2_R(\Gc)\big)$.} Note:
\begin{align*} 
\bw \t (z,J)&=(w^{-1},1_A,1_E)\t  (z,J)=(w^{-1}zw,w^{-1} \trr J),\\
(\bw\trr\bJ)\t (w,K)&=(1_G,w^{-1}\trr J^{-1},1_E) \t (w,K)
= (w,w^{-1}\trr J^{-1}\,\, K\,\, J)\\
\big(1_G,t^{-1}\trr (w^{-1}\trr J^{-1}\,\, K\,\, J)^{-1}\big)&=(1_G,t^{-1}\trr J^{-1}\,\, t^{-1} \trr K^{-1}\,\,(t^{-1}w^{-1}) \trr J\big)=\overline{t}\trr \bJ\,\,\overline{t}\trr \bK\,\,(\overline{t}\,\overline{w}) \trr J.
\end{align*}
And {(for the other side of the Reidemeister III move}):
\begin{equation}\label{d2}
\hskip-1.5cm
\xymatrix@R=25pt@C=5pt{
  & (z,J)\ar[d] 
  & (w,K)\ar[dr]|<<<<{\overline{t}\,\,\,\,\,\,\bullet\,\,\,\,\,}|\hole\ar[dr]|\hole
  & (t,L)\ar[dl]|<<<<<<{\,\,\,\,\,\,\,\,\,\,\,\bullet \,\,\,\,\,\overline{t}\trr \bK}\ar[dl]
  \\
  & (z,J)\ar[dr]|<<<<<<<{\overline{t}\,\,\,\,\,\,\,\,\bullet\,\,\,\,\,\,}|\hole\ar[dr]|\hole
  & (\bt \trr\bK)\t(t,L)\ar[dl]|<<<<<<{\,\,\,\,\,
    \,\,\,\,\,\bullet \,\,\,\,\,\,\,\overline{t}\trr \bJ}|\hole\ar[dl] & \bt \t (w,K)\ar[d] \\
& \big ((\bt \trr \bJ) (\bt \trr \bK)\big)\t (t,L)\ar[d]  
  & \bt \t (z,J)\ar[dr]|<<<<<{\overline{t}\,\overline{w}\, t\,\,\bullet\,\, \quad
    \,\,}|\hole\ar[dr]|\hole &  \bt \t(w,K)\ar[dl]\ar[dl]|<<<<<{\,\,\quad
    \quad \,\,\bullet\, \, \,\, (\bt \bw)\trr \bJ} \\
&  \big ((\bt \trr \bJ) (\bt \trr \bK)\big)\t (t,L)
  & \big((\bt \bw )\trr \bJ \,\,\bt\big)\t (w,K)
    & (\bt \bw) \t (z,a)
    } \quad\,\,\,\xymatrix@R=1pt{
\bw=(w^{-1},1_A,1_E)\\
\overline{t}=(t^{-1},1_A,1_E)\\
{t}=(t,1_A,1_E)\\
(\bt\,\bw)\trr \bJ=(1_G,(t^{-1}w^{-1})\trr J^{-1},1_E)\\
\overline{t}\trr \bJ=(1_G, t^{-1}\trr J^{-1},1_E)\\
\overline{t}\trr \bK=(1_G, t^{-1} \trr K^{-1},1_E)
\\
(\overline{t}\,\overline{w}) \trr \bJ=(1_G, (t^{-1}w^{-1})\trr J^{-1},1_E)
}\xymatrix@R=1pc{\\ \\ \\.}
\end{equation}
Note:
\begin{align*}
&\bt\t (w,K)=(t^{-1}wt,t^{-1} \trr K),
&&\bt\t (z,J)=(t^{-1}zt,t^{-1} \trr J).
\end{align*}
{We can easily see} that the group elements in $S^2(\Gc)\ltimes_\tp E$ associated to each strand in the {diagrams in} \eqref{d1} and \eqref{d2} coincide. Namely in the first strand we have $\bt\,\bw=\bt\,\bw\, t\, \bt$; in the second strand we have: 
$$\bt\, (\bw \trr \bJ) = (t^{-1},1_A,1_E) \,(1_G,w^{-1} \trr J^{-1},1_E)=(1_G,(t^{-1}w^{-1}) \trr J^{-1},1_E)\,  (t^{-1},1_A,1_E)=\big((\bt\bw)\trr \bJ)\,\, \bt. $$
And in the third strand we have:
$(\overline{t}\trr \bJ\,\,\overline{t}\trr \bK\,\,(\overline{t}\,\overline{w}) \trr J)\,\, (\overline{t}\,\overline{w}) \trr \bJ)=\overline{t}\trr \bJ\,\,\overline{t}\trr \bK. $

Hence given any $(z,J),(w,K),(t,L)\in S^2(\Gc)$, then diagrams \eqref{d1} and \eqref{d2} in 
$\TRANS\big(T^2_R(\Gc)\big)^3 \rtimes \Sigma_3$ coincide. Hence, by construction, so do the bottom lines of  \eqref{d1} and \eqref{d2}, thus it  follows that $(S^2(\Gc),/,\backslash)$ in \eqref{XR2} is a birack. And then the equality of  \eqref{d1} and \eqref{d2} mean exactly that $X^+_{{gr^{*}}}$ in \eqref{X2} is a \biker.

That the \biker\  is welded (Def. \ref{de:wbiker}), follows in exactly the same way, by comparing the two diagrams below \eqref{d3} and \eqref{d4} in $\TRANS(S^2(\Gc))^3 \rtimes \Sigma_3$, where $(z,J),(w,K),(t,L)\in S^2(\Gc)$:
\begin{equation}\label{d3}
\hskip-2cm
\xymatrix@R=19pt@C=5pt{ & (z,J) \ar[dr]  & (w,K)\ar[dl]\ar[dl] & (t,L)\ar[d]  \\
  & (w,K) \ar[d] &      \big(z,J\big)
  \ar[dr]|<<<{\overline{t}\,\,\,\,\bullet\,\,\,\,}|\hole\ar[dr]|\hole &
  (t,L)\ar[dl]|<<<{\quad\,\,\, 
    \bullet\,\,\,\, \overline{t}\,\trr \bJ}\ar[dl] \\
  & (w,K) \ar[dr]|<<<<{\overline{t}\,\,\,\,\,\,\,\bullet\,\,\,\,\,}|\hole
  \ar[dr]|\hole & (\bt \trr \bJ) \t (t,L)\ar[dl]|<<<{\quad
   \,\,\,\,\, \bullet \,\,\,\,\,\,\,\overline{t}\trr \bK}\ar[dl] &\bt \t (z,J)\ar[d] \\
  & (\bt \trr \bK\,\,\bt \trr \bJ) \t (t,L) & {\bt \t (w,K)} & \bt\t (z,J) }
\quad\xymatrix@R=1pt{
\overline{t}=(t^{-1},1_A,1_E)\\
\overline{t}\trr \bJ=(1_G, t^{-1}\trr J^{-1},1_E)\\
\overline{t}\trr \bK=(1_G, t^{-1} \trr K^{-1},1_E)
}\xymatrix@R=1pc{\\ \\.}
\end{equation}
and:
\begin{equation}\label{d4}
\xymatrix@R=19pt@C=14pt{
  &(z,J)\ar[d]
  & (w,K)\ar[dr]|<<<<{\overline{t}\,\,\,\,\,\,\bullet\,\,\,\,}|\hole\ar[dr]|\hole
  & (t,L)\ar[dl]|<<<<<<{\quad \,\,\,\,\bullet \,\,\,\,\,\,\,\overline{t}\trr \bK}\ar[dl]  \\
  &(z,J)\ar[dr]|<<<<<<{\overline{t}\,\,\,\,\,\,\,\,\bullet\,\,\,\,\,\,}|\hole\ar[dr]|\hole
  & {(\overline{t}\trr\bK)}\t (t,L)\ar[dl]|<<<<<{ \,\,\,\,\,\,\,\,\,\ \bullet\,\,\,\,\,\, \,\,\overline{t}\trr \bJ}|\hole\ar[dl] & \overline{t}\t (w,K)\ar[d] \\
  & {(\bt \trr \bJ\,\, \bt \trr \bK) \t (t,L)}\ar[d]
  & \bt \t (z,J)\ar[dr]\ar[dr]|\hole & \bt \t (w,K)\ar[dl]\\
  &{(\bt \trr \bJ\,\, \bt \trr \bK) \t (t,L)}  
  &  \bt \t (w,K)
  &  \bt \t (z,J)
,  }\quad\xymatrix@R=1pt{
\overline{t}=(t^{-1},1_A,1_E)\\
\overline{t}\trr \bJ=(1_G, t^{-1}\trr J^{-1},1_E)\\
\overline{t}\trr \bK=(1_G, t^{-1} \trr K^{-1},1_E)}\xymatrix@R=1pc{\\ \\.}
\end{equation}
Morphisms \eqref{d3} and \eqref{d4} in $\TRANS(S^2(\Gc))^3 \rtimes \Sigma_3$ coincide, for each $(z,J),(w,K),(t,L)\in S^2(\Gc)$. This follows since $\overline{t}\trr \bJ\,\, \overline{t}\trr \bK = \overline{t}\trr \bK \,\, \overline{t}\trr \bJ$, as the group operation in $A=\ker(\partial)$ is commutative; see Lem. \ref{xmodprop}. 
\end{proof}

\begin{Remark}\label{refA}{Let $\Gc=(\d\colon E \to G,\trr)$ be a crossed module. We have an abelian $gr$-group $(G,A=\ker(\partial))$ \peq{grt}. Write $A$ in multiplicative notation. The W-\biker\ $X^+_{gr^*}$ {is} obtained  from the W-bikoid $X^+_{gr}$ in \eqref{kau-mar}, {and {\it vice versa}}. {We have} an inclusion morphism  of groupoids
${\rm Inc}\colon \TRANS(G,A) \to \TRANS(S^2(\Gc))$, with:}
 $${\rm Inc} \Big((h,K) \ra {(g,J)} (g,J)  (h,K) (g,J)^{-1}\Big)= \Big((h,K) \ra {(g,J,1_E)} (g,J,1_E)\t  (h,K)\Big).$$
This {inclusion} is bijective on objects. And then $X_{gr^*}^+$ is obtained by transporting \peq{transport} $X_{gr}^+$ along ${\rm Inc}$.

{In particular, \S \ref{bandh} gives a topological interpretation for the existence of the W-bikoid $X^+_{gr^*}$.}
\end{Remark}
\begin{Remark} {(Cf. \peq{grt} \peq{grtox} and Rem. \ref{refA}.)} The representations of {$\WBG_n$} derived from  $X^+_{gr^*}$ do not appear to encode more information that those derived from $X^+_{gr}$. The extra degree of structure appearing in the underpinning groupoid $\TRANS(S^2(\Gc))$ of  $X^+_{gr^*}$ is {\peq{lastc}} related to gauge transformations betwen 2-connections. Hence  $X^+_{gr^*}$ is likely more fundamental than $X^+_{gr}$. From a strictly mathematical sense, this extra degree of freedom  that $\TRANS(S^2(\Gc))$ has in comparison with  $\TRANS(G,A)$ is essential when addressing invariants of welded knots derived from crossed modules, by using $X^+_{gr^*}$. This will be addressed in \cite{prox}.
\end{Remark}


%
%
%

\begin{Remark}\label{whyto-1}
{Let $\Gc=(\d\colon E \to G,\trr)$ be a crossed module. The underlying W-birack \eqref{XR2} of $X^+_{gr^*}$  is  studied in \cite{martins_kauffman}. {In the cases in Examples \ref{gr1} and \ref{gr2}, this W-birack yields non-trivial invariants of welded knots that see beyond their knot groups; see also \cite{martins_tams}.} This indicates that the representations of {$\WBG_n$} derived from $\Gc=(\d\colon E \to G,\trr)$ carry more topological information than those derived  (see \peq{deQD}, \peq{deQD-1}) from $G$ alone.} 
\end{Remark}

\begin{Example}\label{whyto-2} {Continuing Rem. \ref{whyto-1}.} Consider the crossed module $\Gc(\Z_2,\Z_3)$ derived \peq{grtox} from the abelian $gr$ group in {Example \ref{gr2}. Let $V$ be the right-regular representation of  {$\C\big(\TRANS(S^2(\Gc(\Z_2,\Z_3)))\big)$}; {Example}  \ref{rre}. Let  $V_{{\rm OBJ}}$ be the object-regular representation of 
{$\C\big(\TRANS(S^2(\Gc(\Z_2,\Z_3)))\big)$}; see {Example}  \ref{obre}. Let also $U$ and $U_{{\rm OBJ}}$ be the right-regular and object-regular representations of $\C(\AUT(\Z_2))$; see {Example}  \ref{autg}. Consider the representations of $\WBG_2$ derived from the W-bikoid $X^+_{gr*}$ in \eqref{X2}, on $V \tn V$ and on $V_{{\rm OBJ}}\tn V_{{\rm OBJ}}$; see Thm. \ref{rep12} or Thm. \ref{uf}. 
 Consider the representations  of $\WBG_2$ on  $U \tn U$ and on $U_{{\rm OBJ}}\tn U_{{\rm OBJ}}$, derived from the W-bikoid  $X^+_{\Z_2}$ in \eqref{fgwb}. The degree of $(-) \trl^*S^+_1[2]$ in \eqref{act2b}, induced by the braid generator $S^+_1[2]$ in \eqref{eq:SV} is:}
\begin{align}
 &2 \textrm{ for the action of } \WBG_2 \textrm{ on }  U_{{\rm OBJ}}\tn U_{{\rm OBJ}}, && 4 \textrm{ for the action of } \WBG_2 \textrm{ on }   U\tn U, \label{fl}\\
 &12 \textrm{ for the action of } \WBG_2 \textrm{ on }  V_{{\rm OBJ}}\tn V_{{\rm OBJ}}, && 12 \textrm{ for the action of } \WBG_2 \textrm{ on }   V\tn V\label{sl}.
\end{align}
{Hence Equation \eqref{fl} tell us that the representations of $\WBG_n$ derived from a {W-bikoid} may carry more information than those derived from its underlying {W-birack} \peq{uncat}. Comparing \eqref{fl} and \eqref{sl} gives an example of a case when the  representations of $\WBG_n$ derived from a crossed module $(\d\colon G \to E,\trr)$ -- in this case $\Gc(\Z_2,\Z_3)$ -- are finner than those derived from the {W-bikoid associated to $G$ alone}.}
\end{Example}

\subsubsection{A W-\biker\ structure in $\TRANS(T^2_R(\Gc))$}\label{mainbiker}

Let  $\Gc=(\d\colon E \to G,\trr)$ be a crossed module. Let $R\in E$. (The physical meaning of $R$ is {explained} in  \peq{defR}.)
 We now describe the {W-bikoid} $X^+_R$ mentioned in the end of \S\ref{phys}. (Hence $X^+_R$ is related to loop-particles in topological higher gauge theory.)  {This $X^+_R$ is derived from and is isomorphic to $X^+_{gr^*}$ in Equation \eqref{X2}.}

Cf. Defs. \ref{ts2}, \ref{tt2}, and Lem. \ref{iso}. There is a groupoid isomorphism  $\Phi_R\colon \TRANS(S^2(\Gc))\to \TRANS(T^2_R(\Gc))$. 

\begin{Theorem}
 {We have  a W-\biker\  
defined in the groupoid $X^+_R$
in $\TRANS(T^2_R(\Gc))$. This {$X^+_R$} is obtained by transporting \peq{transport} $(\TRANS(S^2(\Gc)),X^2_{gr^*})$ {in \eqref{X2}} to $\TRANS(T^2_R(\Gc))$ along $\Phi_R$. The explicit form of  $X^+_R$ is (recalling the notation of Defs.   \ref{defThetaR} and \ref{defbetaR},  and that we sometimes use $\overline{(-)}$ to denote inverses in a group):} 
\begin{equation}\label{X3}
\hskip-0.4cm\xymatrix@R=1pt{\\\\\\\\ X^+_R\Big((z,\d(R),e),(w,\d(R),f)\Big)=}\hskip-1cm
\hskip-2.1cm\xymatrix@R=70pt@C=10pt{& (z,\d(R),e)  \ar[dr]|\hole \ar[dr]|<<<<<<<<<{(\overline{\Theta_R(w,\d(R),f)},1_E)
\,\,\,\,\bullet
    \quad\qquad\qquad\qquad}|\hole
  & (w,\d(R),f)  \ar[dl]\ar[dl]|<<<<<<<<<<<<{\quad \qquad\qquad\qquad \qquad \qquad \,\bullet\,\,\,\big(\overline{\Theta_R(w,\d(R),f)}\t \overline{(z,\d(R),e)},1_E)}
\ar[dl]|<<<<<<<<<<<<<<<<<<<<<<<<{\qquad \quad \qquad\qquad \qquad \qquad \qquad \bullet\,\,\,\,
(\Theta_R\big(\overline{\Theta_R(w,\d(R),f)}\t (z,\d(R),e)\big),1_E)
}
\\ 
  &   \beta_R \Big(\overline{\Theta_R(w,\d(R),f)}\t (z,\d(R),e)\Big) \t (w,\d(R),f)
 & \overline{\Theta_R(w,\d(R),f)}\t (z,\d(R),e) }\xymatrix@R=1pc{\\ \\.}
\end{equation}
{(Recall \peq{rrrn}.)}
 {Or, in another form:}
\begin{equation}\label{X3p}
\hskip-0.45cm\xymatrix@R=1pt{\\\\\\\\\\ X^+_R\Big((z,\d(R),e),(w,\d(R),f)\Big)=}\hskip-1cm
\hskip-2.35cm\xymatrix@R=60pt@C=10pt{& (z,\d(R),e)  \ar[dr]|\hole \ar[dr]|<<<<<<<<<<<{\big(\overline{\Theta_R(w,\d(R),f)},1_E)
\,\,\bullet \qquad
    \quad\qquad\qquad}
|\hole
  & (w,\d(R),f)  \ar[dl]\ar[dl]|<<<<<<<<<<<<<<<<{\quad \qquad\qquad\qquad \qquad \qquad \quad\,\,\,\,\,\,\,\,\,\, \bullet\,\, \,\,\,\,\,\big(\beta_R \big(\overline{\Theta_R(w,\d(R),f)}\t (z,\d(R),e)\big),1_E\big)}
\ar[dl]
\\ 
  &   \beta_R \Big(\overline{\Theta_R(w,\d(R),f)}\t (z,\d(R),e)\Big) \t (w,\d(R),f)
 & \overline{\Theta_R(w,\d(R),f)}\t (z,\d(R),e) }\xymatrix@R=1pc{\\ \\.}
\end{equation}
{Explicitly, on the second strand of  \eqref{X3}, we have the following element of $T^2_R(\Gc)$ (see \eqref{dEFT} and \eqref{betaversions}):}
\begin{equation}\label{otherform}
\begin{split}
\Theta_R\big(\overline{\Theta_R(w,\d(R),f)}\t (z,\d(R),e)\big)&\,\,
\overline{\Theta_R(w,\d(R),f)}\t \overline{(z,\d(R),e)}\\&=\beta_R\Big(\overline{\Theta_R(w,\d(R),f)}\t (z,\d(R),e)\Big)\\
&=\beta_R\big (w^{-1}zw,{\d(R),} R\,\, w^{-1} \trr R^{-1} \,\, w^{-1} \trr e\,\, (w^{-1}z)\trr R\,\, (w^{-1}zw)\trr R^{-1} \big)\\
 &=\big(1,{\d(R),}R\,\,( w^{-1}z) \trr R^{-1}\,\, w^{-1} \trr e^{-1} \,\, w^{-1} \trr R\,\, R^{-1}\big).
\end{split}
\end{equation}
Hence the element of $T^2_R(\Gc)$ associated to the target of the  second strand of the {diagram in \eqref{X3}} is:
\begin{equation}\label{K2}
\begin{split}\big(\beta_R (\overline{\Theta_R(w,\d(R),f)}\t &(z,\d(R),e)\big),1_E\big)  \t (w,\d(R),f)\
\\&=(w,{\d(R),}R\,\, (w^{-1}z) \trr R^{-1}\,\, w^{-1} \trr e^{-1} \,\, w^{-1} \trr R\,\, R^{-1}\,\, f \,\, w \trr R\,\, R^{-1} \,\, e \,\, z\trr R \,\, w \trr R^{-1}).
\end{split}
\end{equation}
And, on the first strand of {the diagram in \eqref{X3},} we have the following element of $T^2_R(\Gc)$:
\begin{align}\label{K3}\overline{\Theta_R(w,\d(R),f)}=(w^{-1},\d(R),R\,\,w^{-1}\trr R^{-1}).
\end{align}
Hence the element of $T^2_R(\Gc)$ associated to the target of the first strand  in \eqref{X3} is:
\begin{align}\label{K4} 
\overline{\Theta_R(w,\d(R),f)}\t (z,\d(R),e)=\big(w^{-1}zw,\d(R),R\,\, w^{-1} \trr R^{-1}\,\, w^{-1} \trr e\,\, (w^{-1}z) \trr R\,\,(w^{-1}zw) \trr R^{-1}).
\end{align}  
\end{Theorem}
\begin{Remark}
{A more compact formulation of $X^+_R$ in  Equations \eqref{X3p} and \eqref{otherform} is in Equation \eqref{compact}.}
\end{Remark}
\begin{proof}
{By using  \eqref{betaversions}, it follows that the formulations for $X^+_R$ in \eqref{X3} and \eqref{X3p} coincide, and it also follows \eqref{otherform}.  Equation \eqref{K2} follows from \eqref{tT2}. Analogously,  \eqref{K3} and \eqref{K4} follow from Lem. \ref{defThetaR} and  \eqref{tT2}.}

{We  prove that $X^+_R$ in \eqref{X3p} is obtained by  transporting the bikoid $X^+_{gr^*}$ in \eqref{X2} along $\Phi_R\colon \TRANS(S^2(\Gc))\to \TRANS(T^2_R(\Gc))$ in Lem. \ref{iso}; recall \peq{transport}. We use the formulation in Equation \eqref{X2p} for $X^2_{gr^*}$.
Recall Defs. \ref{defpsiphi} and \ref{comp4}.
That $X^+_R$ in  \eqref{X3} is obtained by transporting $X^+_{gr^*}$ in  \eqref{X2p} along $\Phi_R$ follows swiftly from the construction of  $\Phi_R\colon \TRANS(S^2(\Gc))\to \TRANS(T^2_R(\Gc))$ in Lem. \ref{iso}, together with Lem. \ref{defThetaR},  \ref{def:t} and \peq{betasame}.}
\end{proof}
%

%
%
%
%
%
%

\bibliographystyle{ieeetr}

\bibliography{RepsLBG.bib}
\end{document}